\documentclass[final,onefignum,onetabnum]{siamart171218}

\usepackage{lipsum}
\usepackage{amsfonts}
\usepackage{graphicx}
\usepackage{epstopdf}
\usepackage{algorithmic} 
\ifpdf
  \DeclareGraphicsExtensions{.eps,.pdf,.png,.jpg}
\else
  \DeclareGraphicsExtensions{.eps}
\fi

\usepackage{nicefrac}
\usepackage{bm}
\usepackage{tikz}
\usetikzlibrary{decorations.pathmorphing}
\usetikzlibrary{matrix}
\newcommand{\veq}{\mathrel{\rotatebox{90}{$=$}}}
\usepackage{mathtools}
\usepackage{accents}
\newcommand{\ubar}[1]{\underaccent{\bar}{#1}}
\newcommand{\ucbar}[1]{\text{\b{$#1$}}}
\usepackage{multirow, makecell}

\usepackage{enumitem}
\setlist[enumerate]{leftmargin=.5in}
\setlist[itemize]{leftmargin=.5in}

\newsiamremark{remark}{Remark}
\newsiamremark{hypothesis}{Hypothesis}
\crefname{hypothesis}{Hypothesis}{Hypotheses}
\newsiamthm{claim}{Claim}

\headers{Holistic Analysis of Nonlinear Dynamic Compartmental Systems}{Huseyin Coskun}

\title{Nonlinear Decomposition Principle \\ and Fundamental Matrix Solutions \\ for Dynamic Compartmental Systems}

\author{Huseyin Coskun\thanks{Department of Mathematics, University of Georgia, Athens, GA 30602 (\email{hcoskun@uga.edu}).} }

\usepackage{amsopn}
\DeclareMathOperator{\diag}{diag}

\ifpdf
\hypersetup{
  pdftitle={Nonlinear Decomposition Principle},
  pdfauthor={Huseyin Coskun}
}
\fi

\begin{document}

\maketitle

\begin{abstract}
A decomposition principle for nonlinear dynamic compartmental systems is introduced in the present paper. This theory is based on the mutually exclusive and exhaustive, analytical and dynamic, novel system and subsystem partitioning methodologies. A deterministic mathematical method is developed for the dynamic analysis of nonlinear compartmental systems based on the proposed theory. The dynamic method enables tracking the evolution of all initial stocks, external inputs, and arbitrary intercompartmental flows, as well as the associated storages derived from these stocks, inputs, and flows individually and separately within the system. The transient and the dynamic direct, indirect, acyclic, cycling, and transfer (\texttt{diact}) flows and associated storages transmitted along a particular flow path or from one compartment\textemdash directly or indirectly\textemdash to any other are then analytically characterized, systematically classified, and mathematically formulated. Thus, the dynamic influence of one compartment, in terms of flow and storage transfer, directly or indirectly on any other compartment is ascertained. Consequently, new mathematical system analysis tools are formulated as quantitative system indicators. The proposed mathematical method is then applied to various models from literature to demonstrate its efficiency and wide applicability.
\end{abstract} 

\begin{keywords}
nonlinear decomposition principle, fundamental matrix solutions, dynamic system and subsystem decomposition, complex systems theory, nonlinear dynamic compartmental systems, \texttt{diact} flows and storages, dynamic input-output analysis, dynamic input-output economics, socioeconomic systems, epidemiology, infectious diseases, toxicology, pharmacokinetics, neural networks, chemical and biological systems, control theory, information theory, information diffusion, social networks, computer networks, malware propagation, traffic flow
\end{keywords}


\section{Introduction}
\label{sec:intro}

Compartmental systems are mathematical abstractions of networks composed of discrete, homogeneous, interconnected components that approximate the behavior of continuous physical systems. The system compartments are interrelated through the flow of a conserved quantity, such as energy, matter, information, disease, or currency between them and their environment based on conservation principles. Therefore, for an accurate quantification of the compartmental system function and behavior, analytical and explicit formulation of system flows and the associated storages generated by these flows are of paramount importance.

Today's major natural problems involve change, and this makes the need for methods of dynamic nonlinear system analysis not only appropriate, but also urgent. Such methods for compartmental system analysis, however, have remained a long-standing open problem. Sound rationales are offered in the literature for compartmental system analysis, but they are for special cases, such as linear models and static systems \cite{Leontief1936,Leontief1966,Hannon1973,Patten1978,Szyrmer1987}. Various mathematical aspects of compartmental systems are also studied in the literature \cite{Jacquez1993,Anderson2013}.

This is the first manuscript that potentially addresses the mismatch between the needs for dynamic nonlinear compartmental system analysis and current static and computational simulation methods. The present manuscript is structured in three levels: theory, methods, and applications. The underlying mathematical theory will be called the {\em nonlinear decomposition principle}. The theory is based on the novel {\em dynamic system} and {\em subsystem decomposition methodologies}. A deterministic mathematical method is then developed for the dynamic analysis of nonlinear compartmental systems. We consider that the proposed mathematical theory and methodology brings a novel complex systems theory to the service of urgent and challenging natural problems of the day.

The {\em system} decomposition methodology explicitly generates mutually exclusive and exhaustive subsystems, each driven by a single initial stock or external input, that are running within the original system and have the same structure and dynamics as the system itself. Therefore, the composite compartmental flows and storages are dynamically decomposed into subcompartmental flow and storage segments based on their constituent sources from the initial stocks and external inputs. Consequently, the system decomposition methodology yields the {\em subthroughflow} and {\em substorage matrix functions} that respectively represent the flows and storages derived from the initial stocks and external inputs. Equipped with these matrix measures, the system partitioning ascertain the dynamic distribution of initial stocks and external inputs, as well as the organization of the associated storages derived from these stocks and inputs individually and separately within the system. In other words, the system decomposition enables dynamically tracking the evolution of the initial stocks (initial conditions) and external inputs (source terms), as well as associated storages (state variables) individually and separately within the system.

The subsystems are then dynamically decomposed along a set of mutually exclusive and exhaustive directed subflow paths. The {\em subsystem} decomposition methodology yields the {\em transient subflows} transmitted from one subcompartment directly or indirectly to any other along a particular flow path and the associated {\em substorages} generated by these subflows within subsystems. Therefore, arbitrary composite intercompartmental flows and storages can dynamically be decomposed into transient subflow and substorage segments along a given set of flow paths. In other words, the transient subflows and substorages represent the dynamic distribution of arbitrary intercompartmental flows and the organization of associated storages along the given set of pathways. Consequently, the subsystem decomposition enables dynamically tracking the fate of arbitrary intercompartmental flows and storages within the system. The spread of an arbitrary flow or storage segment from one compartment to the entire system can then be determined and monitored. Moreover, a history of compartments visited by particular system flows and storages can also be compiled. In summary, the proposed mathematical method, as a whole, decomposes the system flows and storages to the utmost level. 

The dynamic {\em direct}, {\em indirect}, {\em acyclic}, {\em cycling}, and {\em transfer} ($\texttt{diact}$) flows and storages transmitted from one compartment, directly or indirectly, to any other in the system are also analytically characterized, systematically classified, and mathematically formulated for the quantification of intercompartmental flow and storage dynamics. The direct influence of one compartment on another can be determined through the state-of-the-art techniques. The proposed method, however, makes the dynamic analysis of both direct and indirect influence of one compartment, in terms of flow and storage transfer, on any other possible in a nonlinear system. 

The proposed methodology, therefore, develops new mathematical system analysis tools as quantitative system indicators. All these dynamic measures\textemdash that is, the subthroughflows, substorages, as well as the transient and \texttt{diact} flows and storages\textemdash are introduced for the first time in literature. They are used in the analysis of illustrative models in Section~\ref{sec:results} and Appendix~\ref{apxsec:ex} to demonstrate the wide applicability and efficiency of the proposed methodology. The results indicate that the proposed method provides significant advancements in the theory, methodology, and practicality of the current compartmental system analysis.

The applicability of the proposed method extends to various realms regardless of their naturogenic and anthropogenic nature, such as ecology, economics, pharmacokinetics, chemical reaction kinetics, epidemiology, chemical and biomedical systems, neural networks, social networks, and information science. An input-output analysis in economics was developed several decades ago, but only for static systems \cite{Leontief1936,Leontief1966}. The proposed methodology in the context of economics, in particular, can be considered as the mathematical foundation of the {\em dynamic input-output economics}. Essentially, the methodology is applicable to any real-world phenomenon where compartmental models of conserved quantities can be constructed. Considering hypothetical complex networks with multiple interacting compartments, the compartments can, for example, model species in an ecosystem, sectors in an economy, organs in an organism, molecules in a chemical reaction, neurons in a neural network, or communities in a social network. The conserved quantities that need to be investigated within these systems, then, would be a nutrient, amount or value of goods and services, a certain drug, a specific type of atom, particular ions, or a piece of information, respectively.

An illustrative SIRS model from epidemiology is analyzed in detail in Section~\ref{sec:results}, and more case studies are presented from ecosystem ecology in the Appendices. The SIRS model consists of three compartments that represent the populations of three groups: the susceptible or uninfected, $S$, infectious, $I$, and recovered or immune, $R$. The model determines the number of individuals infected with a contagious illness over time \cite{Kermack1927}. It is shown that, the proposed dynamic {\em system} decomposition methodology decomposes each of the composite SIR populations into subpopulations based on their constituent sources from the initial and newborn populations. Consequently, the method enables tracking the evolution of the health states of the newborn or initial SIR populations individually and separately within the total population. The proposed dynamic {\em subsystem} decomposition methodology then enables tracking the evolution of the health states of an arbitrary population in any SIR group along a given infection path. Therefore, the effect of an arbitrary population on any other group, through not only direct but also indirect interactions, can be ascertained. Consequently, the spread of the disease from an arbitrary population to the entire system can be determined and monitored. It is worth noting for a comparison that the solution to the SIRS model through the state-of-the-art techniques can only provide the composite SIR populations without distinguishing subpopulations based on their constituent sources.

The paper is organized as follows: the mathematical method is introduced in Sections~\ref{sec:csystems},~\ref{sec:systemd}, and~\ref{sec:dsd}, the fundamental matrix solutions and decomposition principle are introduced in Sections~\ref{sec:nfm} and~\ref{sec:ndp}, system analysis is discussed in Section~\ref{sec:sa}, and results, examples, discussions, and conclusions follow at the end of the manuscript.

\section{Methods}
\label{sec:method}

The {\em nonlinear decomposition principle} for dynamic compartmental systems is introduced in this section based on the mutually exclusive and exhaustive, analytical and dynamic, novel {\em system} and {\em subsystem decomposition methodologies}. A deterministic mathematical method is then developed for the dynamic analysis of nonlinear compartmental systems. The proposed theory and method construct a base for the formulation of new mathematical system analysis tools, such as the {\em subthroughflows, substorages}, and the \texttt{diact} {\em flows} and {\em storages}, as quantitative system indicators. These novel concepts and quantities are defined and formulated in this section.

\subsection{Compartmental systems}
\label{sec:csystems}

We assume that the components of a physical system are modeled as compartments that are interconnected through flow of energy, matter, information, disease, or currency. In such a compartmental model, the state variable $x_i(t)$ represents the amount of storage in compartment $i$, and $f_{ij}(t,\bm{x})$ represents the non-negative flow rate from compartment $j$ to $i$ at time $t$ (see Fig.~\ref{fig:sc}).

Let the governing equations for this nonlinear dynamic compartmental system, formulated based on conservation principles, be given as
\begin{equation}
\label{eq:modelV}
\dot{\bm x}(t) = {\bm \tau}(t,\bm{x})
\end{equation}
with the initial conditions of ${\bm x}(t_0) = {\bm x}_0$. The {\em state vector} $\bm{x}(t)=[x_1(t),\ldots,x_n(t)]^T$ is a differentiable function of time, $t$. The function ${\bm \tau}(t,\bm{x})=[ \tau_1(t,\bm{x}),\ldots,\tau_n(t,\bm{x}) ]^T$ will be called the {\em net throughflow rate} vector and expressed as
\begin{equation}
\label{eq:modelTau}
{\bm \tau}(t,\bm{x}) \coloneqq {\check{\bm \tau}}(t,\bm{x}) - {\hat{\bm \tau}}(t,\bm{x})
\end{equation}
where the respective {\em inward} and {\em outward throughflow rate} vector functions are
\[\check{\bm \tau}(t,\bm{x}) = [\check{\tau}_1(t,\bm{x}),\ldots,\check{\tau}_n(t,\bm{x})]^T
\, \, \, \mbox{and}  \, \, \,
{\hat{\bm \tau}}(t,\bm{x}) = [ \hat{\tau}_1(t,\bm{x}),\ldots,\hat{\tau}_n(t,\bm{x}) ]^T . \]
The components of the throughflow vectors can further be expanded as
\begin{equation}
\label{eq:in_out_flows}
\check{\tau}_i(t,\bm{x}) \coloneqq \sum_{j=0}^n f_{ij}(t,\bm{x}) \quad \mbox{and} \quad \hat{\tau}_i(t,\bm{x}) \coloneqq \sum_{j=0}^n f_{ji}(t,\bm{x})
\end{equation}
for $i=1,\ldots,n$. Index $j=0$ represents the system exterior.

Let $\Omega \subset \mathbb{R}^n$ be a domain (connected, open set) and $\mathcal{I} \subset \mathbb{R}$ be an open interval. We assume that $f_{ij} (t,\bm{x})$ is a continuous, nonlinear function of ${\bm x}$ and $t$, and continuously differentiable function of $\bm{x}$ on $\mathcal{I} \times \Omega$. Because of being linear combinations of $f_{ij} (t,\bm{x})$, $\check{\tau}_i (t,\bm{x})$, $\hat{\tau}_i (t,\bm{x})$, and ${\tau}_i (t,\bm{x})$ have also the same properties. These conditions imply the existence and uniqueness of the solutions to the system of governing equations, Eq~\ref{eq:modelTau}.

We assume the following conditions on the {\em flow rate functions}:
\begin{equation}
\label{eq:sf}
f_{ij}(t,\bm{x}) \coloneqq q^x_{ij}(t,\bm{x}) \, {x}_j (t), \quad f_{ij}(t,\bm{x}) \geq 0,
\quad \forall i,j
\end{equation}
where $q^x_{ij}(t,\bm{x})$ is a nonlinear function of ${\bm x}$ and $t$, and has the same properties as $f_{ij}(t,\bm{x})$. In general, it is assumed for compartmental systems that $f_{ii}(t,\bm{x}) = 0$, but the following analysis is also valid for nonnegative flow from a compartment to itself. The first condition of Eq.~\ref{eq:sf} guaranties nonnegativity of the state variables, that is, $x_j(t) \geq 0$ for all $j$. The {\em external input} and {\em output flow rates} into and from compartment $i$, $f_{i0}(t,\bm{x})$ and $f_{0i}(t,\bm{x})$, are respectively denoted by
\[z_i(t,\bm{x}) \coloneqq f_{i0}(t,\bm{x}) \quad \mbox{and} \quad y_i(t,\bm{x}) \coloneqq f_{0i}(t,\bm{x}) . \]

The governing system can then be rewritten componentwise as
\begin{equation}
\label{eq:model1}
\dot{x}_i(t) = \check{\tau}_i(t,\bm{x}) - \hat{\tau}_i(t,\bm{x})
\end{equation}
for $i=1,\ldots, n $. When the external input and output are separated, the system Eq~\ref{eq:model1} takes the following standard form:
\begin{equation}
\label{eq:model2}
\dot{x}_i(t) = \left( z_i(t,\bm{x}) + \sum_{\substack{j=1}}^n f_{ij}(t,\bm{x})  \right) - \left( y_i(t,\bm{x}) + \sum_{\substack{j=1}}^n f_{ji}(t,\bm{x}) \right)
\end{equation}
with the initial conditions of $x_i(t_0) = x_{i,0}$. The positive input or initial condition ensures that the corresponding state variable is always strictly positive. That is, if $z_i(t,\bm{x}) > 0$ or $x_{i,0} > 0$, then $x_i(t) >0$.

The proposed methodology is designed for {\em conservative compartmental systems}, as defined below.
\begin{definition}
\label{def:conservative}
A dynamical system will be called {\em compartmental} if it can be expressed in the form of Eq.~\ref{eq:model2} with the conditions given in Eq.~\ref{eq:sf}. The compartmental system will be called {\em conservative} if all internal flow rates add up to zero when the system is closed, that is, when there is neither external input nor output. Formally,
\begin{equation}
\label{eq:consv}
\sum_{i=1}^{n} \dot{x}_ i(t) = 0 \quad \mbox{when} \quad \bm{z}(t,\bm{x}) = \bm{y}(t,\bm{x}) = \bm{0} \quad \mbox{on } \, \mathcal{I}
\end{equation}
where $\bm 0$ is the zero vector of size $n$.
\end{definition}

We define the {\em state}, {\em input}, and {\em output matrix} functions as
\[ \mathcal{X}(t) \coloneqq \diag\left ( \bm{x}(t) \right ), \quad \mathcal{Z}(t,\bm{x}) \coloneqq \diag\left ( \bm{z}(t,\bm{x}) \right ), \quad \mbox{and} \quad \mathcal{Y}(t,\bm{x}) \coloneqq \diag\left ( \bm{y}(t,\bm{x}) \right ),
\]
respectively. The notation $\diag{(x(t))}$ represents the diagonal matrix whose diagonal elements are the elements of vector $x(t)$, and $\diag{(X(t))}$ represents the diagonal matrix whose diagonal elements are the same as the diagonal elements of matrix $X(t)$. The {\em external input} and {\em output vectors} are
\[\bm{z}(t,\bm{x}) =[z_1(t,\bm{x}),\ldots,z_n(t,\bm{x})]^T \quad \mbox{and} \quad \bm{y}(t,\bm{x}) = [y_1(t,\bm{x}),\ldots,y_n(t,\bm{x})]^T ,\]
respectively. Clearly,
\[ \bm{x}(t) = \mathcal{X}(t) \, \bm{1} , \quad \bm{z}(t,\bm{x}) = \mathcal{Z}(t,\bm{x}) \, \bm{1}, \quad \mbox{and} \quad \bm{y}(t,\bm{x}) = \mathcal{Y}(t,\bm{x}) \, \bm{1} \]
where $\mathbf{1}$ is the vector of size $n$ whose entries are all equal to $1$. Excluding the external input and output, we define the {\em flow rate matrix} function as the matrix of intercompartmental direct flows:
\begin{equation}
\label{eq:matrix_rates}
F(t,\bm{x}) \coloneqq \left ( f_{ij}(t,\bm{x}) \right ).
\end{equation}

Using these notations, $\check{\bm{\tau}}(t,\bm{x})$ and $\hat{\bm{\tau}}(t,\bm{x})$, defined in Eq.~\ref{eq:in_out_flows}, can be expressed in compact form as
\begin{equation}
\label{eq:vector_rates}
\begin{aligned}
\check{\bm{\tau}}(t,\bm{x}) &= \mathcal{Z}(t,\bm{x})  \, \bm{1} + F(t,\bm{x}) \, \bm{1} = \bm{z}(t,\bm{x}) + F(t,\bm{x}) \, \bm{1} , \\
\hat{\bm{\tau}}(t,\bm{x}) &= \mathcal{Y}(t,\bm{x})  \, \bm{1} + F^T(t,\bm{x}) \, \bm{1} = \bm{y}(t,\bm{x})  + F^T(t,\bm{x}) \, \bm{1} .
\end{aligned}
\end{equation}
The governing system, Eq.~\ref{eq:model2}, then becomes
\begin{equation}
\label{eq:model_orgM1}
\begin{aligned}
\dot{\bm{x}}(t) = \left ( \bm{z}(t,\bm{x}) + F(t,\bm{x}) \, \bm{1} \right ) - \left ( \bm{y}(t,\bm{x})  + F^T(t,\bm{x}) \, \bm{1} \right )
\end{aligned}
\end{equation}
with the initial conditions of ${\bm x}(t_0) = {\bm x}_0$. Separating external inputs from the intercompartmental flows and outputs, the governing system, Eq.~\ref{eq:model_orgM1}, can be written as
\begin{equation}
\label{eq:model_orgM2}
\begin{aligned}
\dot{\bm{x}}(t) = \bm{z}(t,\bm{x}) + {\sf F} (t,\bm{x}) \, \mathbf{1}
\end{aligned}
\end{equation}
where
\begin{equation}
\label{eq:matrix_lambda}
\begin{aligned}
{\sf F}(t,\bm{x}) & \coloneqq F (t,\bm{x}) - \mathcal{Y} (t,\bm{x}) - \diag {\left ( F^T (t,\bm{x}) \, \mathbf{1} \right ) } = F (t,\bm{x}) - \mathcal{T}(t,\bm{x}) ,
\end{aligned}
\end{equation}
and $\mathcal{T}(t,\bm{x}) \coloneqq \diag \left ( \hat{\bm{\tau}}(t,\bm{x}) \right ) = \mathcal{Y} (t,\bm{x}) + \diag {\left ( F^T (t,\bm{x}) \, \mathbf{1} \right ) } $.

\subsection{System decomposition methodology}
\label{sec:systemd}

In this section, we introduce the {\em dynamic system decomposition methodology} for analytically partitioning the governing system into mutually exclusive and exhaustive {\em subsystems}. The system decomposition methodology dynamically decomposes composite compartmental flows and storages into subcompartmental segments based on their constituent sources from the initial stocks and external inputs. The system decomposition, consequently, yields subthroughflows and substorages for the distribution of the initial stocks and external inputs, as well as the organization of the associated storages derived from these stocks and inputs individually and separately within the system. In other words, this methodology enables tracking the evolution of the initial stocks and external inputs, as well as associated storages individually and separately within the system. The dynamic {\em system decomposition methodology} has two components: the {\em state} and {\em flow decompositions}, as introduced in this section (see Figs.~\ref{fig:sc} and~\ref{fig:fd}).

The system is partitioned explicitly and analytically into mutually exclusive and exhaustive {\em subsystems} as follows: Each compartment is partitioned into $n+1$ subcompartments; $n$ initially empty subcompartments for $n$ environmental inputs and $1$ subcompartment for the initial stock of the compartment. The notation $i_k$ is used to represent the $k^{th}$ subcompartment of the $i^{th}$ compartment for $i=1,\ldots,n$ and $k=0,\ldots,n$. The subscript index $k=0$ represents the initial subcompartment of compartment $i$ (see Fig.~\ref{fig:sc}).

The initial subsystem is driven by the initial stocks. Except the initial subsystem, each subsystem is generated by a single external input. Therefore, the number of subcompartments in each compartment is equal to the number of inputs (or compartments), plus one for the initial stocks. If an input or all initial conditions are zero, the corresponding subsystem becomes null. Consequently, in a system with $n$ compartments, each compartment has $n+1$ subcompartments, and, therefore, the system has $n+1$ subsystems, indexed by $k=0,\ldots,n$. The initial subsystem ($k=0$) represents the evolution of the initial stocks, receives no external input, and has the same initial conditions as the original system. The initial conditions for all the other subcompartments are zero.

The system decomposition introduced in this section is input-oriented. The governing system can be partitioned based on external outputs instead of inputs, by conceptually reversing all system flows. The following condition on the flow rates, instead of Eq.~\ref{eq:sf}, ensures the possibility of the system decomposition and analysis in both the input- and output-orientations:
\begin{equation}
\label{eq:sfO}
f_{ij}(t,\bm{x}) \coloneqq q^x_{ij}(t,\bm{x}) \, {x}_i (t) \, {x}_j (t) , \quad i,j = 1. \ldots, n .
\end{equation}
This form of flow rates makes both the original and reversed decomposed systems well-defined.

The initial subsystem will be further decomposed into {\em initial subsystems} to track the fate of the initial stocks individually and separately within the system. This dynamic {\em initial system decomposition methodology} is introduced in Appendix~\ref{sec:initsystemd} (see Fig.~\ref{fig:sc} and~\ref{fig:fd}).

\subsubsection{State decomposition}
\label{sec:sd}

In this section, we will introduce the {\em state decomposition} ({\em subcompartmentalization}) methodology.
\begin{figure}[t]
\begin{center}
\begin{tikzpicture}
\centering
   \draw[very thick, draw=black] (-.05,-.05) rectangle node(R1) [pos=.5] { } (2.1,2.1) ;
   \draw[very thick, fill=red!3, draw=red, text=red] (0,0) rectangle node(R1) [pos=.5] {$x_{1_1}$} (1,0.8) ;
   \draw[very thick, fill=gray!10, draw=black, text=black] (1.05,0) rectangle node(R2) [pos=.5] {${x}_{1_0}$} (2.05,1) ;
   \draw[very thick, fill=brown!3, draw=brown, text=brown] (1.05,1.05) rectangle node(R3) [pos=.5] {$x_{1_3}$} (2.05,2.05) ;
   \draw[very thick, fill=blue!3, draw=blue, text=blue] (0,0.85) rectangle node(R4) [pos=.5] {$x_{1_2}$} (1,2.05) ;
    \draw[very thick,-stealth,draw=red]  (-1,.5) -- (-.1,.5) ;
    \node (z) [text=red] at (-.7,0.7) {$z_1$};
    \draw[very thick,-stealth,draw=black]  (-.3,1.8) -- (-1.3,1.8) ;
    \node (z) [text=black] at (-.7,2.05) {$y_1$};
    \node (x) at (1,-.4) {${x}_{1}$};
   \draw[very thick, draw=black] (5.95,-.05) rectangle node(R1) [pos=.5] { } (8.1,2.1) ;
   \draw[very thick, fill=red!3, draw=red, text=red] (6,0) rectangle node(R1) [pos=.5] {$x_{2_1}$} (7,0.8) ;
   \draw[very thick, fill=gray!10, draw=black, text=black] (7.05,0) rectangle node(R2) [pos=.5] {${x}_{2_0}$} (8.05,.6) ;
   \draw[very thick, fill=brown!3, draw=brown, text=brown] (7.05,2.05) rectangle node(R3) [pos=.5] {$x_{2_3}$} (8.05,0.65) ;
   \draw[very thick, fill=blue!3, draw=blue, text=blue] (6,0.85) rectangle node(R4) [pos=.5] {$x_{2_2}$} (7,2.05) ;
    \draw[very thick,-stealth,draw=blue]  (5,1.5) -- (5.9,1.5) ;
    \node (z) [text=blue] at (5.3,1.7) {$z_2$};
    \draw[very thick,-stealth,draw=black]  (8.4,1.8) -- (9.4,1.8) ;
    \node (z) [text=black] at (8.8,2.05) {$y_2$};
    \node (x) at (7,-.4) {${x}_{2}$};
   \draw[very thick, draw=black] (2.95,2.95) rectangle node(R1) [pos=.5] { } (5.1,5.1) ;
   \draw[very thick, fill=red!3, draw=red, text=red] (3,3) rectangle node(R1) [pos=.5] {$x_{3_1}$} (4,4.2) ;
   \draw[very thick, fill=gray!10, draw=black, text=black] (4.05,3) rectangle node(R2) [pos=.5] {${x}_{3_0}$} (5.05,3.7) ;
   \draw[very thick, fill=brown!3, draw=brown, text=brown] (4.05,3.75) rectangle node(R3) [pos=.5] {$x_{3_3}$} (5.05,5.05) ;
   \draw[very thick, fill=blue!3, draw=blue, text=blue] (3,4.25) rectangle node(R4) [pos=.5] {$x_{3_2}$} (4,5.05) ;
    \draw[very thick,-stealth,draw=brown]  (6,4.5) -- (5.1,4.5) ;
    \node (z) [text=brown] at (5.9,4.7) {$z_3$};
    \draw[very thick,-stealth,draw=black]  (2.7,4.5) -- (1.7,4.5) ;
    \node (z) [text=black] at (2.3,4.75) {$y_3$};
    \node (x) at (4,5.4) {${x}_{3}$};
    \draw[very thick,-stealth]  (0.8,2.5) -- (2.6,4.2) ;
    \draw[very thick,-stealth]  (2.6,3.9)  -- (1.15,2.5) ;
    \node (x) at (2.1,2.8) {${f}_{13}$};
    \node (x) at (1.2,3.5) {${f}_{31}$};
    \draw[very thick,-stealth]  (7.3,2.5) -- (5.4,4.2) ;
    \draw[very thick,-stealth]  (5.4,3.9) -- (7,2.5) ;
    \node (x) at (5.9,2.9) {${f}_{23}$};
    \node (x) at (6.9,3.5) {${f}_{32}$};
    \draw[very thick,-stealth]  (2.5,.77) -- (5.5,0.77) ;
    \draw[very thick,-stealth]  (5.5,1) -- (2.5,1) ;
    \node (x) at (4,1.35) {${f}_{12}$};
    \node (x) at (4,.45) {${f}_{21}$};
\end{tikzpicture}
\end{center}
\caption{Schematic representation of the dynamic subcompartmentalization in a three-compartment model system. Each subsystem is colored differently; the second subsystem ($k=2$) is blue, for example. Only the subcompartments in the same subsystem ($x_{1_2}(t)$, $x_{2_2}(t)$, and $x_{3_2}(t)$ in the second subsystem, for example) interact with each other. Subsystem $k$ receives external input only at subcompartment ${k_k}$. The initial subsystem (gray) receives no external input. The dynamic flow decomposition is not represented in this figure. Compare this figure with Fig.~\ref{fig:fd}, in which the subcompartmentalization and corresponding flow decomposition are illustrated for $x_1(t)$ only.}
\label{fig:sc}
\end{figure}
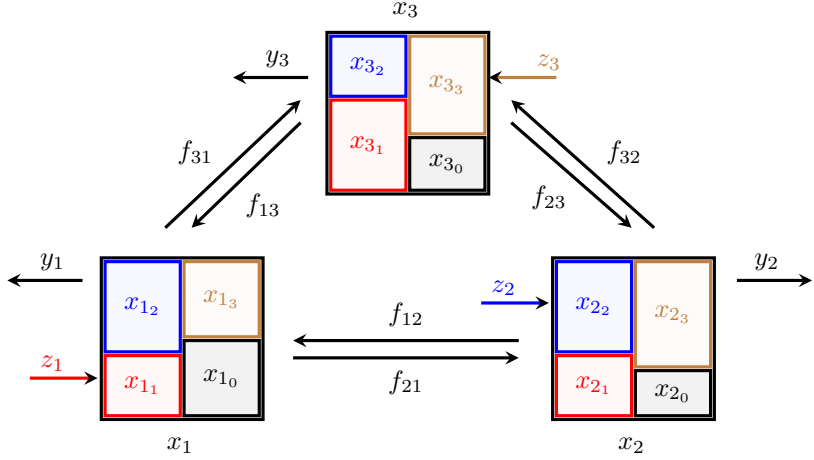

The $k^{th}$ {\em substate} of the $i^{th}$ state variable, $x_i(t)$, is denoted by $x_{i_k}(t)$. This substate variable, $x_{i_k}(t)$, represents the storage in {\em compartment} $i$ at time $t$ that is generated by external input into compartment $k \neq 0$, $z_k(t,\bm{x})$, during $[t_0,t]$ (see Fig.~\ref{fig:sc}). Therefore, $x_{i_k}(t)$ will also be called {\em substorage function}. The $0^{th}$ substate of the $i^{th}$ state, $x_{i_0}(t)$, will be called the {\em initial substorage} or {\em substate} of $x_i(t)$, and it is initially equal to the initial condition $x_{i,0}$. The initial substates represent the evolution of the initial stocks for $t>t_0$. Consequently, all substate variables are assumed to be zero initially, except the initial substate. Due to the mutual exclusiveness and exhaustiveness of the system decomposition, we then have
\begin{equation}
\label{eq:decomposition}
x_i(t) = \sum \limits_{k=0}^n x_{i_k}(t)
\end{equation}
for $i=1,\ldots,n$, and the initial conditions are
\begin{equation}
\label{eq:inits}
x_{i_k}(t_0) = \left \{
\begin{aligned}
& x_i(t_0) = x_{i,0}, \quad k = 0 \\
& 0.  \quad \quad \quad \quad \quad \quad k \neq 0
\end{aligned}
\right .
\end{equation}

Similar to the decomposition of a compartment, each initial subcompartment is further decomposed into $n$ subcompartments, as explained in Appendix~\ref{sec:isd} (see Fig.~\ref{fig:fd}). We will use the notation $x_{i_{k,0}} (t)$ for the $k^{th}$ substate of the $i^{th}$ initial substate function, $x_{i_0}(t)$, or ${\ubar x}_{i_k} (t) \coloneqq x_{i_{k,0}} (t)$ for notational convenience. Based on this further decomposition of the initial substates, we also have
\begin{equation}
\label{eq:i_decomposition}
{\ubar x}_{i}(t) \coloneqq x_{i_0}(t) = \sum \limits_{k=1}^n {\ubar x}_{i_k}(t),
\end{equation}
for $i=1,\ldots,n$, and the corresponding initial conditions become
\begin{equation}
\label{eq:i_inits}
{\ubar x}_{i_k}(t_0) = \delta_{ik} \, x_{i_0}(t_0)  = \left \{
\begin{aligned}
& x_{i_0}(t_0) = x_{i,0}, \quad i = k \\
& 0. \quad \quad \quad \quad \quad \quad \, \, \, i \neq k
\end{aligned}
\right .
\end{equation}

Let the {\em $k^{th}$ substate} and {\em initial substate vector} functions be defined as
\begin{equation}
\label{eq:s_vects}
\begin{aligned}
\mathbf{x}_k(t) & \coloneqq [x_{1_k}(t), \ldots, x_{n_k}(t) ]^T, \quad k=0,\ldots,n, \quad \mbox{and} \\  \ubar{\mathbf{x}}_k(t) & \coloneqq [\ubar{x}_{1_k}(t), \ldots, \ubar{x}_{n_k}(t) ]^T, \quad k=1,\ldots,n.
\end{aligned}
\end{equation}
The vector function ${\mathbf x}(t)$ of all initial substate and substate variables of the decomposed system will be denoted by
\begin{equation}
\label{eq:barkappa}
\begin{aligned}
{\mathbf x}(t) \coloneqq & \left[ \ubar{\bf{x}}^T_1(t), \ldots, \ubar{\bf{x}}^T_n(t), {\bf{x}}^T_1(t), \ldots, {\bf{x}}^T_n(t) \right ]^T \\
 = & \left[ \ubar{x}_{1_1}(t), \ldots,\ubar{x}_{n_1}(t), \ldots, \ubar{x}_{1_n}(t), \ldots,\ubar{x}_{n_n}(t) , \right . \\ & \left . x_{1_1}(t), \ldots,x_{n_1}(t), \ldots, x_{1_n}(t), \ldots,x_{n_n}(t) \right ]^T  \in \mathbb{R}^{2n^2} . \nonumber
\end{aligned}
\end{equation}
The state decompositions formulated in Eqs.~\ref{eq:decomposition} and~\ref{eq:i_decomposition} and the corresponding initial values in Eqs.~\ref{eq:inits} and~\ref{eq:i_inits} can then be expressed in vector form as
\begin{equation}
\label{eq:v_inits}
\begin{aligned}
\bm{x}(t) = \ubar{\bm{x}}(t) + \bar{\bm{x}}(t), \quad & \bm{x}(t_0) = \bm{x}_0 \quad \mbox{where}  \\
\ubar{\bm{x}}(t) \coloneqq \mathbf{x}_0(t) = \ubar{\mathbf{x}}_1(t) + \ldots + \ubar{\mathbf{x}}_n(t), \quad & \ubar{\mathbf{x}}_k(t_0) = {x}_{k,0} \, \bm{e}_k, \quad \ubar{\bm{x}}(t_0) = \bm{x}_0 , \quad \mbox{and} \\
\bar{\bm{x}}(t) \coloneqq \mathbf{x}_1(t) + \ldots + \mathbf{x}_n(t), \quad & \mathbf{x}_k(t_0) = \bm{0} , \quad \bar{\bm{x}}(t_0) = \bm{0} ,
\end{aligned}
\end{equation}
for $k=1,\ldots,n $, where $\mathbf{e}_k$ is the standard elementary unit vector of size $n$. The vector functions $\ubar{\bm{x}}(t)$ and $\bar{\bm{x}}(t)$ are the partitions of the state variable $\bm{x}(t)$, which respectively represent the storages derived from the initial stocks and external inputs. Equations~\ref{eq:decomposition} and~\ref{eq:i_decomposition} also imply that
\begin{equation}
\label{eq:xdervs}
\begin{aligned}
x_i(t) = \sum_{k=1}^n \ubar{x}_{i_k}(t) + x_{i_k}(t)
\quad \Rightarrow \quad
{\dot x}_i(t) = \sum_{k=1}^{n}  \dot {\ubar x}_{i_k}(t) + {\dot x}_{i_k}(t) .
\end{aligned}
\end{equation}
In vector notation, that is,
\begin{equation}
\label{eq:xdervsV}
\begin{aligned}
\dot{\bm x}(t) = \dot{\ubar{\bm x}}(t) + \dot{\bar{\bm{x}}}(t) = \left( \dot{\ubar{\mathbf{x}}}_1(t) + \ldots + \dot{\ubar{\mathbf{x}}}_n(t) \right) +  \left( \dot{\mathbf{x}}_1(t) + \ldots + \dot{\mathbf{x}}_n(t) \right) .
\end{aligned}
\end{equation}

We define the {\em substate} and $k^{th}$ {\em substate matrix} functions, $X(t)$ and $\mathcal{X}_k(t)$, as
\begin{equation}
\label{eq:X}
X(t) \coloneqq \left ( x_{i_k}(t) \right ) = \left [ \mathbf{x}_1(t) \, \ldots \, \mathbf{x}_n(t) \right ] \quad \mbox{and} \quad \mathcal{X}_k(t) \coloneqq \diag\left ( \mathbf{x}_k(t) \right )
\end{equation}
for $k=0,\ldots,n$, together with the initial conditions given in Eq.~\ref{eq:inits},
\begin{equation}
\label{eq:Xinit}
X(t_0) = \mathbf{0}, \quad \mathcal{X}_k(t_0) = \mathbf{0}  \, \, \, \mbox{for} \, \, \, k \neq 0 , \quad \mbox{and} \quad \mathcal{X}_0(t_0) = \diag\left ( \bm{x}_0 \right ) .
\end{equation}
These matrices will, alternatively, be called the {\em substorage} and $k^{th}$ {\em substorage matrix} functions, respectively. Note that we use the notation $\bm 0$ for both the $n \times 1$ zero vector and $n \times n$ zero matrix, which should be distinguished from the context. We then have
\begin{equation}
\label{eq:decom_sum}
\begin{aligned}
\bm{x}(t) = \ubar{\bm{x}}(t) + \bar{\bm{x}}(t) = \mathbf{x}_0(t) + {X}(t) \, \mathbf{1} \quad \mbox{and} \quad  \mathbf{x}_k(t) = \mathcal{X}_k(t) \, \mathbf{1} .
\end{aligned}
\end{equation}
The dynamic state decomposition methodology can be schematized as follows:
\begin{center}
\begin{tikzpicture}
  \matrix (m) [matrix of math nodes,row sep=0em,column sep=13em,minimum width=1em]
  {
  \bm{x}(t) &
 \hspace{-4.4cm} \mathbf{x}_0(t) + X(t) \, \bm{1}\\
 \veq &  \hspace{-4.4cm}  \veq \\
     {
 \begin{bmatrix}
  x_{1} \\
  x_{2} \\
  \vdots  \\
  x_{n}
 \end{bmatrix}
 }
 &
 {
  \begin{bmatrix}
  x_{1_0} \\
  x_{2_0} \\
  \vdots  \\
  x_{n_0}
 \end{bmatrix}
+
 \begin{bmatrix}
  x_{1_1} & x_{1_2} & \cdots & x_{1_n} \\
  x_{2_1} & x_{2_2} & \cdots & x_{2_n} \\
  \vdots  & \vdots  & \ddots & \vdots  \\
  x_{n_1} & x_{n_2} & \cdots & x_{n_n}
 \end{bmatrix}
  \begin{bmatrix}
  {1} \\
  {1} \\
  \vdots  \\
  {1}
 \end{bmatrix}
 }
 \\
 };
  \path[-stealth, decorate]
    (m-3-1.east|-m-3-2) edge node [below] {state decomposition}
            node [above] {$x_i(t) = \sum \limits_{k=0}^n x_{i_k}(t)$} (m-3-2) ;
\end{tikzpicture}
\end{center}

\subsubsection{Flow decomposition}
\label{sec:rd}

For a coherent system partitioning, flow rates are also decomposed into flow segments that will be called the {\em subflow rate} functions. These subflow rates represent the rate of flow segments between the subcompartments in the same subsystem (see Fig.~\ref{fig:fd}).
\begin{figure}[t]
\begin{center}
\begin{tikzpicture}[xscale=.7, yscale=.6]
\centering
    \draw[fill=gray!3]    (5,0) -- ++(3,0) -- ++(2,3) -- ++(-2,3) -- ++(-3,0);  %
    \draw[draw=none,fill=red!5]  (5,2.8) -- ++ (4.75,0) -- ++ (.15,.15) -- ++ (-.6,.9) -- ++ (-4.3,0);   %
    \draw[draw=none,fill=gray!30]  (5,.05) -- ++ (2.95,0) -- ++ (1.8,2.7) -- ++ (-4.75,0);
    \draw[draw=none,fill=blue!5]  (5,3.9) -- ++ (4.25,0) -- ++ (-.6,.9) -- ++ (-3.65,0);
    \draw[draw=none,fill=brown!5]  (5,5.9) -- ++ (2.95,0) -- ++ (.7,-1.05) -- ++ (-3.7,0);
   \draw[very thick, draw=black] (-.05,-.05) rectangle node(R1) [pos=.5] { } (4.1,6.1) ;
   \draw[very thick, fill=red!5, draw=red, text=red] (0,0) rectangle node(R1) [pos=.5] {$x_{1_1}$} (2,2) ;
   \draw[very thick, fill=gray!30, draw=black, text=black] (2.05,0) rectangle node(R2) [pos=.5] {} (4.05,3.5) ;
   \draw[very thick, fill=brown!5, draw=brown, text=brown] (2.05,3.55) rectangle node(R3) [pos=.5] {$x_{1_3}$} (4.05,6.05) ;
   \draw[very thick, fill=blue!5, draw=blue, text=blue] (0,2.05) rectangle node(R4) [pos=.5] {$x_{1_2}$} (2,6.05) ;
   \draw[very thick, draw=red, text=red] (2.1,0.05) rectangle node(R21) [pos=.5] {\footnotesize $\ubar{x}_{1_1}$} (4,1.15) ;
   \draw[very thick, draw=blue, text=blue] (2.1,1.2) rectangle node(R22) [pos=.5] {\footnotesize $\ubar{x}_{1_2}$} (4,2.55) ;
   \draw[very thick, draw=brown, text=brown] (2.1,2.6) rectangle node(R23) [pos=.5] {\footnotesize $\ubar{x}_{1_3}$} (4,3.45) ;
    \draw[very thick]  (5,0) -- (8,0) ;
    \draw[very thick]  (8,0) -- (10,3) ;
    \draw[very thick]  (10,3) -- (8,6) ;
    \draw[very thick]  (8,6) -- (5,6) ;
    \draw[very thick,-stealth,draw=red]  (-1.7,.5) -- (-.1,.5) ;
    \node (z) [text=red] at (-1.1,1) {$z_1$};
    \draw[very thick,-stealth,draw=red]  (5,.4) -- (7.9,.4) ;
    \node (fb1) [text=red] at (6.2,0.8) {\footnotesize $\ubar{f}_{j_1 1_1}$};
    \draw[very thick,-stealth,draw=blue]  (5,1.2) -- (8.4,1.2) ;
    \node (fb2)  [text=blue] at (6.2,1.6) {\footnotesize $\ubar{f}_{j_2 1_2}$};
    \draw[very thick,-stealth,draw=brown]  (5,2) -- (8.8,2) ;
    \node (fb3)  [text=brown] at (6.2,2.4) {\footnotesize $\ubar{f}_{j_3 1_3}$};
    \node (f0) at (9.7,1) {${f}_{j_0 1_0}$};
    \node (x0) at (3,-.5) {${x}_{1_0}$};
    \draw[very thick,-stealth,draw=red]  (5,3) -- (9.5,3) ;
    \node (f1) [text=red] at (6.2,3.4) {${f}_{j_1 1_1}$};
    \draw[very thick,-stealth,draw=blue]  (5,4.05) -- (8.8,4.05) ;
    \node (f2) [text=blue] at (6.2,4.45) {${f}_{j_2 1_2}$};
    \draw[very thick,-stealth,draw=brown]  (5,5.1) -- (8.2,5.1) ;
    \node (f3) [text=brown] at (6.2,5.5) {${f}_{j_3 1_3}$};
    \node (x) at (2,6.6) {${x}_{1}$};
    \node (x) at (6.5,6.6) {${f}_{j1}$};
\end{tikzpicture}
\end{center}
\caption{Schematic representation of the dynamic flow decomposition in a three-compartment model system. The figure illustrates the subcompartmentalization of compartment $1$ and the corresponding flow decomposition of $f_{j1}(t,{\bm x})$. The figure also illustrates further decomposition of initial subcompartment $1_0$ and the corresponding initial subflow function, $f_{j_0 1_0}(t,\bf{x})$ (both dark gray).}
\label{fig:fd}
\end{figure}
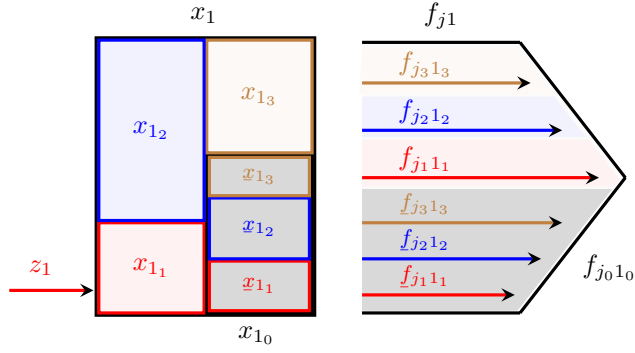

We assume that external input $z_k(t,\bm{x})$ enters into the system at subcompartment $k_k$ (see Fig.~\ref{fig:sc}). Moreover, no other $k^{th}$ subcompartment of any other compartment $i$, that is, subcompartment $i_k$, receives external input. Therefore, the {\em input decomposition} can be expressed as
\begin{equation}
\label{eq:Z}
z_{i_k}(t,{\bf x}) \coloneqq \delta_{ik} \, z_{k}(t,\bm{x}) = \left \{
\begin{aligned}
& z_{k_k}(t,{\bf x}) = z_{k}(t,\bm{x}), \quad i = k \\
& 0. \quad \quad \quad \quad \quad \quad \quad \quad \, \, \, \, \, i \neq k
\end{aligned}
\right .
\end{equation}

The flow rates, $f_{ij}(t,\bm{x})$, and output functions, $y_j(t,\bm{x})$, are also decomposed into the subflow rate functions. First, we define the {\em flow intensity} per unit storage directed from compartment $j$ to $i$ at time $t$ as
\begin{equation}
\label{eq:rintensities}
q^x_{ij}(t,\bm{x}) = \frac{f_{ij} (t,\bm{x}) }{x_j (t)}
\end{equation}
for $i=0,\ldots,n$, $j=1,\ldots,n$, as formulated in Eq.~\ref{eq:sf}. The subflow rates are then defined to be the flow segments proportional to the corresponding substorages, $x_{j_k}(t)$, with the proportionality factor of the flow intensity in the flow direction, $q^{x}_{ij}(t,{\bm x})$. That is,
\begin{equation}
\label{eq:cons}
f_{{i_k}{j_k}}(t,{\mathbf x}) \coloneqq  {x_{j_k}(t)} \, \frac{f_{ij}(t,\bm{x})} {x_{j}(t)} = \frac{x_{j_k}(t)} {x_{j}(t)} \, {f_{ij}(t,\bm{x})}
\end{equation}
for $i,k=0,\ldots,n$ and $j=1,\ldots,n$. The index $0_k$ is equivalent to the index $0$, and both represent the system exterior. We will use index $0$ in both cases for notational convenience. Similar to Eq.~\ref{eq:sf}, the functions $f_{i_k j_k}(t,\mathbf{x}) \geq 0$ represent nonlinear, nonnegative {\em subflow rates} from subcompartment ${j_k}$ to ${i_k}$. Due to the mutual exclusiveness and exhaustiveness of the system decomposition and Eq.~\ref{eq:decomposition}, we have
\begin{equation}
\label{eq:rate_decomposition}
{f_{ij}(t,\bm{x})} = \sum \limits_{k=0}^n  f_{{i_k}{j_k}}(t,{\mathbf x})
\end{equation}
for $i,j=1,\ldots,n$. It can be seen from Eq.~\ref{eq:cons} that the flow and subflow intensities between the same compartments in the same flow direction are the same, that is,
\begin{equation}
\label{eq:csc_dense}
\frac{ f_{{i_k}{j_k}}(t,{\mathbf x}) }{ x_{j_k}(t) } =  \frac{ f_{ij}(t,\bm{x}) } {x_{j}(t)}
\end{equation}
(see Fig.~\ref{fig:fd}).

In Eq.~\ref{eq:cons} above,
\begin{equation}
\label{eq:cons2}
d_{j_k} ({\mathbf x}) \coloneqq \frac{x_{j_k}(t)} {x_{j}(t)}
\end{equation}
will be called the {\em decomposition factor}. Due to the state decomposition, Eq.~\ref{eq:decomposition}, the decomposition factors form a continuous partition of unity:
\begin{equation}
\label{eq:cfactors}
0 \leq d_{j_k}({\mathbf x}) \leq 1 \quad \mbox{and} \quad \sum_{k=0}^{n} d_{j_k}({\mathbf x}) = 1.
\end{equation}
The {\em decomposition} and $k^{th}$ {\em decomposition} matrices, $D({\mathbf x}) \coloneqq \left ( {d}_{j_k}({\mathbf x}) \right )$ and $\mathcal{D}_k({\mathbf x}) \coloneqq \diag{([ {d}_{1_k}({\mathbf x}), \ldots, {d}_{n_k}({\mathbf x})])}$, can be formulated in matrix form as
\begin{equation}
\label{eq:cfactorsM2}
\begin{aligned}
D({\mathbf x}) &= \mathcal{X}^{-1}(t) \, X(t) \quad \mbox{and} \quad
\mathcal{D}_k({\mathbf x}) = \mathcal{X}^{-1}(t) \, \mathcal{X}_k(t)
\end{aligned}
\end{equation}
for $k=0,\ldots,n$.
Equations~\ref{eq:decom_sum},~\ref{eq:cons2} and~\ref{eq:cfactors} imply that
\begin{equation}
\label{eq:cfactorsM}
\begin{aligned}
\bm{1} = \mathcal{X}^{-1}(t) \, \bm{x}(t) & = \mathcal{X}^{-1}(t) \, \mathbf{x}_0(t) + \mathcal{X}^{-1}(t) \, X(t) \, \bm{1} = \mathcal{D}_0({\mathbf x}) \, \bm{1} + D({\mathbf x}) \, \bm{1} .
\end{aligned}
\end{equation}

We define the {\em $k^{th}$ subflow rate matrix} function as
\begin{equation}
\label{eq:matrix_subratesK}
F_k(t,{\mathbf x}) \coloneqq \left ( f_{{i_k}{j_k}}(t,{\mathbf x}) \right )
\end{equation}
for $k=0,\ldots,n$. Using Eq.~\ref{eq:cons}, $F_k(t,{\mathbf x})$ can be expressed in matrix form as
\begin{equation}
\label{eq:matrix_subratesM}
F_k(t,{\mathbf x}) = F(t,\bm{x}) \, \mathcal{D}_k({\mathbf x}) = F(t,\bm{x}) \, \mathcal{X}^{-1}(t) \, \mathcal{X}_k(t) .
\end{equation}
That is, the $k^{th}$ decomposition matrix, $\mathcal{D}_k({\mathbf x})$, decomposes the compartmental direct flow matrix, $F(t,\bm{x})$, into the subcompartmental subflow matrices, ${F}_k(t,{\mathbf x})$. Similarly, the $k^{th}$ {\em output matrix} function,
\begin{equation}
\label{eq:kthom}
\mathcal{Y}_k(t,{\mathbf x}) \coloneqq  \diag{\left ( \left[ f_{0 1_k}(t,{\mathbf x}) , \ldots, f_{0 n_k}(t,{\mathbf x}) \right] \right )} ,
\end{equation}
can be expressed in matrix form as
\begin{equation}
\label{eq:output_subratesM}
\mathcal{Y}_k(t,{\mathbf x}) = \mathcal{Y}(t,\bm{x}) \, \mathcal{D}_k({\mathbf x}) = \mathcal{Y}(t,\bm{x}) \, \mathcal{X}^{-1}(t) \, \mathcal{X}_k(t) ,
\end{equation}
and the {\em $k^{th}$ input matrix} function can be written as
\begin{equation}
\label{eq:kinputm}
\mathcal{Z}_k(t,{\mathbf x}) \coloneqq
\diag{\left ( {z}_k (t,\bm{x}) \, \bm{e}_k \right )} ,
\end{equation}
for $k=0,\ldots,n$, where we set $\bm{e}_0 = \bm{0}$. The respective $k^{th}$ {\em output} and {\em input vector} functions, ${\mathbf{y}}_k (t,{\mathbf x})$ and ${\mathbf{z}}_k (t,{\mathbf x})$, can then be defined as
\begin{equation}
\label{eq:in_out_k}
\begin{aligned}
{\mathbf{y}}_k (t,{\mathbf x}) \coloneqq {\mathcal{Y}}_k (t,{\mathbf x}) \, \mathbf{1} \quad \mbox{and} \quad {\mathbf{z}}_k (t,{\mathbf x}) \coloneqq {\mathcal{Z}}_k (t,{\mathbf x}) \, \mathbf{1} .
\end{aligned}
\end{equation}

Using these notations, the flow decomposition given in Eq.~\ref{eq:rate_decomposition} and input decomposition formulated in Eq.~\ref{eq:Z} can be written in matrix form as follows:
\begin{equation}
\label{eq:flow_decomposition}
F(t,\bm{x}) = \sum \limits_{k=0}^n  F_k(t,{\mathbf x}) , \quad \mathcal{Y}(t,\bm{x}) = \sum \limits_{k=0}^n \mathcal{Y}_k(t,{\mathbf x}), \quad \mathcal{Z}(t,\bm{x}) = \sum \limits_{k=0}^n  \mathcal{Z}_k(t,{\mathbf x})  .
\end{equation}
The equivalence of the flow and subflow rate intensities given in Eq.~\ref{eq:csc_dense} can also be expressed in matrix form as
\begin{equation}
\label{eq:intensity_flow}
F(t,{\bm x}) \, \mathcal{X}^{-1}(t) = F_k(t,{\mathbf x})  \, \mathcal{X}_k^{-1}(t)
\end{equation}
for $k=0,\ldots,n$.

The flow decomposition given in Eq.~\ref{eq:cons} can then be schematized as follows:
\begin{center}
\begin{tikzpicture}
  \matrix (m) [matrix of math nodes,row sep=0em,column sep=13em,minimum width=1em]
  {
   F (t,\bm{x})
   &
 F_k (t,{\mathbf x}) , \, \, \, k=0,\ldots,n
   \\
 \veq &    \hspace{-2cm} \veq \\
    {
      \begin{bmatrix}
  f_{11} & \cdots & f_{1n} \\
  f_{21} & \cdots & f_{2n} \\
  \vdots  & \ddots & \vdots  \\
  f_{n1} & \cdots & f_{nn}
 \end{bmatrix}
       }
 &
     {
     \begin{bmatrix}
  f_{1_k 1_k} & \cdots & f_{1_k n_k} \\
  f_{2_k 1_k} & \cdots & f_{2_k n_k} \\
  \vdots  & \ddots & \vdots  \\
  f_{n_k 1_k} & \cdots & f_{n_k n_k}
 \end{bmatrix}
      }
 \\ } ;
  \path[-stealth, decorate]
    (m-3-1.east|-m-3-2) edge node [below] {flow decomposition}
            node [above] {$f_{{i_k}{j_k}}(t,{\mathbf x}) = \frac{x_{j_k} (t)} {x_{j} (t)} \, {f_{ij}  (t,\bm{x}) } $} (m-3-2) ;
\end{tikzpicture}
\end{center}

\subsubsection{Domain decomposition}
\label{sec:dd}

The dynamic system decomposition methodology uniquely yields new substate variables in a higher dimensional domain from the original ones. We will show in this section that the properties guarantee the existence and uniqueness of the solutions to the decomposed system on the new domain are inherited from those of the original system on the original domain.

There exists a unique decomposition ${\mathbf x}$ for each $\bm x \in \Omega$ with the relationship $x_i(t) = \sum_{k=1}^n \ubar{x}_{i_k}(t) + x_{i_k}(t) $, as given in Eq.~\ref{eq:xdervs}. The new domain that includes these substate variables is denoted by $\mho \subset \mathbb{R}^{2n^2}$ and will be called the {\em decomposed domain}. This process can be represented as
\begin{equation}
\label{eq:dd}
\mathbb{R}^n \supset \Omega \ni \bm{x} \xrightarrow[\hspace{3.5cm}]{\rm domain \, \, \, decomposition}
\mathbf{x} \in \mho \subset \mathbb{R}^{2n^2} .
\end{equation}
There is a one-to-one correspondence between the original and decomposed domains. This correspondence is due to the existence and uniqueness of the governing systems in both the original and decomposed forms as shown further below in this section.

\begin{proposition}
\label{thm:exist_dcmp}
The nonlinear subflow rate functions $f_{i_k j_k} (t,{\mathbf x})$ and $\ubar f_{i_k j_k} (t,{\mathbf x})$ are continuous and continuously differentiable in ${\mathbf x}$ on domain $\mathcal{I} \times \mho \subset \mathbb{R} \times \mathbb{R}^{2n^2}$.
\end{proposition}
\begin{proof}

It is given that $f_{ij} (t,\bm{x})$ is continuous and continuously differentiable in $\bm{x}$ on $\mathcal{I} \times \Omega  \subset \mathbb{R} \times \mathbb{R}^n$. Note that, due to Eq.~\ref{eq:decomposition}, $x_{j_k}(t) \leq x_{j}(t)$. The decomposition factors $d_{j_k}({\mathbf x}) = x_{j_k}(t)/x_j(t)$ are, therefore, well-defined even if $x_j(t) \to 0$. Note also that the decomposition factors are continuous and continuously differentiable with respect to $x_{j_k}$ on $\mathcal{I} \times \mho$. Therefore,
\begin{equation}
\label{eq:proof_lips}
\begin{aligned}
f_{i_k j_k} (t,{\mathbf x}) = \frac{x_{j_k}(t)}{x_j(t)} \, f_{ij} (t,\bm{x})
\end{aligned}
\end{equation}
is also continuous and continuously differentiable in $x_{j_k}$ on $\mathcal{I} \times \mho$.

By construction, the subthroughflow functions, $\check{\tau}_{i_k} (t,{\mathbf x})$ and ${\hat{\tau}}_{i_k} (t,{\mathbf x})$, as well as the net subthroughflow function, ${\tau}_{i_k} (t,{\mathbf x})$, are linear combination of $f_{i_k j_k} (t,{\mathbf x})$, as formulated in Eq.~\ref{eq:matrix_subrates}. Therefore, they have the same properties as $f_{i_k j_k} (t,{\mathbf x})$. That is, they are continuous and continuously differentiable in ${\mathbf x}$ on domain $\mathcal{I} \times \mho \subset \mathbb{R} \times \mathbb{R}^{2n^2}$. The same arguments with the same conclusions are also valid for $\ubar{f}_{i_k j_k} (t,{\mathbf x})$, $\ubar {\check{\tau}}_{i_k} (t,{\mathbf x})$, $\ubar {\hat{\tau}}_{i_k} (t,{\mathbf x})$, and $\ubar {\tau}_{i_k} (t,{\mathbf x})$.
\end{proof}

\subsubsection{Subsystems}
\label{sec:subsys}

The governing system is analytically and explicitly decomposed into mutually exclusive and exhaustive {\em subsystems} through the dynamic system decomposition methodology composed of the state and flow decomposition components. The $k^{th}$ subcompartments of each compartment together with the corresponding $k^{th}$ substates, subflow rates, inputs, and outputs constitute the $k^{th}$ {\em subsystem}. The subsystems are running within the original system and have the same structure and dynamics as the system itself, except for their external inputs and initial conditions. These otherwise-decoupled mutually exclusive an exhaustive subsystems are coupled through the decomposition factors. By {\em mutual exclusiveness}, we mean that transactions are possible only within corresponding subcompartments of the same subsystem. By {\em exhaustiveness}, we mean that all generated subsystems sum to the entire system so partitioned. The initial subsystem is further decomposed into $n$ subsystems, as formulated in Appendix~\ref{sec:dic} (see Fig.~\ref{fig:sc} and~\ref{fig:fd}).

In Section~\ref{sec:csystems}, the system of governing equations are formulated for the original system, Eq.~\ref{eq:model_orgM1}. In what follows, we will similarly introduce the governing equations for each subsystem. The governing equations for the $k^{th}$ subsystem can be written in vector form as
\begin{equation}
\label{eq:subsystemsV}
\begin{aligned}
\dot {\mathbf x}_k(t)
& = \check{\bm{\tau}}_k(t,{\mathbf x}) - \hat{\bm{\tau}}_k(t,{\mathbf x}) \\
& = \left ( \mathbf{z}_k(t,\bm x) + F_k(t,\mathbf x) \, \bm{1} \right ) - \left ( \mathbf{y}_k(t,\mathbf x)  + F_k^T(t,\mathbf x) \, \bm{1} \right )
\end{aligned}
\end{equation}
for $k = 0,\ldots,n$. The initial conditions are ${\mathbf x}_0(t_0) = \bm{x}_0$ and ${\mathbf x}_k(t_0) = \mathbf{0}$ for $k \neq 0$.

The $k^{th}$ {\em inward} and {\em outward subthroughflow vectors}, $\check{\bm{\tau}}_k(t,{\mathbf x})$ and $\hat{\bm{\tau}}_k(t,{\mathbf x})$, for the $k^{th}$ subsystem given in Eq.~\ref{eq:subsystemsV} can be expressed as
\begin{equation}
\label{eq:matrix_subrates}
\begin{aligned}
{\check{\bm \tau}}_k (t,{\mathbf x}) & \coloneqq {\mathbf z}_k (t,\bm{x}) + F_k (t,{\mathbf x}) \, \mathbf{1} \\
& = {\mathcal Z}_k (t,\bm{x})  \, \mathbf{1} + F (t,\bm{x}) \, \mathcal{X}^{-1}(t) \, {\mathbf x}_k(t) , \\
{\hat{\bm \tau}}_k (t,{\mathbf x}) & \coloneqq {\mathbf y}_k (t,\mathbf{x}) + F^T_k (t,{\mathbf x}) \, \mathbf{1} \\
& = \mathcal{Y}(t,\bm{x}) \, \mathcal{X}^{-1}(t) \, \mathcal{X}_k(t) \, \mathbf{1} + \mathcal{X}_k(t) \, \mathcal{X}^{-1}(t) \, F^T(t,\bm{x}) \, \mathbf{1} \\
& = \left ( \mathcal{Y}(t,\bm{x}) + \diag \left ( F^T (t,\bm{x}) \, \mathbf{1} \right ) \right ) \, \mathcal{X}^{-1}(t) \, {\mathbf x}_k(t) \\
& = \mathcal{T}(t,\bm{x}) \, \mathcal{X}^{-1}(t) \, {\mathbf x}_k(t) .
\end{aligned}
\end{equation}
The $k^{th}$ {\em net subthroughflow rate} vector, ${\bm \tau}_k (t,{\mathbf x}) \coloneqq \left [ {\tau}_{1_k}(t,{\mathbf x}), \ldots, {\tau}_{n_k}(t,{\mathbf x}) \right ]^T$, then becomes
\begin{equation}
\label{eq:total_flows}
\begin{aligned}
{\bm \tau}_k (t,{\mathbf x}) = \, & {\check{\bm \tau}}_k (t,{\mathbf x}) - {\hat{\bm \tau}}_k (t,{\mathbf x})
= \, {\mathbf z}_k (t,\bm{x}) + A (t,\bm{x}) \, {\mathbf x}_k(t)
\end{aligned}
\end{equation}
where
\begin{equation}
\label{eq:matrix_A}
\begin{aligned}
A(t,\bm{x}) \coloneqq {\sf F}(t,\bm{x}) \, \mathcal{X}^{-1} (t)
& =  \left ( F (t,\bm{x}) - \mathcal{T}(t,\bm{x}) \right ) \, \mathcal{X}^{-1}(t) \\
& = Q^x (t,\bm{x}) - \mathcal{R}^{-1}(t,\bm{x}) ,
\end{aligned}
\end{equation}
$Q^x (t,\bm{x}) \coloneqq F (t,\bm{x}) \, \mathcal{X}^{-1}(t)$, and $\mathcal{R}^{-1}(t,\bm{x}) \coloneqq \mathcal{T}(t,\bm{x}) \, \mathcal{X}^{-1}(t)$, assuming $\mathcal{R}(t,\bm{x})$ is invertible. Note that, $A(t,\bm{x})$ is the difference of two matrices $Q^x (t,\bm{x})$ and $\mathcal{R}^{-1}(t,\bm{x})$ whose entries are the intercompartmental flow intensities and outward throughflow intensities per unit storage, respectively. We will, therefore, call $A(t,\bm{x})$ the {\em flow intensity matrix} per unit storage. It is sometimes called the {\em compartmental matrix} \cite{Jacquez1993}. The novel matrix measure $\mathcal{R}(t,\bm{x})$ introduced in this work will be called the {\em residence time matrix}, and the matrix measure $Q^x (t,\bm{x})$ will be called the {\em flow intensity matrix} per unit storage or the {\em flow distribution matrix} for system storages \cite{Coskun2017SCSA}.

For each fixed $j$, Eq.~\ref{eq:csc_dense} implies that
\begin{equation}
\label{eq:thr_dense}
\begin{aligned}
\frac{\hat{\tau}_{j}(t,\bm{x}) }{x_{j}(t)} = \frac{ \sum_{i=0}^{n} f_{i j}(t,\bm{x}) }{x_{j}(t)} = \frac{  \sum_{i=0}^{n} f_{i_k j_k}(t,{\mathbf x}) }{x_{j_k}(t)} = \frac{\hat{\tau}_{j_k}(t,{\mathbf x}) }{x_{j_k}(t)}
\end{aligned}
\end{equation}
for $k=0,\ldots,n$, where the denominators are nonzero. This equivalence between the outward throughflow and subthroughflow intensities given in the first and last equalities of Eq.~\ref{eq:thr_dense} can be expressed in the following matrix form:
\begin{equation}
\label{eq:thr_dense2}
\begin{aligned}
\mathcal{R}^{-1}(t,\bm{x}) = \mathcal{T}(t,\bm{x}) \, \mathcal{X}^{-1}(t)  =
\hat{\mathcal{T}}_k(t,{\bf x}) \, \mathcal{X}_k^{-1} (t)
\end{aligned}
\end{equation}
The residence time matrix, $\mathcal{R}(t,\bm{x})$, has a central role in the integration of various system components and the holistic analysis of the compartmental systems \cite{Coskun2017SCSA,Coskun2017SESM}.

Equations~\ref{eq:csc_dense} and~\ref{eq:thr_dense} also imply that
\begin{equation}
\label{eq:thr_dense3}
\begin{aligned}
\frac{{\hat{\tau}}_{j_k}(t,{\mathbf x}) }{{\hat{\tau}}_{j_\ell}(t,{\mathbf x}) } = \frac{{x}_{j_k}(t) }{{x}_{j_\ell}(t)} = \frac{{f}_{i_k j_k}(t,{\mathbf x}) }{{f}_{i_\ell j_\ell}(t,{\mathbf x})}
\end{aligned}
\end{equation}
for $k,\ell=0,\ldots,n$, where the denominators are nonzero. This relationship indicates the proportionality of the parallel subflows and the corresponding subthroughflows and substorages. By {\em parallel subflows}, we mean the intercompartmental flows that transit through different subcompartments of the same compartment along the same flow path at the same time. The flow path terminology is developed in Appendix~\ref{apxsec:subflow_paths}.

The $k^{th}$ {\em inward} and {\em outward subthroughflow matrices}, $\check{\mathcal{T}}_k(t,{\bf x}) \coloneqq \diag{ \left( {\check{\bm \tau}}_k (t,{\mathbf x}) \right)}$ and $\hat{\mathcal{T}}_k(t,{\bf x}) \coloneqq \diag{ \left( {\hat{\bm \tau}}_k (t,{\mathbf x}) \right)}$, can be formulated as
\begin{equation}
\label{eq:kth_thrmatA}
\begin{aligned}
\check{\mathcal{T}}_k(t,{\bf x})  & = {\mathcal Z}_k (t,\bm{x}) + \diag{ \left( F (t,\bm{x}) \, \mathcal{X}^{-1}(t) \, \mathcal{X}_k(t) \, \bm{1} \right) } , \\
\hat{\mathcal{T}}_k(t,{\bf x})  & = \mathcal{T}(t,\bm{x}) \, \mathcal{X}^{-1}(t)  \, \mathcal{X}_k (t) .
\end{aligned}
\end{equation}
Note that, using Eq.~\ref{eq:kth_thrmatA}, ${F}_k (t,{\mathbf x})$ defined in Eq.~\ref{eq:matrix_subratesM} can, alternatively, be written in terms of the system flows only:
\begin{equation}
\label{eq:TinoutA}
\begin{aligned}
{F}_k (t,{\bf x}) & = {F} (t,\bm{x}) \, \mathcal{T}^{-1}(t,\bm{x})  \, \hat{\mathcal{T}}_k(t,{\bf x})  .
\end{aligned}
\end{equation}
We will call $Q^\tau (t,\bm{x}) \coloneqq F (t,\bm{x}) \, \mathcal{T}^{-1}(t,\bm{x})$ the {\em flow intensity matrix} per unit throughflow or the {\em flow distribution matrix} for system flows in the context of the proposed methodology \cite{Coskun2017SCSA}. Note that the elements of $Q^\tau (t,\bm{x})$, $q^\tau_{ij}(t,\bm{x})$, are sometimes called {\em transfer coefficients}, {\em technical coefficients} in economics, or {\em stoichiometric coefficients} in chemistry.

We also define the {\em inward} and {\em outward subthroughflow matrices}, $\check{T}(t,{\mathbf x})$ and $\hat{T}(t,{\mathbf x})$, as the matrices whose $k^{th}$ columns are the $k^{th}$ inward and outward subthroughflow vectors, ${\check{\bm \tau}}_k(t,{\mathbf x})$ and ${\hat{\bm \tau}}_k(t,{\mathbf x})$, $k=1,\ldots,n$, respectively:
\begin{equation}
\label{eq:matrix_subratesM4}
\begin{aligned}
\check{T} (t,{\mathbf x}) & \coloneqq \left( \check{\bm{\tau}}_{i_k}(t,{\mathbf x}) \right) = \left [ \check{\bm \tau}_1 (t,{\mathbf x}) \, \ldots \, {\check{\bm \tau}}_n (t,{\mathbf x}) \right ] ,
\\
\hat{T} (t,{\mathbf x}) & \coloneqq  \left( \hat{\bm{\tau}}_{i_k}(t,{\mathbf x}) \right) = \left [ {\hat{\bm \tau}}_1 (t,{\mathbf x}) \, \ldots \, {\hat{\bm \tau}}_n (t,{\mathbf x}) \right ] .
\end{aligned}
\end{equation}
Using the relationships in Eq.~\ref{eq:matrix_subrates}, these subthroughflow matrices can be expressed in matrix form as
\begin{equation}
\label{eq:matrix_subrates1}
\begin{aligned}
\check{T} (t,{\mathbf x}) & = \mathcal{Z} (t,\bm{x}) + F (t,\bm{x}) \, \mathcal{X}^{-1} (t) \, X(t) , \\
\hat{T} (t,{\mathbf x}) & = \mathcal{T}(t,\bm{x}) \, \mathcal{X}^{-1}(t) \, X(t) .
\end{aligned}
\end{equation}
We also define the {\em net subthroughflow matrix}, $T(t,{\bf x})$, as
\begin{equation}
\label{eq:total_flowM}
\begin{aligned}
T(t,{\mathbf x}) & \coloneqq \check{T} (t,{\mathbf x}) - \hat{T} (t,{\mathbf x})
= \mathcal{Z} (t,\bm{x}) + A (t,\bm{x}) \, X(t) .
\end{aligned}
\end{equation}

The decomposition matrix $D({\mathbf x})$ can be expressed in terms of the subthroughflow functions, instead of the substate functions, as follows:
\begin{equation}
\label{eq:matrix_Pss0}
\begin{aligned}
D({\mathbf x}) = \mathcal{X}^{-1}(t) \, X(t) = \mathcal{T}^{-1}(t,\bm{x}) \, \hat{T} (t,{\mathbf x})
\end{aligned}
\end{equation}
due to Eq.~\ref{eq:matrix_subrates1}. Note that, the subthroughflow matrices can be written in the following various forms:
\begin{equation}
\label{eq:matrix_subrates1_2}
\begin{aligned}
\tilde{T} (t,{\mathbf x}) & = F(t,\bm{x}) \, D(\mathbf{x}) = Q^x(t,\bm{x}) \, X(t) = Q^\tau(t,\bm{x}) \, \hat{T}(t,\mathbf{x})  \\
\hat{T} (t,{\mathbf x}) & = \mathcal{R}^{-1}(t,\bm{x}) \, X(t)
\end{aligned}
\end{equation}
where $\tilde{T} (t,{\mathbf x}) \coloneqq \check{T} (t,{\mathbf x}) - \mathcal{Z} (t,\bm{x})$ will be called the {\em intercompartmental subthroughflow matrix}. Componentwise, the intercompartmental subthroughflow matrix,\break $\tilde{T} (t,{\mathbf x}) = (\tilde{\tau}_{i_k}(t,{\mathbf x}))$, can be expressed as $\tilde{\tau}_{i_k}(t,{\mathbf x}) = \check{\tau}_{i_k}(t,{\mathbf x}) - z_{i_k}(t,{\mathbf x})$. The $k^{th}$ decomposition matrix, $\mathcal{D}_k({\mathbf x})$, can also be expressed as
\begin{equation}
\label{eq:matrix_Pss0K}
\begin{aligned}
\mathcal{D}_k({\mathbf x}) = \mathcal{X}^{-1}(t) \, \mathcal{X}_k(t) = \mathcal{T}(t,\bm{x})^{-1} \, \hat{\mathcal{T}}_k (t,{\mathbf x}) ,
\end{aligned}
\end{equation}
using Eq.~\ref{eq:thr_dense2}, similar to the decomposition matrix formulated in Eq.~\ref{eq:matrix_Pss0}. The subsystem level counterparts of the system level relationships formulated in Eq.~\ref{eq:matrix_subrates1_2} then becomes
\begin{equation}
\label{eq:F_k2}
\begin{aligned}
{F}_k(t,{\bf x}) & = F(t,{\bm x}) \, \mathcal{D}_k({\bf x}) =  Q^x(t,{\bm x}) \, \mathcal{X}_k(t)  = Q^\tau(t,{\bm x}) \, \hat{\mathcal{T}}_k(t,{\bf x})  \\
\hat{\mathcal{T}}_k(t,{\bf x}) & = \mathcal{R}^{-1}(t,{\bm x}) \, \mathcal{X}_k(t) .
\end{aligned}
\end{equation}
Eq.~\ref{eq:matrix_subrates1_2} or~\ref{eq:F_k2} implies that the distribution matrices are related as
\begin{equation}
\label{eq:F_k2D}
\begin{aligned}
Q^\tau(t,{\bm x}) = Q^x(t,{\bm x}) \, \mathcal{R}(t,{\bm x}) .
\end{aligned}
\end{equation}

The flow and subthroughflow matrices can be integrated as follows:
\begin{equation}
\label{eq:tau}
\begin{aligned}
\bm{z}(t,\bm{x}) + F(t,\bm{x}) \, \mathbf{1} & = \sum_{k=0}^n {\mathbf z}_k (t,{\bf x}) + F_k (t,{\mathbf x}) \, \mathbf{1} = \sum_{k=0}^n {\check{\bm \tau}}_k (t,{\mathbf x}) = \check{\bm \tau} (t,\bm{x})
\\ &= {\check{\bm \tau}}_0 (t,{\mathbf x}) + \check{T} (t,{\mathbf x}) \, \mathbf{1} , \\
\bm{y}(t,\bm{x}) + F^T(t,\bm{x}) \, \mathbf{1} & = \sum_{k=0}^n {\mathbf y}_k (t,{\bf x}) + F^T_k (t,{\mathbf x}) \, \mathbf{1} = \sum_{k=0}^n {\hat{\bm \tau}}_k (t,{\mathbf x}) = {\hat{\bm \tau}} (t,\bm{x}) \\ & = {\hat{\bm \tau}}_0 (t,{\mathbf x}) +  \hat{T} (t,{\mathbf x}) \, \mathbf{1} . \nonumber
\end{aligned}
\end{equation}

The governing equations for the subsystems can then be written in vector form as
\begin{equation}
\label{eq:subsystemsV2}
\begin{aligned}
\dot {\mathbf x}_k(t) = \mathbf{z}_k(t,{\bf x}) + A(t,\bm{x}) \, \mathbf{x}_k(t)
\end{aligned}
\end{equation}
with the initial conditions ${\mathbf x}_0(t_0) = \bm{x}_0$ and ${\mathbf x}_k(t_0) = \mathbf{0}$ for $k = 1,\ldots,n$. The governing equations for the decomposed system can similarly be expressed in matrix form using the matrix functions introduced above as follows:
\begin{equation}
\label{eq:model_mat}
\begin{aligned}
{\dot X}(t) & = T(t,{\mathbf x}) = \check{T}(t,{\mathbf x}) - {\hat{T}}(t,{\mathbf x}) ,  \quad \, \, \, \, X (t_0) = \mathbf{0} , \\
\dot {\mathbf x}_0(t) & = \tau_0(t,{\mathbf x}) =  {\check{\tau}}_{0}(t,{\mathbf x}) - {\hat{\tau}}_{0}(t,{\mathbf x}) , \quad {\mathbf x}_0(t_0) = \bm{x}_0 .
\end{aligned}
\end{equation}
This system can also be expressed in terms of the flow intensity matrix, $A(t,\bm{x})$:\begin{equation}
\label{eq:model_M1}
\begin{aligned}
\dot{X}(t) &
= \mathcal{Z} (t,\bm{x}) + A(t,\bm{x}) \, X(t), \quad X(t_0) = \mathbf{0} , \\
\dot {\mathbf x}_0(t) & =  A(t,\bm{x}) \, \mathbf{x}_0(t) , \,  \quad \quad \quad \quad \quad {\mathbf x}_0(t_0) = \bm{x}_0 .
\end{aligned}
\end{equation}

The system decomposition methodology that yields the governing equations for each subsystem in vector form, Eq.~\ref{eq:subsystemsV2}, or for the entire system in matrix form, Eq.~\ref{eq:model_mat}, can be schematized as follows:
\begin{center}
\begin{tikzpicture}
  \matrix (m) [matrix of math nodes,row sep=0em,column sep=0em,minimum width=1em]
  {
   {} & \hspace{-3cm} { \dot{\bm x}(t) } = {{\bm \tau}(t,\bm{x})} & {} &  & \hspace{-1.6cm} \dot{\bf x}_k(t) = {\bm \tau}_k(t,{\mathbf x}), \, \, \,  k=0,\ldots,n
   \\
 {} &  \hspace{-3cm} \veq & {} & {} &  \hspace{-1cm} \veq
   \\
   {}
   &
 \hspace{-3cm}
 {
 \begin{bmatrix}
  \dot{x}_{1} \\
  \dot{x}_{2} \\
  \vdots  \\
  \dot{x}_{n}
 \end{bmatrix}
 }
 =
 { \begin{bmatrix}
  {\tau}_{1} \\
  {\tau}_{2} \\
  \vdots  \\
  {\tau}_{n}
 \end{bmatrix}
 }
 &
{}
&
{}
 &
  \hspace{-1cm}
 {
 \begin{bmatrix}
  \dot{x}_{1_k} \\
  \dot{x}_{2_k} \\
  \vdots  \\
  \dot{x}_{n_k}
 \end{bmatrix}
 }
=
{ \begin{bmatrix}
  {\tau}_{1_k} \\
  {\tau}_{2_k} \\
  \vdots  \\
  {\tau}_{n_k}
 \end{bmatrix}
 }

\\
   {} & {} & {} & {} & {}
\\
   {}
   &
 {
 \begin{bmatrix}
  \dot{x}_{1_1} & \cdots & \dot{x}_{1_n} \\
  \dot{x}_{2_1} & \cdots & \dot{x}_{2_n} \\
  \vdots  & \ddots & \vdots  \\
  \dot{x}_{n_1} & \cdots & \dot{x}_{n_n}
 \end{bmatrix}
 }
=
{
 \begin{bmatrix}
  {\tau}_{1_1} & \cdots & {\tau}_{1_n} \\
  {\tau}_{2_1} & \cdots & {\tau}_{2_n} \\
  \vdots & \ddots & \vdots  \\
  {\tau}_{n_1} & \cdots & {\tau}_{n_n}
 \end{bmatrix}
 }
&
{}
&
\hspace{-.5cm} \mbox{and} \hspace{-.4cm}
&
{
 \begin{bmatrix}
  \dot{x}_{1_0} \\
  \dot{x}_{2_0} \\
  \vdots  \\
  \dot{x}_{n_0}
 \end{bmatrix}
 }
 =
{ \begin{bmatrix}
  {\tau}_{1_0} \\
  {\tau}_{2_0} \\
  \vdots  \\
  {\tau}_{n_0}
 \end{bmatrix}
}
 \\
    { } &  {\veq} & {}  & {} & \veq  \\
   {} &  {  \dot{X}(t) = T(t,{\mathbf x}) } & {}& {} & \dot{\bf x}_0(t) = {\bm \tau}_0(t,{\mathbf x}) \\
};
  \path[-stealth, decorate]
    (m-3-2.east|-m-3-4) edge node [below] {system decomposition}
            node [above] {vector form} (m-3-4) ;
\hspace{-.4cm}
  \path[-stealth, decorate]
    (m-3-1) edge node [above,rotate=270] {matrix form}
            node [below,rotate=270] {system decomposition} (m-5-1) ;
\end{tikzpicture}
\end{center}

\subsubsection{Decomposed system}
\label{sec:ds}

The proposed methodology decomposes the original system into mutually exclusive and exhaustive subsystems by formulating the governing equations for each subsystem separately, as detailed in the previous sections. Consequently, for each compartment $i=1,\ldots,n$, there are $n$ governing equations for $n$ substates, $x_{i_k}(t)$, and $n$ equations for $n$ initial substates, ${\ubar x}_{i_k}(t)$. Therefore, there are $2n^2$ equations in total; $n^2$ of them are for the initial substates and the other $n^2$ equations are for the substates.

The governing equations for $x_{i_k}(t)$ and ${\ubar x}_{i_k}(t)$ can be written componentwise as
\begin{equation}
\label{eq:model_d1}
\begin{aligned}
\dot{x}_{i_k}(t) & = \left( z_{i_k}(t,{\mathbf x}) + \sum_{\substack{j=1}}^n f_{{i_k}{j_k}}(t,{\mathbf x})  \right) - \left( y_{i_k}(t,{\mathbf x}) + \sum_{\substack{j=1}}^n f_{{j_k}{i_k}}(t,{\mathbf x}) \right) \\
\end{aligned}
\end{equation}
with the initial conditions of $x_{i_k}(t_0) = 0$, and
\begin{equation}
\label{eq:model_d2}
\begin{aligned}
\dot{\ubar x}_{i_k}(t) & = \left( \sum_{\substack{j=1}}^n {\ubar f}_{i_k j_k}(t,{\mathbf x})  \right) - \left( {\ubar y}_{i_k}(t,{\mathbf x}) + \sum_{\substack{j=1}}^n {\ubar f}_{j_k i_k}(t,{\mathbf x}) \right)
\end{aligned}
\end{equation}
with the initial conditions of ${\ubar x}_{i_k} (t_0) = \delta_{ik} \, x_{i,0}$, for $i,k=1,\ldots,n$.

The governing equations for the decomposed system, Eqs.~\ref{eq:model_d1} and~\ref{eq:model_d2}, are already expressed in vector forms in Eqs.~\ref{eq:subsystemsV2} and~\ref{eq:initial_subsystemsV} as follows:
\begin{equation}
\label{eq:model_V}
\begin{aligned}
\dot {\mathbf x}_k(t) & = \mathbf{z}_k(t,{\bf x}) + A(t,\bm{x}) \, \mathbf{x}_k(t) ,
\quad {\mathbf x}_k(t_0) = \mathbf{0} \\
\dot{\ubar{\mathbf x}}_k(t) & = A(t,\bm{x}) \, \ubar{\mathbf{x}}_k(t) ,
\quad  \quad \quad \quad \quad \, \, \, \ubar{\mathbf x}_k(t_0) = x_{k,0} \, \mathbf{e}_k
\end{aligned}
\end{equation}
for $k = 1,\ldots,n$. Summing up the governing equations over $k$ separately for both the subsystems and initial subsystems formulated in Eq.~\ref{eq:model_V} yields
\begin{equation}
\label{eq:model_Vsum}
\begin{aligned}
\dot{\bar{\bm{x}}}(t) & = \bm{z}(t,\bm{x}) + A(t,\bm{x}) \, \bar{\bm{x}}(t) ,
\quad \, \, \bar{\bm{x}}(t_0) = \mathbf{0} , \\
\dot{\ubar{\bm x}}(t) & = A(t,\bm{x}) \, \ubar{\bm{x}}(t) ,
\quad  \quad \quad \quad \quad \, \, \, \ubar{\bm x}(t_0) = \bm{x}_{0} .
\end{aligned}
\end{equation}
This system enables the analysis of the evolution of external inputs and initial conditions within the system separately. Adding these two equations side by side gives back the original system, Eq.~\ref{eq:model_orgM2}, in the following form:
\begin{equation}
\label{eq:model_orgV}
\begin{aligned}
\dot {\bm x}(t) & = \bm{z}(t,\bm{x}) + A(t,\bm{x}) \, \bm{x}(t) , \quad {\bm x}(t_0) = \bm{x}_{0},
\end{aligned}
\end{equation}
as ${\sf F}(t,\bm{x}) \, \bm{1} = A(t,\bm{x}) \, \bm{x}(t) $. The partitioning in Eq.~\ref{eq:model_Vsum} could directly be obtained from the original system by defining a decomposition with two subsystems\textemdash one for external inputs and the other for the initial conditions.

The governing equations, Eqs.~\ref{eq:model_d1} and~\ref{eq:model_d2}, for the decomposed system are already expressed in matrix form in Eqs.~\ref{eq:model_M1} and~\ref{eq:model_M2} as follows:
\begin{equation}
\label{eq:model_M}
\begin{aligned}
\dot{X}(t) & = \mathcal{Z} (t,\bm{x}) + A(t,\bm{x}) \, X(t) ,
\quad X(t_0) = \mathbf{0} , \\
\dot{\ubar X}(t) & = A(t,\bm{x}) \, \ubar{X}(t) ,
\quad \quad \quad \quad \quad \, \, {\ubar X}(t_0) = \mathcal{X}_0  .
\end{aligned}
\end{equation}
Let a new matrix ${\sf X}(t)$ be defined, componentwise, as ${\sf X}_{ik} (t) \coloneqq \ubar{x}_{i_k}(t) + x_{i_k}(t)$. The governing equation for ${\sf X}(t) = \ubar{X}(t) + X(t)$ then becomes
\begin{equation}
\label{eq:model_MC}
\begin{aligned}
\dot{\sf X}(t) & = \mathcal{Z} (t,\bm{x}) + A(t,\bm{x}) \, {\sf X} (t) ,
\quad {\sf X}(t_0) = \mathcal{X}_0 .
\end{aligned}
\end{equation}
Note that ${\sf X}_{ik} (t)$ represents composite storage in compartment $i$ that is derived from the initial stock in compartment $k$, $x_{k,0}$, and the external input into compartment $k$, $z_k(t)$, during $[t_0,t]$.

\subsection{Fundamental matrix solutions}
\label{sec:nfm}

In this section, the {\em fundamental matrix solutions} will be introduced for nonlinear dynamic compartmental systems. They will be called the fundamental matrix solutions due to the common properties outlined in Thm.~\ref{thm:fund_M} with the fundamental matrix solutions to systems of linear ordinary differential equations. We will first show the existence and uniqueness of the decomposed system in vector form.
\begin{theorem}
\label{thm:eu_dcmp}
Let $(t_0,{\mathbf x}_0) \in \mathcal{I} \times \mho $. There exists a positive integer $r>0$ and an interval $\mathcal{I}' = (t_0 -r,t_0+r) \subset \mathcal{I}$ such that the governing equations for the decomposed system, Eq.~\ref{eq:model_V}, has a unique solution passing through $(t_0,{\mathbf x}_0)$ on $\mathcal{I}'$.
\end{theorem}
\begin{proof}
The net subthroughflow functions ${\tau}_{i_k} (t,{\mathbf x})$ and $\ubar{\tau}_{i_k} (t,{\mathbf x})$ on the right hand side of the decomposed system, Eq.~\ref{eq:model_V}, are continuous and continuously differentiable in ${\mathbf x}$ on $\mathcal{I} \times \mho$, as shown in Proposition~\ref{thm:exist_dcmp}. The existence and uniqueness of the solutions to Eq.~\ref{eq:model_V} on $\mathcal{I}'$ is an immediate consequence of Picard's local existence and uniqueness theorem.
\end{proof}

The definitions of fundamental matrix solutions and their main properties are outlined in the following theorem.
\begin{theorem}
\label{thm:fund_M}
Let $X(t)$ and $\ubar X(t)$ be the matrix functions defined in Eqs.~\ref{eq:X} and~\ref{eq:X0}.
\begin{enumerate}
\item $X(t)$ and $\ubar X(t)$ are the unique matrix solutions to the decomposed, nonlinear system, Eq.~\ref{eq:model_M}. They will be called the fundamental matrix solutions of the system.
\item For any given $(t_0,\bm{x}_0) \in \mathcal{I} \times \Omega$, the unique solution to the original system, Eq.~\ref{eq:model_orgV}, is given by
\begin{equation}
\label{eq:gs}
{\bm x}(t) = \ubar{X}(t) \, \bm{1} + X(t) \, \bm{1} .
\nonumber
\end{equation}
That is, $\bm{x}(t)$ is the linear combination of the columns of $\ubar X(t)$ and $X(t)$, where all combination coefficients are $1$.
\item Let $x_{i,0}>0$ and $z_{i}(t,\bm{x})>0$, $t \in \mathcal{I}$, for all $i$. The column vectors of $\ubar X(t)$ and $X(t)$ are linearly independent vectors in $\mathbb{R}^n$. Therefore, both $\ubar X(t)$ and $X(t)$ are invertible matrices at any time $t \in \mathcal{I}$ under given conditions.
\end{enumerate}
\end{theorem}
\begin{proof}
The proof for each item of the theorem is given below.
\begin{enumerate}
\item The existence and uniqueness of the solution to the decomposed system in vector form, Eq.~\ref{eq:model_V}, is shown in Thm.~\ref{thm:eu_dcmp}. The existence and uniqueness of the system in matrix form, Eq.~\ref{eq:model_M}, follows those of the system Eq.~\ref{eq:model_V}, by a column-wise comparison of both sides of the matrix equation, Eq.~\ref{eq:model_M}.

\item By the principle of nonlinear decomposition stated in Thm.~\ref{thm:ndp},
\begin{equation}
\label{eq:fb}
{\bm x}(t) = {\sf X} (t) \, \bm{1} = \ubar{X}(t) \, \bm{1} + X(t) \, \bm{1} = \sum_{k=1}^n  \ubar{\mathbf{x}}_{k}(t) + \sum_{k=1}^n \mathbf{x}_{k}(t)
\end{equation}
is a solution of the original system Eq.~\ref{eq:model_orgV}. The uniqueness of this solution follows the uniqueness of the decomposed system Eq.~\ref{eq:model_V} as shown in part (1) of this theorem.

\item Let $\ubar{\mathbf{x}}_i(t)$ and $\mathbf{x}_i(t)$ be the solutions of decomposed system Eq.~\ref{eq:model_V}. We would like to show that, for each fixed $t$, the set of vectors $\{ \ubar{\mathbf{x}}_1(t), \ldots, \ubar{\mathbf{x}}_n(t) \}$ and $\{\mathbf{x}_1(t), \ldots, \mathbf{x}_n(t) \}$ in ${\mathbb R^n }$ are linearly independent.

\quad We will first show that $\{\mathbf{x}_1(t), \ldots, \mathbf{x}_n(t) \}$ is linearly independent set in ${\mathbb R^n }$ for each fixed $t \in \mathcal{I}$. Suppose that, there exists a $t_1 \in \mathcal{I}$ such that the column vectors in  $\{ {\mathbf{x}}_1(t_1), \ldots, {\mathbf{x}}_n(t_1) \}$ are linearly dependent. There exists then a combination constants ${c}_1,\ldots,{c}_n$, not all zero, such that
\begin{equation}
\label{eq:proof1}
\begin{aligned}
\bm{0} & = {c}_1 \, {\mathbf{x}}_1(t_1) + \ldots + {c}_n \, {\mathbf{x}}_n(t_1) = X(t_1) \, \bm{c}
\end{aligned}
\end{equation}
where $\bm{c} = (c_1,\ldots,c_n)^T \in \mathbb{R}^n$.
Let
\begin{equation}
\label{eq:proof3}
\begin{aligned}
{\bm{\alpha}}(t) & \coloneqq {c}_1 \, {\mathbf{x}}_1(t) + \ldots + {c}_n \, {\mathbf{x}}_n(t) = {X}(t) \, {\bm{c}}, \quad t \in \mathcal{I}.
\end{aligned}
\end{equation}
Therefore, $\dot {\bm{\alpha}}(t) = \dot{X}(t) \, \bm{c}$. Due to the governing equation for $X(t)$, Eq.~\ref{eq:model_M}, and Eq.~\ref{eq:proof3}, we have
\begin{equation}
\label{eq:proof5}
\begin{aligned}
\dot {\bm{\alpha}}(t)  =  \mathcal{Z} (t,\bm{x}) \, \bm{c} + A(t,\bm{x}) \, \bm{\alpha}(t), \quad \bm{\alpha}(t_0) = \bm{0}.
\end{aligned}
\end{equation}
Equation~\ref{eq:proof1} implies that $\bm{\alpha}(t_1) = \bm{0}$. Without loss of generality, assume that $c_i >0$ for some $i$. We then have
\begin{equation}
\label{eq:proof_1}
\begin{aligned}
{\alpha}_i(t_0) & = 0, \quad \dot{\alpha}_i(t_0) = c_i \, {z_i}(t_0,\bm{x}) >0, \quad \mbox{and} \\ {\alpha}_i(t_1) & = {0}, \quad \dot{\alpha}_i(t_1) = c_i \, {z_i}(t_1,\bm{x}) >0,
\end{aligned}
\end{equation}
as $z_i(t,\bm{x}) >0$, $\forall i$. Since $\bm{\alpha}(t)$ is a differentiable and, therefore, is a continuous function, Eq.~\ref{eq:proof_1} implies that there exists at least one $t^* \in (t_0,t_1)$ such that ${\alpha}_i(t^*) = 0$ and $\dot{\alpha}_i(t^*) <0$. Due to Eq.~\ref{eq:proof5}, this result implies that $z_i(t^*,\bm{x})<0$. This contradiction completes the first part of the proof.

\quad Now, we would like to show that $\{ \ubar{\mathbf{x}}_1(t), \ldots, \ubar{\mathbf{x}}_n(t) \}$ is a linearly independent set in ${\mathbb R^n }$ for each fixed $t \in \mathcal{I}$. The condition for this case is that $x_{i,0} > 0$, $\forall i$. Suppose now that at the same $t_1 \in \mathcal{I}$, the column vectors in $\{ \ubar{\mathbf{x}}_1(t_1), \ldots, \ubar{\mathbf{x}}_n(t_1) \}$ are linearly dependent in $\Omega$. There then exists constants $\ubar{c}_1,\ldots,\ubar{c}_n$, not all zero, such that
\begin{equation}
\label{eq:proof4}
\begin{aligned}
\bm{0} & = \ubar{c}_1 \, \ubar{\mathbf{x}}_1(t_1) + \ldots + \ubar{c}_n \, \ubar{\mathbf{x}}_n(t_1) = \ubar{X}(t_1) \, \ubar{\bm{c}}
\end{aligned}
\end{equation}
where $\ubar{\bm{c}} = (\ubar{c}_1,\ldots,\ubar{c}_n)^T \in \mathbb{R}^n$. Let
\begin{equation}
\label{eq:proof2}
\begin{aligned}
\ubar{\bm{\alpha}}(t) & \coloneqq \ubar{c}_1 \, \ubar{\mathbf{x}}_1(t) + \ldots + \ubar{c}_n \, \ubar{\mathbf{x}}_n(t) = \ubar{X}(t) \, \ubar{\bm{c}}, \quad t \in \mathcal{I}.
\nonumber
\end{aligned}
\end{equation}
This implies that $\dot {\ubar{\bm{\alpha}}}(t) = \dot{\ubar{X}}(t) \, \ubar{\bm{c}}$. Due to the governing equation for $\ubar{X}$, Eq.~\ref{eq:model_M}, and Eq.~\ref{eq:proof4}, we have
\[ \dot {\ubar{\bm{\alpha}}}(t)  = A(t,\bm{x}) \, \ubar{\bm{\alpha}}(t) , \quad \ubar{\bm{\alpha}}(t_1) = \bm{0} .\]
Due to the uniqueness theorem, Thm.~\ref{thm:eu_dcmp}, $\ubar{\bm{\alpha}}(t) = \bm{0}$, $t \in \mathcal{I}$. In particular,
\[ \ubar{\bm{\alpha}}(t_0) = {\ubar{X}}(t_0) \, \ubar{\bm{c}} = \mathcal{X}_0 \, \ubar{\bm{c}} = \bm{0} \quad \Rightarrow \quad \ubar{\bm{c}} = \bm{0} , \]
as $x_{i_0} > 0$, $\forall i$, which is a contradiction.

\quad These two contradictions for each part of the decomposed system, Eq.~\ref{eq:model_M}, indicate that neither the set of the column vectors of $\ubar{X}(t)$ nor that of ${X}(t)$ can be linearly dependent at any time $t_1 \in \mathcal{I}$. Therefore, the column vectors of $\ubar{X}(t)$ and ${X}(t)$ form linearly independent sets, and, consequently, the matrices are invertible for all $t \in \mathcal{I}$.
\end{enumerate}
\end{proof}

\subsection{Nonlinear decomposition principle}
\label{sec:ndp}

We will state the {\em nonlinear decomposition principle} for dynamic compartmental systems in the following theorem. It essentially asserts that the solution for each subsystem also solves the original system, so is any arbitrary combination of these solutions as specified in the theorem.

\begin{theorem}
\label{thm:ndp}
Let $\mathbf{x}_k(t)$ and $\ubar{\mathbf{x}}_k(t)$ be the respective solutions on $\mho$ to the $k^{th}$ subsystem and initial subsystem of the decomposed system with the corresponding external inputs and initial conditions specified in Eq.~\ref{eq:model_V}. The following combination of the vector functions
\begin{equation}
\label{eq:gs2}
{\bm x}(t) = \sum_{k=1}^n \alpha_k \, {\mathbf{x}}_{k}(t) + \beta_k \, \ubar{\mathbf{x}}_{k}(t), \quad \alpha_k,\beta_k \in \{0,1\} ,
\end{equation}
then is a solution to the original system on $\Omega$ with the following external inputs and initial conditions:
\begin{equation}
\label{eq:gsic}
\begin{aligned}
{\bm z}(t,\bm{x}) &= \sum_{k=1}^n \alpha_k \, \mathbf{z}_k(t,{\bf x}) = \sum_{k=1}^n \alpha_k \, z_k(t,\bm{x}) \, \mathbf{e}_k  \quad \mbox{and} \\
 \bm{x}_0 &= {\bm x}(t_0) = \sum_{k=1}^n \beta_k \, \ubar{\mathbf{x}}_{k}(t_0) =  \sum_{k=1}^n \beta_k \, x_{k,0} \, \mathbf{e}_k .
\end{aligned}
\end{equation}
\end{theorem}
\begin{proof}
\label{prf:ndp}
Note that, if $\alpha_k=0$ or $\beta_k=0$ for some $k$, the corresponding solutions to Eq.~\ref{eq:model_V}, with the conditions given in Eq.~\ref{eq:gsic}, become ${\mathbf{x}}_{k}(t)=\bm{0}$ or $\ubar{\mathbf{x}}_{k}(t)=\bm{0}$, respectively. This is because of the fact that the subsystems are driven either by external inputs or initial stocks. Therefore, if there is no external input (source term) for a subsystem or initial stock (initial condition) for an initial subsystem, the corresponding subsystem becomes null\textemdash the substate variables and the corresponding subflows rates for that subsystem or initial subsystem become zero.

Multiplying both sides of governing system Eq.~\ref{eq:model_MC} by $\mathbf{1}$ yields the original system in the form of Eq.~\ref{eq:model_orgV}, because of the relationship in Eq.~\ref{eq:fb}. Consequently, if $\mathbf{x}_k(t)$ and $\ubar{\mathbf{x}}_k(t)$ are the respective solutions to the $k^{th}$ subsystem and initial subsystem of the decomposed system, Eq.~\ref{eq:model_V}, on $\mho$, $\bm{x}(t)$ is the solution to the original system, Eq.~\ref{eq:model_orgM1}, on $\Omega$.
\end{proof}

This nonlinear decomposition principle corresponds to the superposition principle for linear ordinary differential equations. It is in the sense that, the solution to a nonlinear compartmental system can be decomposed into {\em subsolutions}, each of which, as well as any arbitrary combination of them as specified in Thm.~\ref{thm:ndp}, solves the original system.

\subsection{Subsystem decomposition methodology}
\label{sec:dsd}

We will introduce the {\em dynamic subsystem decomposition} methodology in this section for further partitioning or segmentation of subsystems along a given set of mutually exclusive and exhaustive subflow paths. The subsystem decomposition methodology dynamically decomposes arbitrary composite intercompartmental flows and storages into the constituent transient subflow and substorage segments along given subflow paths. Therefore, the subsystem decomposition determines the distribution of arbitrary intercompartmental flows and the organization of the associated storages generated by these flows within the subsystems. In other words, the subsystem partitioning enables tracking the evolution of arbitrary intercompartmental flows and storages within and monitoring their spread throughout the system.

The dynamic subsystem decomposition methodology will be formulated below using the {\em directed subflow path} terminology introduced in Appendix~\ref{apxsec:dsd}. The {\em natural subsystem decomposition} of each subsystem then yields a mutually exclusive and exhaustive decomposition of the entire system. We will first introduce the transient flows and storages. They will be used in the formulation of the \texttt{diact} flows and storages in the next section.

\subsubsection{Transient flows and storages}
\label{apxsec:transient_flows}

Along a given subflow path $p^w_{n_k j_k}= i_k \mapsto j_k \to \ell_k \to n_k$, the {\em transient inflow} at subcompartment ${\ell_k}$, $f^{w}_{\ell_k j_k i_k}(t)$, generated by the local input from $i_k$ to ${j_k}$ during $[t_1,t]$, $t_1 \geq t_0$, is the input segment that is transmitted from $j_k$ to ${\ell_k}$ at time $t$. Similarly, the {\em transient outflow} generated by the transient inflow at ${\ell_k}$ during $[t_1,t]$, $f^{w}_{n_k \ell_k j_k}(t)$, is the inflow segment that is transmitted from ${\ell_k}$ to the next subcompartment, ${n_k}$, along the path at time $t$. The associated {\em transient substorage} in subcompartment ${\ell_k}$ at time $t$, $x^{w}_{n_k \ell_k j_k}(t)$, is then the substorage segment governed by the transient inflow and outflow balance during $[t_1,t]$  (see Fig.~\ref{fig:subsystemp}).
\begin{figure}[t]
\begin{center}
\begin{tikzpicture}
   \draw[very thick, fill=gray!5, draw=black] (-.05,-.05) rectangle node(R1) [pos=.5] { } (2.1,2.1) ;
   \draw [very thick, fill=blue!10, draw=blue] plot [smooth] coordinates {(-0.05,1) (.5,.5) (1,.8) (1.5,1) (2.1,1.5)};
   \draw [fill=blue!10, draw=none]  (-0.05,1) -- (2.1,1.5) -- (2.1,2.1) -- (-.05,2.1) ;
   \draw[very thick, draw=black, text=blue] (-.05,-.05) rectangle node(R1) [pos=.5, yshift=.6cm] { $x^w_{n_k \ell_k j_k}$ } (2.1,2.1) ;
   \draw [thick, fill=gray!5, draw=black]  (-2,1.2) -- (-0.35,1.2) -- (-0.2,1.4) -- (-.35,1.6) -- (-2,1.6) ;
   \draw [thick, fill=gray!5, draw=black]  (2.3,.9) -- (4,.9) -- (4.3,1.7) -- (4,2.1) -- (2.3,2.1) ;
   \draw [thick, fill=gray!5, draw=black]  (2.3,1.6) -- (4,1.6) -- (4.15,1.8) -- (4,2) -- (2.3,2) ;
    \node (x) at (.9,2.5) {${x}_{\ell_k}$};
    \node (x) at (-1.1,.8) {$f_{\ell_k j_k}$};
   \draw [very thick, -stealth, draw=blue]  (-2,1.4) -- (-.25,1.4) ;
    \node [text=blue] (x) at (-2.7,1.4) {$f^w_{\ell_k j_k i_k}$};
   \draw [very thick, -stealth, draw=blue]  (2.3,1.8) -- (4.1,1.8) ;
    \node [text=blue] (x) at (5,1.9) {$f^w_{n_k \ell_k j_k}$};
    \node (x) at (3.2,1.3) {$f_{n_k \ell_k}$};
    \node (x) at (3.2,.5) {${\hat{\tau}}_{\ell_k}$};
\end{tikzpicture}
\end{center}
\caption{Schematic representation of the dynamic subsystem decomposition. The transient inflow and outflow rate functions, $f^w_{\ell_k j_k i_k}(t)$ and $f^w_{n_k \ell_k j_k}(t)$, at and associated transient substorage, $x^w_{n_k \ell_k j_k}(t)$, in subcompartment ${\ell_k}$ along subflow path $p^w_{n_k j_k}= i_k \mapsto j_k \to \ell_k \to n_k$. }
\label{fig:subsystemp}
\end{figure}
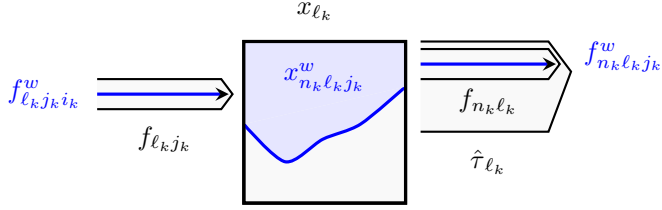

The transient outflow at subcompartment ${\ell_k}$ at time $t$ along subflow path $p^w_{n_k j_k}$ from ${j_k}$ to ${n_k}$, $f^{w}_{n_k \ell_k j_k}(t)$, can be formulated as follows:
\begin{equation}
\label{eq:out_in_fs}
\begin{aligned}
f^{w}_{n_k \ell_k j_k}(t) = \frac{ f_{n_k \ell_k}(t,\bf{x}) }{x_{\ell_k}(t) } \, x^{w}_{n_k \ell_k j_k}(t) ,
\end{aligned}
\end{equation}
similar to Eq.~\ref{eq:cons}, due to the equivalence of flow and subflow intensities, where the transient substorage, $x^{w}_{n_k \ell_k j_k}(t)$, is determined by the governing mass balance equation
\begin{equation}
\label{eq:out_in_fs2}
\begin{aligned}
\dot{x}^{w}_{n_k \ell_k j_k}(t) & = f^{w}_{\ell_k j_k i_k}(t) - \frac{ {\hat{\tau}}_{\ell_k}(t,\bf{x}) }{ x_{\ell_k}(t) } \, {x}^{w}_{n_k \ell_k j_k}(t) , \quad {x}^{w}_{n_k \ell_k j_k}(t_1) = 0 .
\end{aligned}
\end{equation}
The equivalence of the throughflow and subthroughflow intensities, as well as the flow and subflow intensities in the same direction, that is
\[ q^x_{n \ell}(t,{\bm x}) = \frac{  f_{n \ell}(t,{\bm x}) }{x_{\ell}(t) } = \frac{ f_{n_k \ell_k}(t,{\bf x})  }{x_{\ell_k}(t) }
 \quad \mbox{and} \quad r^{-1}_\ell (t,{\bm x}) = \frac{ {\hat{\tau}}_{\ell}(t,\bm{x}) }{ x_{\ell}(t) } = \frac{ {\hat{\tau}}_{\ell_k}(t,\bf{x}) }{ x_{\ell_k}(t) }   \]
are given by Eqs.~\ref{eq:csc_dense} and~\ref{eq:thr_dense}, for $\ell,n=1,\ldots,n$, and $k=0,\ldots,n$, where the denominators are nonzero. Therefore, since the intensities in Eqs.~\ref{eq:out_in_fs} and~\ref{eq:out_in_fs2} can be expressed at both the compartmental and subcompartmental levels, the subsystem decomposition is actually independent from the system decomposition. That is, the same analysis can be done along flow paths within the system, instead of subflow paths within the subsystems. This allows the flexibility of tracking arbitrary intercompartmental flows and storages generated by all or individual environmental inputs within the system. The governing equations, Eqs.~\ref{eq:out_in_fs} and~\ref{eq:out_in_fs2}, establish the foundation of the {\em dynamic subsystem decomposition}. These equations for each subcompartment along a given flow path of interest will then be coupled with the decomposed system, Eqs.~\ref{eq:model_d1} and~\ref{eq:model_d2}, or the original system, Eq.~\ref{eq:model2}, and be solved simultaneously. The equations can alternatively be solved individually and separately once the original or decomposed system is solved.

Additional relationships for the transient subflows and substorages along closed subflow paths are formulated in Appendix~\ref{apxsec:transient}. The transient subflows and substorages within the initial subsystems can be defined and formulated similarly.

\subsubsection{The \texttt{diact} flows and storages}
\label{apxsec:flows}

Five main transaction types are introduced in this section based on the subsystem decomposition methodology: the $\texttt{diact}$ flows and storages. The {\em transfer flows} (denoted by \texttt{t}) and storages will be formulated in detail below, and parallel derivations for {\em direct} (\texttt{d}), {\em indirect} (\texttt{i}), {\em cycling} (\texttt{c}), and {\em acyclic} (\texttt{a}) {\em flows} and {\em storages} can be found in Appendix~\ref{appsec:flows}.

The {\em composite transfer flow} will be defined as the total intercompartmental transient flow that is generated by all external inputs from one compartment, {\em directly} or {\em indirectly} through other compartments, to another. The {\em composite direct}, {\em indirect}, {\em acyclic}, and {\em cycling flows} from the initial compartment to the terminal compartment are then defined as the direct, indirect, non-cycling, and cycling segments at the terminal compartment of the composite transfer flow (see Fig.~\ref{fig:utilityfigs}).

The {\em simple transfer flow} will be defined as the total intercompartmental transient subflow that is generated by the single external input from an input-receiving subcompartment, {\em directly} or {\em indirectly} through other compartments, to another subcompartment. The {\em simple direct}, {\em indirect}, {\em acyclic}, and {\em cycling flows} from the initial input-receiving subcompartment to the terminal subcompartment are then defined as the direct, indirect, non-cycling, and cycling segments at the terminal subcompartment of the simple transfer flow (see Fig.~\ref{fig:utilityfigs}).

The associated simple and composite \texttt{diact} storages are defined as the storages generated by the corresponding \texttt{diact} flows. The simple and composite \texttt{diact} {\em flows} and {\em storages} at both the subcompartmental and compartmental levels are formulated below and in Appendix~\ref{appsec:flows}.
\begin{figure}[t]
\begin{center}
\begin{tikzpicture}
   \draw[very thick, fill=gray!10, draw=black] (-.05,-.05) rectangle node(R1) [pos=.5] { } (2.1,2.4) ;
   \draw[very thick, fill=blue!3, draw=blue, text=blue] (0.3,0.1) rectangle node(R1) [pos=.5] {$x_{i_i}$} (1.2,1.3) ;
    \node (x) at (1,-.5) {${x}_{i}$};
   \draw[very thick, fill=gray!10, draw=black] (8.95,-.05) rectangle node(R2) [pos=.5] { } (11.1,2.4) ;
   \draw[very thick, fill=gray!10, draw=black] (8.95,1.69) rectangle node(R2) [pos=.5] { } (10.5,2.3) ;
   \draw[very thick, fill=blue!3, draw=blue, text=blue] (9.05,0.05) rectangle node(R1) [pos=.5] {$x_{j_i}$} (10.05,1.1) ;
    \node (x) at (10,-.5) {${x}_{j}$};
    \node (x) at (9.7,2) {${x}_{j_0}$};
   \draw[fill=gray!10] (3.05,1.3) -- ++ (-0.6,0) -- ++ (-0.3,.3) -- ++ (.3,.3) -- ++ (.6,0);
   \draw[fill=gray!10] (8.1,1.3) -- ++ (0.84,0) -- ++ (0,.4) -- ++ (-.84,0) ;
   \draw[fill=gray!10] (8.1,1.7) -- ++ (0.84,0) -- ++ (0,.43) -- ++ (-.84,0) ;
    \node[] (t2) at (2.8,1) {${\tau}^\texttt{t}_{i j}$};
    \node[] (t3) at (8.5,2.4) {$\hat{\tau}_{j}$};
    \node[] (t3) at (8.5,1.92) {$\hat{\tau}_{j_0}$};
    \node[] (t2) at (4,2.2) {${\tau}^\texttt{c}_{i j}$};
    \node[] (x) at (4.5,1.1) {${\tau}^\texttt{d}_{i j} $};
    \draw[very thick,-stealth]  (8.8,1.4) -- (2.4,1.4);
   \draw[fill=blue!10] (2.95,0.3) -- ++ (-0.84,0) -- ++ (0,.4) -- ++ (.84,0) ;
    \draw[very thick,-stealth,draw=blue]  (-.8,.5) -- (.25,.5) ;
    \node (z) [text=blue] at (-.45,.8) {$z_i$};
    \node[blue] (t3) at (2.7,0.05) {$\hat{\tau}_{i_i}$};
   \draw[fill=blue!10] (8,0.1) -- ++ (0.6,0) -- ++ (0.3,.3) -- ++ (-.3,.3) -- ++ (-.6,0);
    \node[blue] (t1) at (8.4,1) {$\tilde{\tau}_{j_i}$};
    \draw[very thick,draw=blue, -stealth]  (1.3,0.6) -- (8.7,0.6) ;
    \node[blue] (x) at (6.7,0.9) {${\tau}^\texttt{d}_{j_i} $};
   \node[blue,anchor=west] at (3.3,0.1) (a) {${\tau}^\texttt{i}_{j_i}$};
   \node[blue,anchor=east] at (7,0.4) (b) {};
   \node[blue,anchor=west] at (.9,0.4) (e) {};
   \node[blue,anchor=east] at (3.3,0.1) (f) {};
\draw [very thick,draw=blue, dashed, -stealth] plot [smooth, tension=1] coordinates { (1.3,0.4) (4,0.4) (5,-0.1) (4.5,-0.6) (4,-0.1) (5,0.4) (8.7,0.4) };
\draw [very thick,draw=blue, dashed, -stealth] plot [smooth, tension=1] coordinates { (9,0.4) (9.6,-0.1) (9,-0.5) (7,-0.5) (6.5,-0.1) (7,0.2) (8.7,0.2) };
    \node[blue] (t2) at (7,-0.2) {${\tau}^\texttt{c}_{j_i}$};
   \node[anchor=east] at (7.2,2.1) (c) {$\tau^\texttt{i}_{i j}$};
   \node[blue,anchor=west] at (2.1,1.7) (d) {};
\draw [very thick,dashed, -stealth] plot [smooth, tension=1] coordinates { (8.8,1.6) (6.5,1.6)  (5.5,2.1)  (6,2.6) (6.5,2.1) (5.5,1.6) (2.4,1.6) };
\draw [very thick, dashed, -stealth] plot [smooth, tension=1] coordinates { (2,1.6) (1.5,2.1) (2,2.6) (4,2.6) (4.5,2.1) (4,1.8) (2.4,1.8) };
   \node[blue,anchor=east] at (5.9,2.1) (h) {};
   \node[blue,anchor=west] at (8,1.7) (g) {};
\end{tikzpicture}
\end{center}
\caption{Schematic representation of the simple and composite \texttt{diact} flows. Solid arrows represent direct flows, and dashed arrows represent indirect flows through other compartments (not shown).
The composite \texttt{diact} flows (black) generated by outward throughflow $\hat{\tau}_{j}(t,\bm{x}) - \hat{\tau}_{j_0}(t,{\bf x})$ (i.e. derived from all external inputs): direct flow, $\tau^\texttt{d}_{i j}(t)$, indirect flow, $\tau^\texttt{i}_{ij}(t)$, acyclic flow, $\tau^\texttt{a}_{ij}(t) = \tau^\texttt{t}_{ij}(t) - \tau^\texttt{c}_{ij}(t)$, cycling flow, $\tau^\texttt{c}_{ij}(t)$, and transfer flow, $\tau^\texttt{t}_{ij}(t)$.
The simple \texttt{diact} flows (blue) generated by outward subthroughflow $\hat{\tau}_{i_i}(t,{\bf x})$ (i.e.  derived from single external input $z_{i}(t)$): direct flow, ${\tau}^\texttt{d}_{j_i}(t) = {\tau}^\texttt{d}_{j_i i_i}(t) $, indirect flow, ${\tau}^\texttt{i}_{j_i}(t) = \tau^\texttt{i}_{j_i i_i}(t)$, acyclic flow, ${\tau}^\texttt{a}_{j_i}(t) = \tau^\texttt{a}_{j_i i_i}(t) = {\tau}^\texttt{t}_{j_i}(t) - {\tau}^\texttt{c}_{j_i}(t)$, cycling flow, ${\tau}^\texttt{c}_{j_i}(t) = \tau^\texttt{c}_{j_i i_i}(t)$, and transfer flow, $\tau^{\texttt{t}}_{j_i}(t) = \tilde{\tau}_{j_i}(t,{\bf x}) = \check{\tau}_{j_i}(t,{\bf x}) - z_{j_i}(t)$. Note that the cycling flows at the terminal (sub)compartment may include the segments of the direct and/or indirect flows at that (sub)compartment, if the cycling flows indirectly pass through the corresponding initial (sub)compartment (see Fig.~\ref{fig:cycindices}). Therefore, the acyclic flows are composed of the segments of the direct and/or indirect flows.
}
\label{fig:utilityfigs}
\end{figure}

The {\em composite transfer subflow} will be defined as the total intercompartmental transient subflow from one subcompartment, directly or indirectly through other subcompartments, to another in the same subsystem. Let $P^\texttt{t}_{i_k j_k}$ be the set of mutually exclusive subflow paths $p^w_{i_k j_k}$ from subcompartment $j_k$ directly or indirectly to $i_k$ in subsystem $k$. The {\em composite transfer subflow} from subcompartment $j_k$ to $i_k$, $\tau^\texttt{t}_{i_k j_k}(t)$, can be expressed as the sum of the cumulative transient subflows, $\check{\tau}_{i_k}^{w}(t)$, generated by the outward subthroughflow at subcompartment $j_k$, $\hat{\tau}_{j_k}(t,{\bf x})$, during $[t_1,t]$, $t_1 \geq t_0$, and transmitted into $i_k$ at time $t$ along all subflow paths $p^w_{i_k j_k} \in P^\texttt{t}_{i_k j_k}$. The associated {\em composite transfer substorage}, $x^\texttt{t}_{i_k j_k}(t)$, at subcompartment $i_k$ at time $t$ is the sum of the cumulative transient substorages, $x^{w}_{i_k}(t)$, generated by the cumulative transient inflows, $\check{\tau}_{i_k}^{w}(t)$, during $[t_1,t]$. Alternatively, $x^\texttt{t}_{i_k j_k}(t)$ can be defined as the storage segment generated by the composite transfer subflow $\tau^\texttt{t}_{i_k j_k}(t)$ in subcompartment $i_k$ during $[t_1,t]$.

The composite transfer subflow and substorage will then be formulated as follows:
\begin{equation}
\label{eq:out_in_fsDT}
\begin{aligned}
{\tau}^\texttt{t}_{i_k j_k}(t) \coloneqq
\sum_{w=1}^{w_k}  \check{\tau}_{i_k}^{w}(t)
\quad \mbox{and} \quad
x^\texttt{t}_{i_k j_k}(t) \coloneqq \sum_{w=1}^{w_k} x^{w}_{i_k}(t)
\end{aligned}
\end{equation}
where $w_k$ is the number of subflow paths $p^w_{i_k j_k} \in P^\texttt{t}_{i_k j_k}$. The sum of all composite transfer subflows and associated substorages from subcompartment $j_k$ to $i_k$ within each subsystem $k \neq 0$ will be called the {\em composite transfer flow} and {\em storage} from compartment $j$ to $i$ at time $t$, ${\tau}^\texttt{t}_{i j}(t)$ and $x^\texttt{t}_{i j}(t)$, generated by all environmental inputs during $[t_1,t]$. They can be formulated as 
\begin{equation}
\label{eq:out_in_fsTk}
\begin{aligned}
\tau^\texttt{t}_{i j}(t) \coloneqq \sum_{k=1}^{n} \tau^\texttt{t}_{i_k j_k}(t) \quad & \mbox{and} \quad x^\texttt{t}_{i j}(t) \coloneqq  \sum_{k=1}^{n} x^\texttt{t}_{i_k j_k}(t) .
\end{aligned}
\end{equation}

For notational convenience, we define $n \times n$ matrix functions ${T}^{\texttt{t}}_k(t)$ and ${X}^{\texttt{t}}_k(t)$ whose $(i,j)-$elements are $\tau^\texttt{t}_{i_kj_k}(t)$ and $x^\texttt{t}_{i_kj_k}(t)$, respectively. That is,
\begin{equation}
\label{eq:indirectfmT}
\begin{aligned}
{T}^{\texttt{t}}_k(t) \coloneqq \left( \tau^\texttt{t}_{i_k j_k}(t) \right) \quad \mbox{and} \quad {X}^{\texttt{t}}_k(t) \coloneqq \left( x^\texttt{t}_{i_k j_k}(t) \right) ,
\end{aligned}
\end{equation}
for $k=0,\ldots,n$. These matrix measures ${T}^{\texttt{t}}_k(t)$ and ${X}^{\texttt{t}}_k(t)$ will be called the $k^{th}$ {\em composite transfer subflow} and associated {\em substorage matrix} functions. The corresponding {\em composite transfer flow} and associated {\em storage matrix} functions are then defined as ${T}^{\texttt{t}}(t) \coloneqq \left( \tau^\texttt{t}_{i j}(t) \right)$ and ${X}^{\texttt{t}}(t) \coloneqq \left( x^\texttt{t}_{i j}(t) \right)$, respectively.

The {\em simple transfer flows} and {\em storages} can be formulated in terms of their composite counterparts as follows:
\begin{equation}
\label{eq:simple_diact}
\begin{aligned}
{\tau}^\texttt{t}_{i_k}(t) = {\tau}^\texttt{t}_{i_k k_k}(t) \quad \mbox{and} \quad {x}^\texttt{t}_{i_k}(t) = {x}^\texttt{t}_{i_k k_k}(t) .
\end{aligned}
\end{equation}
To distinguish the composite and simple transfer flow and storage matrices, we use a tilde notation over the simple versions. That is, the simple transfer flow and storage matrices, for example, will be denoted by $\tilde{T}^{\texttt{t}}(t) \coloneqq \left( \tau^\texttt{t}_{i_k}(t) \right)$ and $\tilde{X}^{\texttt{t}}(t) \coloneqq \left( x^\texttt{t}_{i_k}(t) \right) $.

The difference between the composite and simple transfer flows, $\tau^\texttt{t}_{ik}(t)$ and $\tau^\texttt{t}_{i_k}(t)$, and associated storages, $x^\texttt{t}_{ik}(t)$ and $x^\texttt{t}_{i_k}(t)$, is that the composite transfer flow and storage from compartment $k$ to $i$ are generated by outward throughflow $\hat{\tau}_k(t,{\bm x}) - \hat{\tau}_{k_0}(t,{\bf x})$ derived from all external inputs, and their simple counterparts from input-receiving subcompartment $k_k$ to $i_k$ are generated by outward subthroughflow $\hat{\tau}_{k_k}(t,{\bf x})$ derived from single external input $z_k(t,\bm{x})$ (see Fig.~\ref{fig:utilityfigs}). In that sense, the composite and simple transfer flows and storages measure the influence of one compartment on another induced by all and a single external input, respectively.

All the other {\em simple} and {\em composite} $\texttt{diact}$ {\em flows} and {\em storages} can be formulated through the {\em path-based approach}, introduced above for the transfer flows and storages, based on the subsystem partitioning methodology, These derivations are presented in Appendix~\ref{appsec:flows}.

The simple and composite $\texttt{diact}$ subflow and substorage definitions can be extended in parallel to the initial subsystems as well. We will use bar notation below the \texttt{diact} flow and storage symbols for the initial subsystems. The difference between the composite and simple transfer flows, $\ubar{\tau}^\texttt{t}_{ik}(t)$ and $\ubar{\tau}^\texttt{t}_{i_k}(t)$, and associated storages, $\ubar{x}^\texttt{t}_{ik}(t)$ and $\ubar{x}^\texttt{t}_{i_k}(t)$, within the initial subsystems is that the composite transfer flow and storage from compartment $k$ to $i$ are generated by outward throughflow $\hat{\tau}_{k_0}(t,{\bf x}) = \ubar{\hat{\tau}}_{k}(t,{\bf x})$ derived from all initial stocks, and their simple counterparts from initial subcompartment $k_k$ to $i_k$ are generated by outward subthroughflow $\ubar{\hat{\tau}}_{k_k}(t,{\bf x})$ derived from single initial stock $x_{k,0}$. In that sense, the composite and simple transfer flows and storages within the initial subsystems measure the influence of one compartment on another induced by all and a single initial stock, respectively.

The simple and composite \texttt{diact} flows have been explicitly formulated for both dynamic and static systems in the context of ecosystem analysis in the recent studies by \cite{Coskun2017DCSAM,Coskun2017SCSA}. In addition to the path-based approach introduced above through the subsystem decomposition, these alternative approaches for explicit formulation of the \texttt{diact} flows and storages through the system decomposition methodology are called the {\em dynamic} and {\em static approaches}.

The composite \texttt{diact} subflows from subcompartment $k_\ell$ to $i_\ell$ at time $t$ are formulated componentwise through the dynamic approach as follows:
\begin{equation}
\label{eq:comp_diact_subs}
\begin{aligned}
 \tau^{\texttt{d}}_{i_\ell k_\ell}(t) & =  \frac{ f_{i_k k_k}(t,{\bf x}) } {{\hat{\tau}}_{k_k}(t,{\bf x}) } \, {{\hat{\tau}}_{k_\ell}(t,{\bf x}) } = \frac{ f_{i k}(t,{\bm x}) } {{\hat{\tau}}_{k}(t,{\bm x}) } \, {{\hat{\tau}}_{k_\ell}(t,{\bf x}) } \\
 \tau^{\texttt{i}}_{i_\ell k_\ell}(t) & = \frac { {\check{\tau}}_{i_k} (t,{\bf x}) - z_{i_k}(t,{\bf x}) - f_{i_k k_k}(t,{\bf x}) } { {\hat{\tau}}_{k_k}(t,{\bf x}) } \, { {\hat{\tau}}_{k_\ell}(t,{\bf x}) } \\
 \tau^{\texttt{a}}_{i_\ell k_\ell}(t) & = \left[ \frac{ {\check{\tau}}_{i_k} (t,{\bf x}) - z_{i_k}(t,{\bf x}) } {{\hat{\tau}}_{k_k}(t,{\bf x}) } - \frac{ {\check{\tau}}_{i_i}(t,{\bf x}) - z_{i_i}(t,{\bf x}) } {{\hat{\tau}}_{i_i}(t,{\bf x}) } \, \frac{ {\hat{\tau}}_{i_k}(t,{\bf x}) } {{\hat{\tau}}_{k_k}(t,{\bf x}) }  \right] {{\hat{\tau}}_{k_\ell}(t,{\bf x}) } \\
 \tau^{\texttt{c}}_{i_\ell k_\ell}(t) & = \frac{ {\check{\tau}}_{i_i}(t,{\bf x}) - z_{i_i}(t,{\bf x}) } {{\hat{\tau}}_{i_i}(t,{\bf x}) } \, \frac{ {\hat{\tau}}_{i_k}(t,{\bf x}) } {{\hat{\tau}}_{k_k}(t,{\bf x}) }  \, {{\hat{\tau}}_{k_\ell}(t,{\bf x}) } \\
 \tau^{\texttt{t}}_{i_\ell k_\ell}(t) & = \frac{ {\check{\tau}}_{i_k} (t,{\bf x}) - z_{i_k}(t,{\bf x}) }  {{\hat{\tau}}_{k_k}(t,{\bf x}) } \, { {\hat{\tau}}_{k_\ell}(t,{\bf x}) }
\end{aligned}
\end{equation}
for $t > t_0$, $i,k=1,\ldots,n$, and $\ell=0,\ldots,n$, using the proportionality of parallel subflows given in Eq.~\ref{eq:thr_dense3} \cite{Coskun2017DCSAM}. Note that $\hat{\tau}_{k_k}(t_0)=0$ and we assume that $\hat{\tau}_{k_k}(t)$ is nonzero for all $t > t_0$. The second equality of the first equation in Eq.~\ref{eq:comp_diact_subs} for the composite direct subflow is due to the equivalence of flow and subflow intensity in the same direction \cite{Coskun2017DCSAM,Coskun2017SCSA}. The dynamic and static \texttt{diact} flows and storages derived through the dynamic and static approaches are listed in matrix form in Table~\ref{tab:flow_stor1} and~\ref{tab:flow_stor} \cite{Coskun2017DCSAM,Coskun2017SCSA}.

The diagonal matrices $\check{\mathsf{T}}(t,{\bf x})$, $\hat{\mathsf{T}}(t,{\bf x})$, and $\tilde{\mathsf{T}}(t,{\bf x})$ used in Table~\ref{tab:flow_stor1} are defined as
\begin{equation}
\label{eq:two_matx}
\begin{aligned}
\check{\mathsf{T}}(t,{\bf x}) \coloneqq \diag{( \check{T}(t,{\bf x}) )} , \, \, \,
\hat{\mathsf{T}}(t,{\bf x}) \coloneqq \diag{( \hat{T}(t,{\bf x}) )} , \, \, \,
\tilde{\mathsf{T}}(t,{\bf x}) \coloneqq \diag{( \tilde{T}(t,{\bf x}) )} .
\end{aligned}
\end{equation}
The inverted matrices in the table are assumed to be invertible. The \texttt{diact} flow distribution matrices listed in this table for subsystems are the same with the ones for the initial subsystems, due to the proportionality given in Eqs.~\ref{eq:thr_dense3} and~\ref{eq:thr_denseI2}. The simple and composite \texttt{diact} flow matrices for the initial subsystems can be formulated similar to their counterparts in Table~\ref{tab:flow_stor1} as follows:
\begin{equation}
\label{eq:diact_init}
\begin{aligned}
\ubar{T}^\texttt{*}(t,{\bf x}) & \coloneqq {N}^\texttt{*}(t) \, \ubar{\mathcal{T}}(t,{\bf x}) = {N}^\texttt{*}(t) \, \mathcal{T}_0(t,{\bf x}) , \\
\ubar{T}^\texttt{*}_\ell(t,{\bf x}) & \coloneqq {N}^\texttt{*}(t) \, \ubar{\mathcal{\hat T}}_\ell(t,{\bf x}) ,
\quad \mbox{and} \quad
\ubar{\tilde{T}}^\texttt{*}(t,{\bf x}) \coloneqq {N}^\texttt{*}(t) \, \ubar{\hat{\mathsf{T}}}(t,{\bf x}) ,
\end{aligned}
\end{equation}
where the superscript ($^\texttt{*}$) represents any of the \texttt{diact} symbols and $\ubar{\hat{\mathsf{T}}}(t,{\bf x}) \coloneqq \diag{( \ubar{\hat{T}}(t,{\bf x}) )} $.

The simple and composite \texttt{diact} storages can then be formulated using the corresponding \texttt{diact} flows as the transient inflows in Eq.~\ref{eq:out_in_fs2} as follows:
\begin{equation}
\label{eq:out_in_diact2}
\begin{aligned}
\dot{x}^{\texttt{*}}_{i_\ell k_\ell}(t) & = \tau^{\texttt{*}}_{i_\ell k_\ell}(t) - \frac{ {\hat{\tau}}_{i}(t,{x}) }{ x_{i}(t) } \, {x}^{\texttt{*}}_{i_\ell k_\ell}(t) , \quad {x}^{\texttt{*}}_{i_\ell k_\ell}(t_1) = 0
\end{aligned}
\end{equation}
for $t_1 \geq t_0$. The solution to this governing equation, ${x}^{\texttt{*}}_{i_\ell k_\ell}(t)$, represents the \texttt{diact} substorage at time $t$ generated by the corresponding \texttt{diact} subflow, $\tau^{\texttt{*}}_{i_\ell k_\ell}(t)$, during $[t_1,t]$. For linear systems, Eq.~\ref{eq:out_in_diact2} can also be solved analytically as formulated in Eq.~\ref{eq:sln_transient} \cite{Coskun2017DCSAM}. The governing equation for the \texttt{diact} storages within the initial subsystems can be formulated similar to Eq.~\ref{eq:out_in_diact2} using the corresponding \texttt{diact} flows as well.
\begin{table}
     \centering
     \caption{The dynamic \texttt{diact} flow distribution and the simple and composite \texttt{diact} (sub)flow matrices. The superscript ($^\texttt{*}$) in each equation represents any of the \texttt{diact} symbols. For the sake of readability, the function arguments are dropped.}
     \label{tab:flow_stor1}
     \begin{tabular}{c p{6cm} l }
     \hline
\texttt{diact} & {flow distribution matrix} & {flows} \\
     \hline
     \noalign{\vskip 2pt}
\texttt{d} & $ N^\texttt{d} =  F \, \mathcal{T}^{-1}  $  &
\multirowcell{5}{
$
\begin{aligned}
\hfill
{T}^\texttt{*} &= {N}^\texttt{*} \, (\mathcal{T} -  \mathcal{\hat T}_0) \\
{T}^\texttt{*}_\ell &= {N}^\texttt{*} \, \mathcal{\hat T}_\ell \\
\tilde{T}^\texttt{*} &= {N}^\texttt{*} \, \mathsf{\hat T}
\end{aligned}
$
}
\\
\texttt{i} & $ N^\texttt{i} = \displaystyle \tilde{T} \, \hat{\mathsf{T}}^{-1} - F \,\mathcal{T}^{-1}  $ &  \\
\texttt{a} & $ {N}^\texttt{a} = \displaystyle \tilde{T} \, \hat{\mathsf{T}}^{-1} - \tilde{\mathsf{T}} \, \hat{\mathsf{T}}^{-1} \, \hat{T} \, \hat{\mathsf{T}}^{-1} $ & \\
\texttt{c} & $ {N}^\texttt{c} = \displaystyle \tilde{\mathsf{T}} \, \hat{\mathsf{T}}^{-1} \, \hat{T} \, \hat{\mathsf{T}}^{-1} $
& \\
\texttt{t} & $ N^\texttt{t} = \displaystyle  \tilde{T} \, \hat{\mathsf{T}}^{-1}   $ & \\
\noalign{\vskip 1pt}
\hline
     \end{tabular}
\end{table}

\subsection{System analysis and measures}
\label{sec:sa}

The dynamic system decomposition\break methodology yields the subthroughflow and substorage matrices that measure the external influence on system compartments in terms of the flow and storage generation. For the quantification of intercompartmental flow and storage dynamics, the dynamic subsystem decomposition methodology then formulates the transient and dynamic $\texttt{diact}$ flows and storages. These mathematical system analysis tools and their interpretation as quantitative system indicators will be discussed in this section.

The elements of the fundamental matrix solutions, that is, those of the initial substate (substorage) and substate matrices, $\ubar{X}(t)$ and $X(t)$, represent the organization of storages within the system derived from the initial stocks and external inputs, respectively. More specifically, $\ubar x_{i_k}(t)$ represents the storage value in compartment $i$ at time $t$, derived from the initial stock in compartment $k$ during time interval $[t_0,t ]$. Similarly, $x_{i_k}(t)$ represents the storage in compartment $i$ at time $t$ generated by the external input into compartment $k$, $z_k(t)$, during $[t_0,t]$ (see Fig.~\ref{fig:sc}). In other words, the proposed methodology can dynamically partition composite compartmental storages into subcompartmental segments based on their constituent sources from the initial stocks and external inputs. This decomposition enables tracking the evolution of the initial stocks and external inputs, in terms of storage generation, individually and separately within the system. The state variable, $x_i(t)$, which represents the composite compartmental storage, cannot be used to distinguish the portions of this storage derived from different individual initial and external sources separately. Therefore, the solution to the decomposed system brings out inferences that cannot be obtained through the analysis of the original system by the state-of-the-art techniques.

The elements of the net initial subthroughflow and subthroughflow rate matrices, $\ubar T(t,{\mathbf x})$ and $T(t,{\mathbf x})$, represent the distribution of the subthroughflows within the system derived from the initial storages and external inputs, respectively. More specifically, $\ubar \tau_{i_k}(t,{\mathbf x})$ represents the net subthroughflow rate at compartment ${i}$ at time $t$ derived from the initial stock in compartment $k$ during $[t_0,t ]$. Similarly, $\tau_{i_k}(t,{\mathbf x})$ represents the net subthroughflow rate at compartment ${i}$ at time $t$, generated by the external input into compartment $k$ during $[t_0,t ]$ (see Fig.~\ref{fig:fd}).  In other words, the proposed methodology can dynamically partition composite compartmental throughflows into subcompartmental segments based on their constituent sources from the initial stocks and external inputs. This decomposition enables tracking the evolution of the initial stocks and external inputs, in terms of flow generation, individually and separately within the system. Thus, the initial subthroughflow and subthroughflow functions of the decomposed system, $\ubar{\tau}_{i_k}(t,{\bf x})$ and $\tau_{i_k}(t,{\bf x})$, provide more detailed information than the composite throughflow function of the original system, $\tau_i(t,{\bm x})$, similar to the state and substate variables, as explained above. These interpretations can be extended to the inward and outward throughflows for both the initial subsystems and subsystems, $\ubar{\check{\tau}}(t,{\mathbf x})$, $\ubar{\hat{\tau}}(t,{\mathbf x})$, $\check{\tau}(t,{\mathbf x})$, and $\hat{\tau}(t,{\mathbf x})$, as well.

The {\em transient} flows and associated storages transmitted along given flow paths are also formulated systematically, through subsystem decomposition methodology. Arbitrary composite intercompartmental flows and storages can dynamically be decomposed into the constituent transient subflow segments at and substorage portions in each compartment along a given set of subflow paths. Therefore, the dynamic {\em subsystem} decomposition determines the distribution of arbitrary intercompartmental flows and the organization of the associated storages generated by these flows along given subflow paths within the subsystems. In other words, the subsystem decomposition enables dynamically tracking the fate of arbitrary intercompartmental flows and storages within and monitoring their spread throughout the system. Consequently, the proposed methodology determines the dynamic influence of one compartment, through direct or indirect interactions, on any other in a complex network. Moreover, a history of compartments visited by arbitrary system flows and storages can also be compiled.

The dynamic direct, indirect, acyclic, cycling, transfer (\texttt{diact}) flows and storages transmitted from one compartment, directly or indirectly, to any other\textemdash including itself\textemdash within the system are also formulated using the transient flows and storages for the quantification of intercompartmental flow and storage dynamics. More specifically, the dynamic \texttt{diact} flows (subflows) at time $t$, ${\tau}^{\texttt{*}}_{i k}(t)$ (${\tau}^{\texttt{*}}_{i_\ell k_\ell}(t)$), transmitted directly or indirectly from compartment $k$ to $i$ (subcompartment $k_\ell$ to $i_\ell$) and the associated storages (substorages) generated by these flows during $[t_1,t]$, ${x}^{\texttt{*}}_{i k}(t)$ (${x}^{\texttt{*}}_{i_\ell k_\ell}(t)$), can be determined at the compartmental (subcompartmental) level.

The proposed methodology constructs a base for the formulation of new mathematical system analysis tools of matrix, vector, and scalar types as quantitative system indicators. In addition to the system measures summarized in this section, the measures and indices for the \texttt{diact} effects, utilities, exposures, residence times, as well as the corresponding system efficiencies, stress, and resilience have recently been formulated by \cite{Coskun2017DCSAM,Coskun2017DESM} in the context of ecosystem ecology. The static versions of these system analysis tools have also been introduced in separate works \cite{Coskun2017SCSA,Coskun2017SESM}.

\section{Results}
\label{sec:results}

The proposed methodology is applied to various compartmental models from literature in this section and in Appendix~\ref{apxsec:ex}. The results and their interpretations are presented.

\subsection{Case study}
\label{ex:sir}

The SIR model is one of the simplest compartmental models in epidemiology which consists of three compartments that represent the populations of three groups: the susceptible or uninfected, $x_1 = S$, infectious, $x_2 = I$, and recovered or immune, $x_3 = R$. The model determines the number of individuals infected with a contagious illness over time. It is reasonably predictive for infectious diseases transmitted from individual to individual. The first SIR model was proposed in its simplest form by \cite{Kermack1927}.

In this section, we will analyze a modified version of SIR model through the proposed methodology: SIRS model for waning immunity with demographics. The model parameters are adopted from \cite{Anderson1979} and the modeling assumptions can be deduced from the model formulation below or can be readily found in the literature \cite{Anderson1979,Keshet2004}.

The governing system of equations of SIRS model for a laboratory population of mice infected with microbes can be formulated as follows:
\begin{equation}
\begin{aligned}
\frac{d x_1}{dt} & = \alpha + \nu \, x_3 - \beta \, x_1 \, x_2 - \mu \, x_1  \\
\frac{d x_2}{dt} & = \beta \, x_1 \, x_2 - (\gamma + \sigma + \mu ) \, x_2 \\
\frac{d x_3}{dt} & = \gamma \, x_2  - ( \nu + \mu ) \, x_3
\end{aligned}
\label{eq:sir}
\end{equation}
with the initial conditions of $\bm{x}(t_0) = [10,10,1]^T$. The total initial population is given to be $20$ by \cite{Anderson1979}, but the initial population for each group is not specified individually. They are, therefore, arbitrarily chosen in this work. The model parameters are the birth rate (or daily rate of mice introduced into susceptible population) $\alpha=0.33$, the natural mortality rate $\mu=0.006$, the mortality rate caused by the disease $\sigma=0.06$, the infection rate $\beta=0.0056$, the recovery rate $\gamma=0.04$, and the immunity loss rate $\nu=0.021$. All parameters are in units of [day$^{-1}]$ (see Fig.~\ref{fig:sir}).
\begin{figure}[t]
\begin{center}
\begin{tikzpicture}
\centering
   \draw[very thick,  fill=blue!5, draw=black] (-.05,-.05) rectangle node(R1) {$x_1(t)$} (1.5,1.5) ;
       \node (x1c) at (.1,1.8) {$S$};
   \draw[very thick,  fill=blue!5, draw=black] (3.95,-.05) rectangle node(R2) {$x_2(t)$} (5.5,1.5) ;
       \node (x1c) at (5.3,1.8) {$I$};
   \draw[very thick,  fill=blue!5, draw=black] (1.95,2.6) rectangle node(R3) {$x_3(t)$} (3.55,4.1) ;
       \node (x3c) at (1.6,3.8) {$R$};
       \draw[very thick,-stealth,draw=red, line width=5pt, opacity=.2]  (-1.2,1) -- (-0.15,1) -- (1.6,.75) -- (3.8,.75) -- (4,0.75) -- (4.8,1.4) -- (4.8,1.7) -- (3.75,3.2) -- (1.75,3.2) -- (0.7,1.7) -- (0.7,.5) -- (4.5,.5) -- (6.7,.5) ;
       \draw[very thick,-stealth,draw=red, line width=5pt, opacity=.2]  (-1.2,1) -- (-0.15,1) -- (1.6,.75) -- (3.8,.75) -- (4,0.75) -- (4.8,1.4) -- (4.8,1.7) -- (3.75,3.2) -- (1.75,3.2) -- (0.7,1.7) -- (0.7,.5) -- (4.5,.5) -- (5.5,1) -- (6.7,1) ;
       \draw[very thick,-stealth,draw=black]  (1.6,.75) -- (3.8,.75) ;  
       \node (f_21) at (2.7,1.1) {$\beta$};
       \draw[very thick,stealth-,draw=black] (0.7,1.7) -- (1.75,3.2) ;     
       \node (f_32) at (.9,2.5) {$\nu$};
       \draw[very thick,-stealth,draw=black]  (4.8,1.7) -- (3.75,3.2) ;  
       \node (f_32) at (4.6,2.5) {$\gamma$};
       \draw[very thick,-stealth,draw=black]  (5.6,.5) -- (6.6,.5) ;      
       \draw[very thick,-stealth,draw=black]  (5.6,1) -- (6.6,1) ;         
       \node (y_22) at (6.4,1.3) {$\sigma$};
       \node (y_21) at (6.4,.2) {$\mu$};
       \draw[very thick,stealth-,draw=black]  (-1.2,.5) -- (-0.15,.5) ;  
       \node (z_1) at (-1,1.3) {$\alpha$};
       \node (y_1) at (-1,.2) {$\mu$};
       \draw[very thick,-stealth,draw=black]  (-1.2,1) -- (-0.15,1) ;    
       \draw[very thick,-stealth,draw=black]  (3.7,3.5) -- (4.8,3.5) ;   
       \node (y_3) at (4.6,3.8) {$\mu$};
       \node (p) [text=red] at (7.2,1.1) {$p^2_{0_1 1_1}$};
       \node (p) [text=red] at (7.2,.5) {$p^3_{0_1 1_1}$};
\end{tikzpicture}
\end{center}
\caption{Schematic representation of the model network. Arrows are labeled by the corresponding rate constants. Subflow paths $p^2_{0_1 1_1}$ and $p^3_{0_1 1_1}$ along which the transient external outputs are computed are red (subsystems are not shown) (Case study~\ref{ex:sir}).}
\label{fig:sir}
\end{figure}
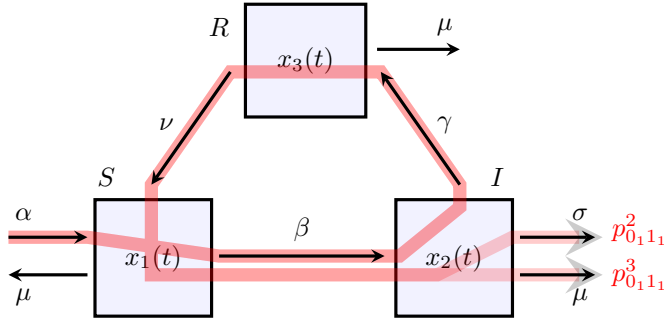

The flow regime for the system can be expressed in matrix form as
\begin{equation}
\label{eq:sir_flows}
\begin{aligned}
F(t, \bm{x}) =
\begin{bmatrix}
    0  & 0 & \nu \, x_3 \\
    \beta \, x_1 \, x_2 & 0 & 0 \\
    0 & \gamma \, x_2 & 0
\end{bmatrix},
\quad
\bm{z}(t,\bm{x}) =
\begin{bmatrix}
    \alpha \\
    0 \\
    0
\end{bmatrix},
\quad
\bm{y}(t,\bm{x}) =
\begin{bmatrix}
    \mu \, x_1 \\
    (\mu + \sigma) \, x_2 \\
    \mu \, x_3
\end{bmatrix} .
\end{aligned}
\nonumber
\end{equation}
The decomposed system can then be expressed in the following matrix form:
\begin{equation}
\label{eq:model_exM}
\begin{aligned}
\dot{X}(t) & = \mathcal{Z} (t,\bm{x}) + A(t,\bm{x}) \, X(t) ,
\quad X(t_0) = \mathbf{0} , \\
\dot{\ubar X}(t) & = A(t,\bm{x}) \, \ubar{X}(t) ,
\quad \quad \quad \quad \quad \, \, {\ubar X}(t_0) = \mathcal{X}_0  .
\end{aligned}
\end{equation}
The fundamental substate matrices, $X(t) = (x_{i_k}(t) )$ and $\ubar{X}(t) = (\ubar{x}_{i_k}(t) )$, the state, external output, and input matrices, $\mathcal{X}(t) = \diag{(\bm{x}(t))}$, $\mathcal{Y}(t,\bm{x}) = \diag{(\bm{y}(t,\bm{x}))}$, and $\mathcal{Z}(t,\bm{x}) = \diag{(\bm{z}(t,\bm{x}))}$, as well as the flow intensity matrix,
\begin{equation}
\label{eq:matrix_Aex}
\begin{aligned}
A(t,\bm{x}) = & \left ( F (t,\bm{x}) - \mathcal{T}(t,\bm{x}) \right ) \, \mathcal{X}^{-1}(t)
\nonumber
\end{aligned}
\end{equation}
where $\mathcal{T}(t,\bm{x}) = \mathcal{Y}(t,\bm{x}) + \diag{( {F}^T(t,\bm{x}) \, \bm{1} )}$ are defined above in the Methods Section.

The numerical results for the state variables ${\bm{x}}(t)$, $\ubar{\bm{x}}(t)$, and $\bar{\bm{x}}(t)$ are presented in Fig.~\ref{fig:sir_comp} and~\ref{fig:sir_ssm}. As seen from the graphs, the proposed methodology enables dynamically tracking the evolution of the initial and newborn populations individually and separately within this nonlinear system. It is worth emphasizing that such analysis cannot be done through the state-of-the-art techniques.
\begin{figure}[t]
\begin{center}
\includegraphics[width=.45\textwidth]{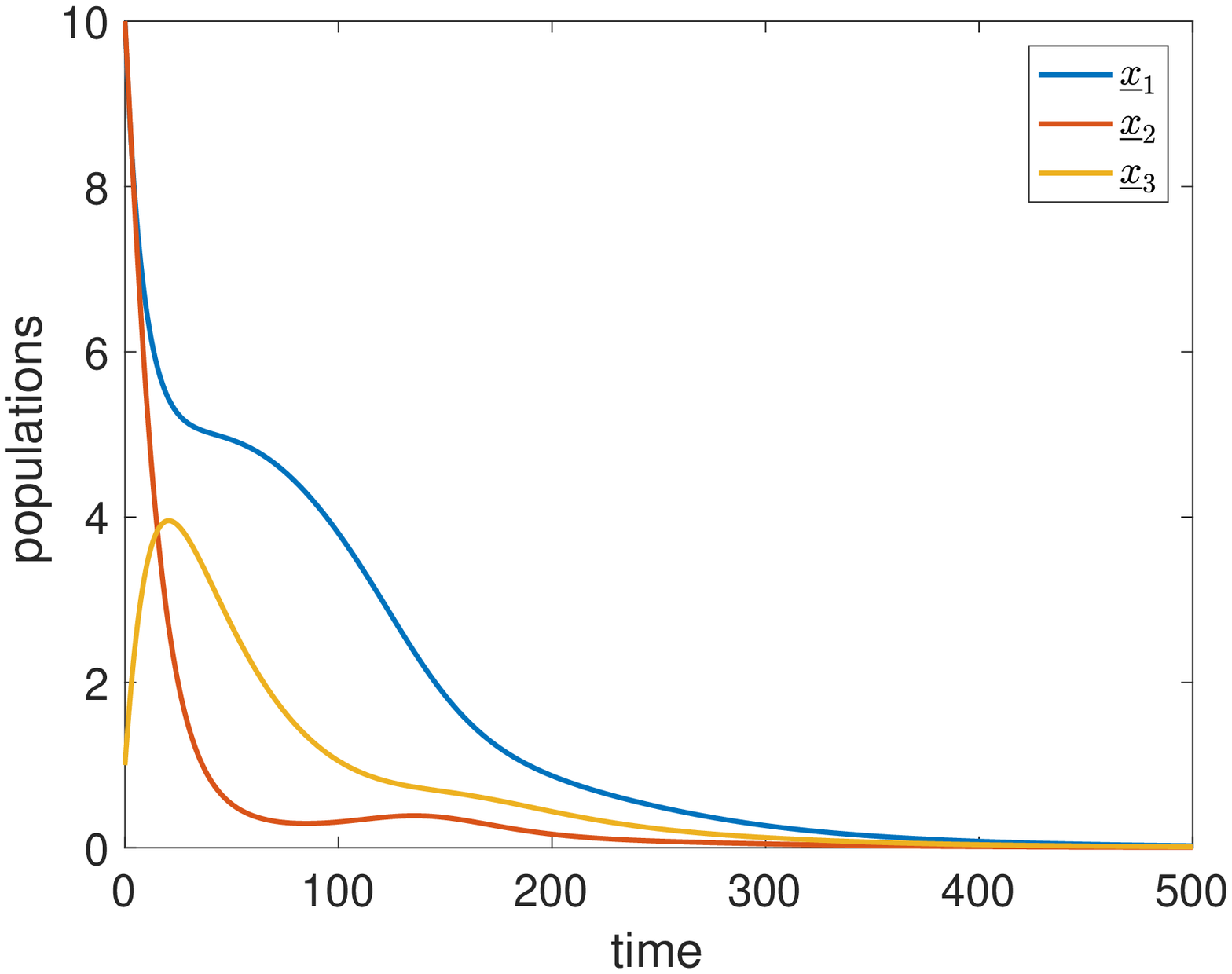}
\includegraphics[width=.45\textwidth]{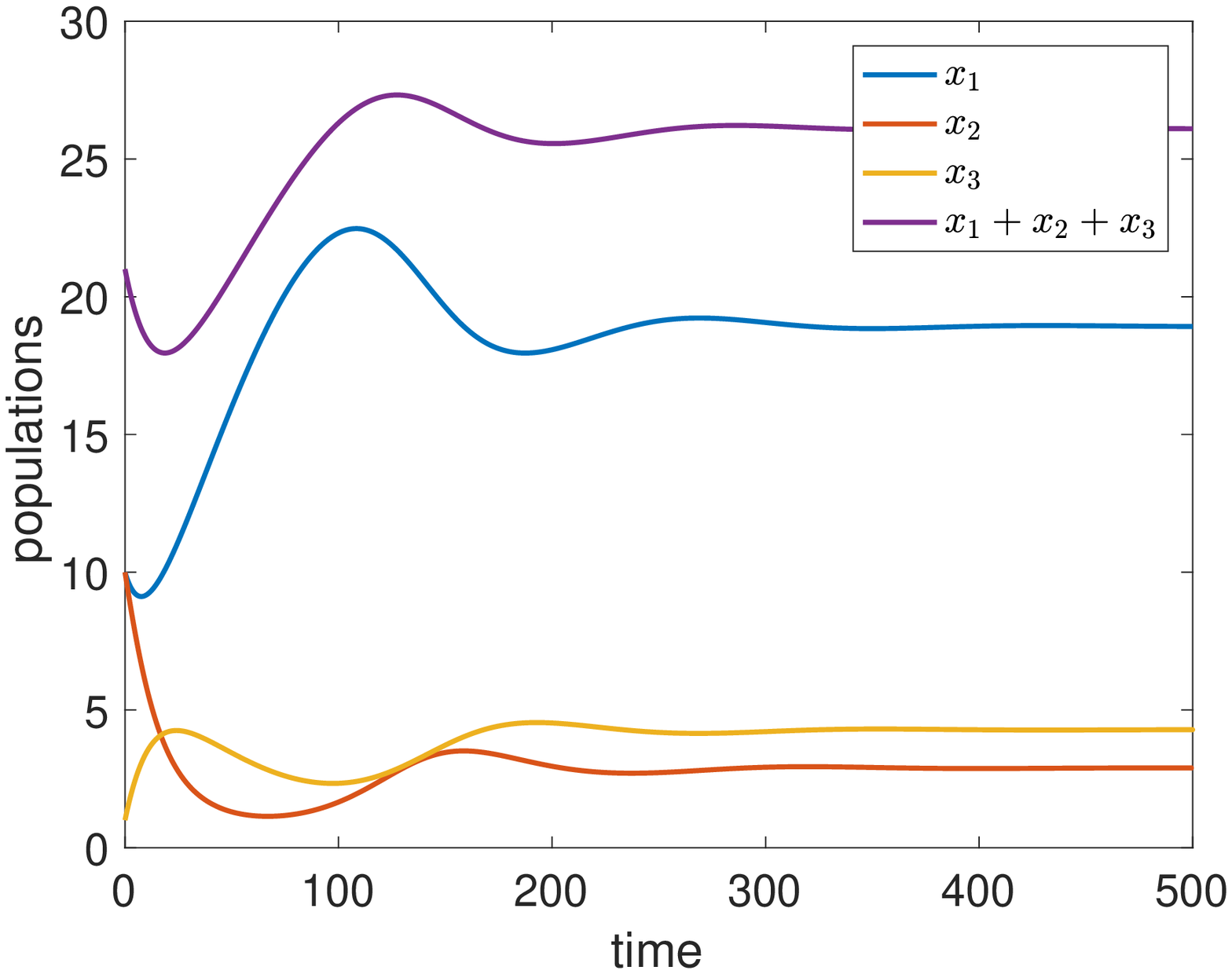}
\end{center}
\caption{Numerical results for the evolution of the initial populations, $\ucbar{\bm{x}}(t)$, and state variables, $\bm{x}(t)$. The populations generated by external inputs alone, $\bar{\bm{x}}(t)$, are presented in Fig.~\ref{fig:sir_ssm} (Case study~\ref{ex:sir}).}
\label{fig:sir_comp}
\end{figure}

The fundamental matrices, that is, the substate (substorage) and initial substate (initial substorage) matrix functions, $X(t)$ and $\ubar{X}(t)$, are also presented in Fig.~\ref{fig:sir_ssm}. Note that the substate functions for the $3^{rd}$ initial subsystem and $2^{nd}$ and $3^{rd}$ subsystems are identically zero because of the zero initial stock, $x_{3_0} =0$, and external inputs, $z_{2}(t) = z_{3}(t) = 0$. That is,
\[ \ubar{x}_{i_3}(t) =0 \quad \mbox{and} \quad {x}_{i_k}(t) =0 \quad \mbox{for} \quad k=2,3 \quad \mbox{and} \quad i=1,2,3 . \]
Since there is only one nonzero input in this model ($z_1(t) >0$), $\bar{x}_1(t) = x_{1_1}(t)$, $\bar{x}_2(t) = x_{2_1}(t)$, and $\bar{x}_3(t) = x_{3_1}(t)$, as presented in Fig.~\ref{fig:sir_ssm}. The initial substate and substate variables $\ubar{x}_{2_1}(t)$ and ${x}_{2_1}(t)$, for example, represent the population in compartment $2$ at time $t$, which are derived from the initial population in compartment $1$, $x_{1,0}$, and external input into compartment $1$, $z_1(t)$, during $[t_0,t]$, respectively. Biologically, $\ubar{x}_{2_1}(t)$ can be interpreted as the population of the infected mice at time $t$, that had initially been susceptible and then infected sometime during $[t_0,t]$. Similarly, ${x}_{2_1}(t)$ represents the population of the infected mice at time $t$, which were born (or introduced) susceptible and then infected during $[t_0,t]$.
\begin{figure}[h]
\begin{center}
\includegraphics[width=.45\textwidth]{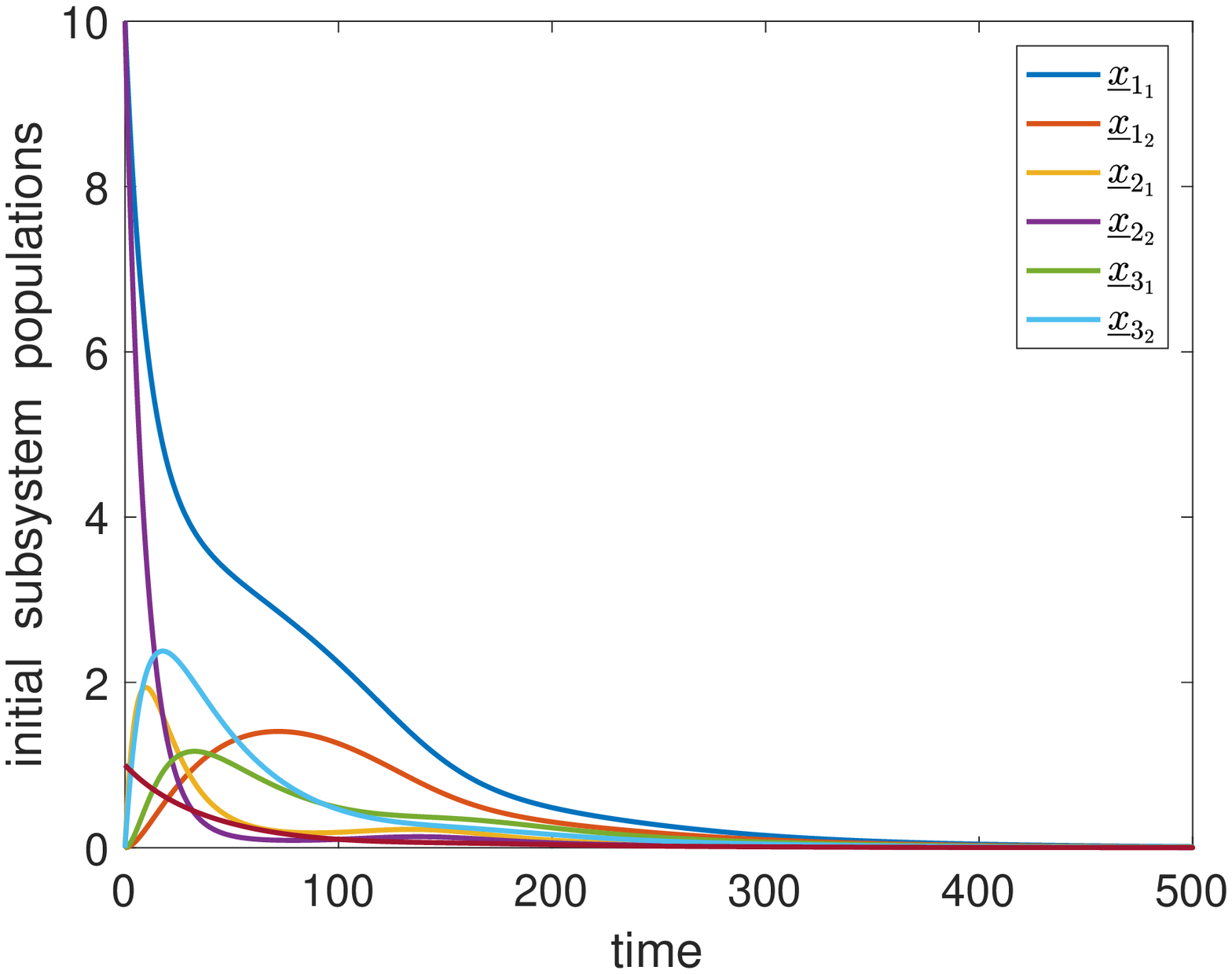}
\includegraphics[width=.45\textwidth]{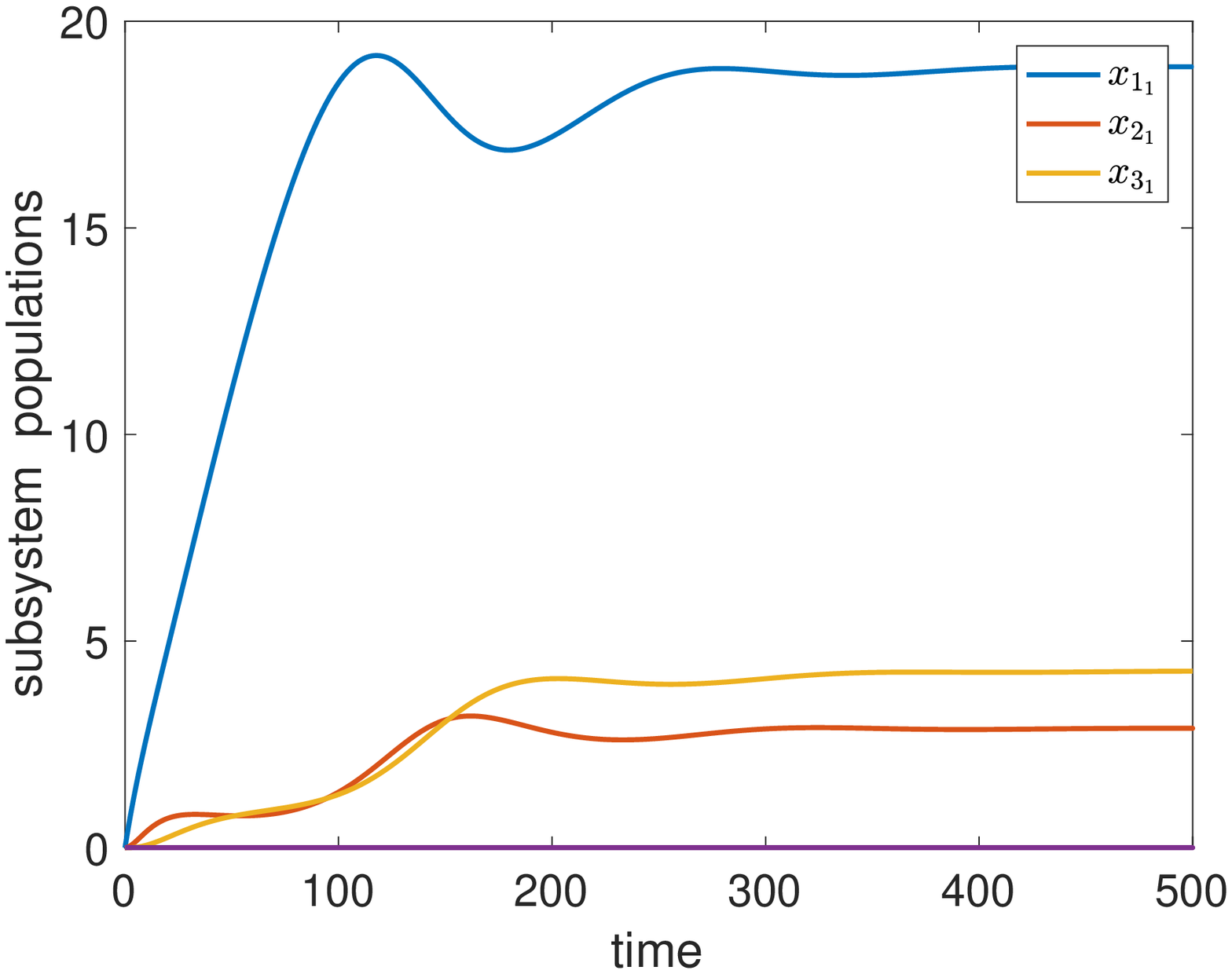}
\end{center}
\caption{The graphical representations of the initial substate (initial substorage) and substate (substorage) functions $\ucbar{x}_{i_k}(t)$ and ${x}_{i_k}(t) = \bar{x}_{i}(t)$ for all $i,k$. The substates that are equal to zero are not labeled. (Case study~\ref{ex:sir}). }
\label{fig:sir_ssm}
\end{figure}

In general terms, the state variable $x_i(t)$ of the original system, Eq.~\ref{eq:sir}, for SIR group dynamics represents the population in group $i$ at time $t$ with the initial population, $x_i(t_0)$. It cannot be used to distinguish the subpopulations derived from the newborn (the only source term in this model, $z_1(t)$) or initial SIR populations (initial conditions, ${x}_{i_0}$). On the other hand, the state variable $x_{i_1}(t)$ of the decomposed system for SIR subgroup dynamics, Eq.~\ref{eq:model_exM}, represents the subpopulation in group $i$ at time $t$, which is transferred from the newborn population during $[t_0,t]$. Similarly, the state variable of the decomposed system, $\ubar{x}_{i_k}(t)$, represents the subpopulation in group $i$ at time $t$, which is transferred from the initial population in group $k$, $x_{k,0}$, during $[t_0,t]$. Parallel interpretations are possible for the throughflow function of the original system, $\tau_i(t,{\bm x})$, and the subthroughflow functions of the decomposed system, $\check{\tau}_{i_k}(t,{\bf x})$, $\hat{\tau}_{i_k}(t,{\bf x})$, $\ubar{\check{\tau}}_{i_k}(t,{\bf x})$, and $\ubar{\hat{\tau}}_{i_k}(t,{\bf x})$, as well.

The proposed dynamic system decomposition methodology, consequently, enables tracking the evolution of the health states of the newborn or initial SIR populations individually and separately. Note that, the solution to the original system through the state-of-the-art techniques can only provide the composite SIR populations without distinguishing the original health states of their subpopulations.

The transient substorages in compartment $2$ along closed subflow path $p^1_{2_1 1_1}\coloneqq 0_1 \mapsto 1_1 \to 2_1 \rightsquigarrow 3_1 \rightsquigarrow 1_1 \to 2_1$ are also computed as an application of the proposed dynamic subsystem decomposition methodology. The links on this path that directly contribute to the cumulative transient substorage ${x}^{1}_{2_1}(t)$ are numbered with red cycle numbers, $m$, in the extended subflow path diagram below:
\begin{equation}
\begin{aligned}
p_{2_1 1_1}^1= 0_1 \mapsto 1_1 \, \xrightarrow{ {\color{red} 1} } \,   2_1 \rightsquigarrow 3_1 \rightsquigarrow   1_1  \, \xrightarrow{ {\color{red} 2} } \,  2_1 \rightsquigarrow  3_1 \rightsquigarrow 1_1  \, \xrightarrow{ {\color{red} 3} } \,  2_1 \rightsquigarrow 3_1 \rightsquigarrow  \cdots
\end{aligned}
\nonumber
\end{equation}
The cumulative transient substorage in subcompartment $2_1$ along $p^1_{2_1 1_1}$, $x^1_{2_1}(t)$, will be approximated by three terms ($m_1=3$) using Eq.~\ref{eq:apxout_in_fs10}:
\begin{equation}
\label{eq:hippe_ss1}
\begin{aligned}
{x}^{1}_{2_1}(t) & \approx \sum_{m=1}^3 {x}^{1,m}_{3_1 2_1 1_1}(t) = {x}^{1,1}_{3_1 2_1 1_1}(t) + {x}^{1,2}_{3_1 2_1 1_1}(t) + {x}^{1,3}_{3_1 2_1 1_1}(t) .
\end{aligned}
\nonumber
\end{equation}
The governing equations, Eqs.~\ref{eq:out_in_fs} and~\ref{eq:out_in_fs2}, for the transient substorage functions, $x^{1,m}_{2_1 1_1 2_1}(t) $, are solved simultaneously together with the decomposed system, Eq.~\ref{eq:model_M}. Numerical results are presented in Fig.~\ref{fig:sir_transient}.
\begin{figure}[h]
\begin{center}
\includegraphics[width=.46\textwidth]{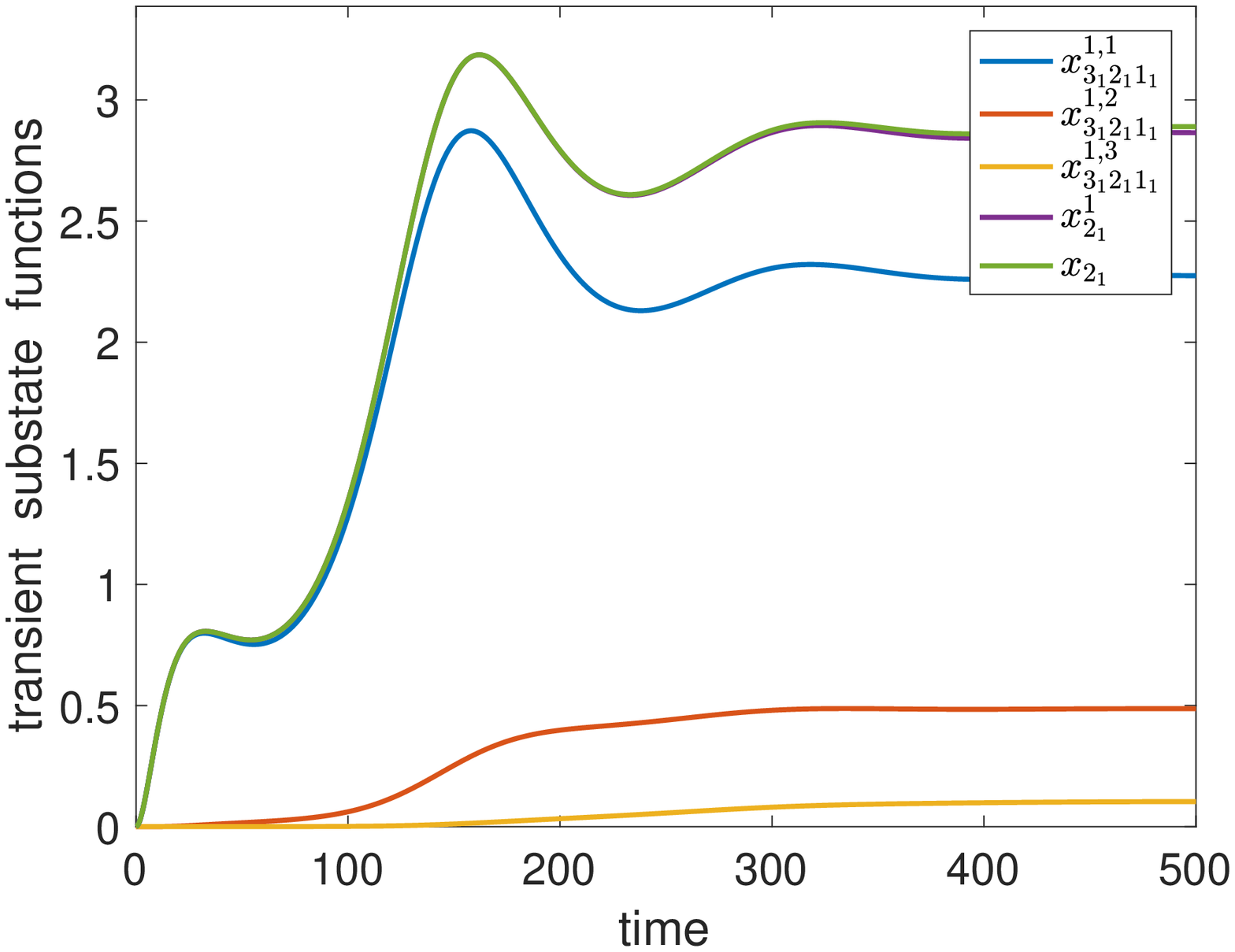}
\includegraphics[width=.45\textwidth]{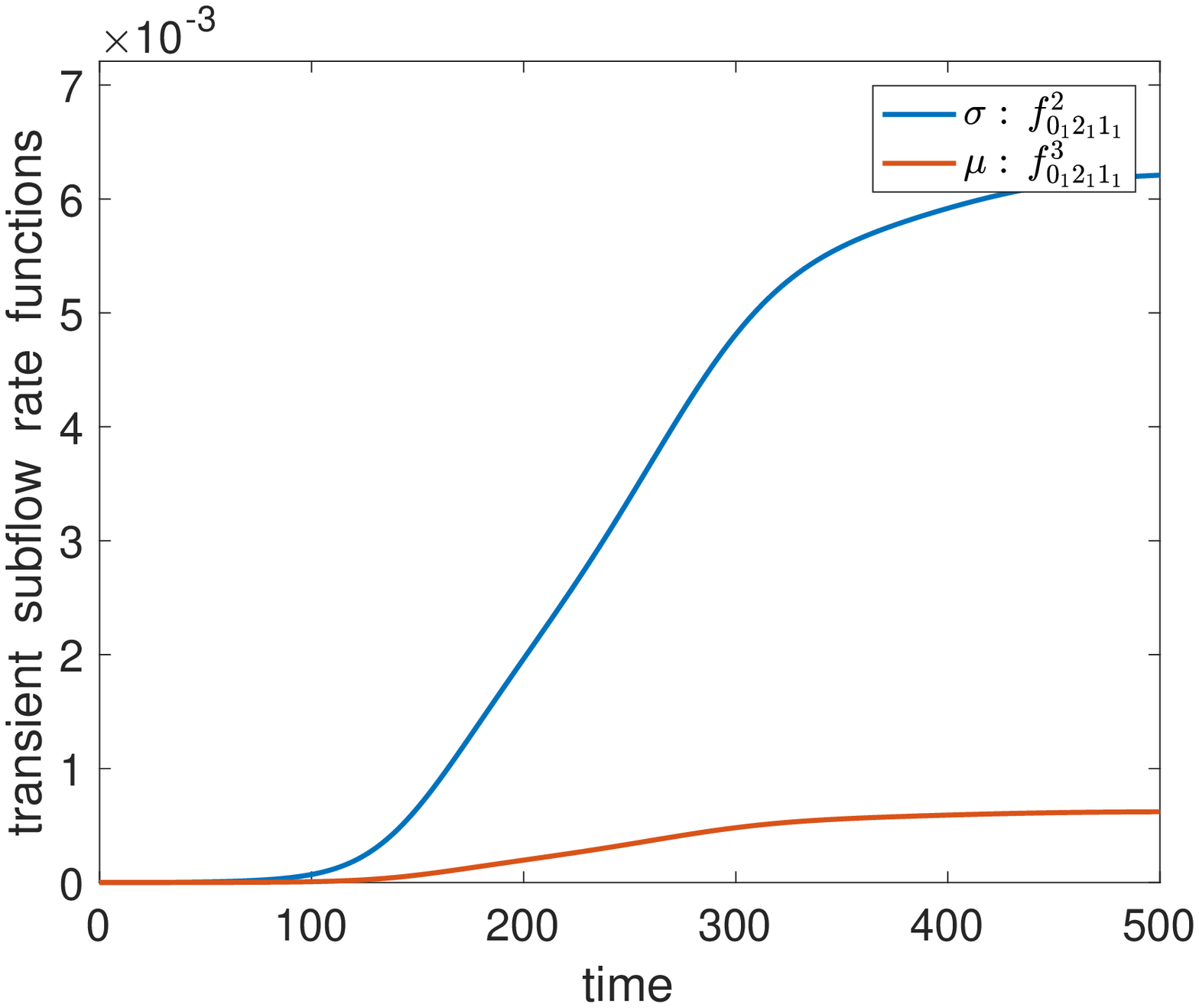}
\end{center}
\caption{The graphical representations of the transient substorage in subcompartment $2_1$, $x^{1,m}_{3_1 2_1 1_1}(t)$, along path $p^1_{2_1 1_1}$, and the transient output rates at subcompartment $2_1$, $f^{2}_{0_1 2_1 1_1}(t)$ and $f^{3}_{0_1 2_1 1_1}(t)$, along paths $p^2_{0_1 1_1}$ and $p^3_{0_1 1_1}$, respectively (Case study~\ref{ex:sir}). }
\label{fig:sir_transient}
\end{figure}
Since the subflow path $p^1_{2_1 1_1}$ covers the entire flow regime in subsystem $1$, $x_{2_1}(t) $ and $x^1_{2_1}(t)$ must be the same. However, the substorage function, $x_{2_1}(t) $, is closely approximated as presented in Fig.~\ref{fig:sir_transient}, that is, $x_{2_1}(t) \approx x^1_{2_1}(t)$. The difference is caused by the truncation errors in the computation of the cumulative transient substorage, and larger ${m_1}$ values improve the approximation. Since subflow path $p_{2_1 1_1}^1$ is closed, $x^1_{2_1}(t)$ is also the simple cycling subflow at subcompartment $2_1$, that is $x^\texttt{c}_{2_1 1_1}(t) = x^\texttt{c}_{2_1}(t) = x^1_{2_1}(t)$. Biologically, the transient substorage functions, $x^{1,m}_{3_1 2_1 1_1}(t)$, represent the population of the mice at time $t$ that are infected $m$ times after being recovered (for $m>1$) during $[t_0,t]$. These subpopulations are decreasing with increasing $m$ values, as expected. Such dynamic characterization of the subpopulations in each SIR group is not possible through the state-of-the-art techniques.

Along the following subflow paths of finite length
\begin{equation}
\begin{aligned}
p_{0_1 1_1}^2 \coloneqq 0_1 \mapsto 1_1 \rightsquigarrow 2_1 \rightsquigarrow 3_1 \rightsquigarrow  1_1 \rightsquigarrow 2_1 \rightsquigarrow 3_1 \rightsquigarrow 1_1  \rightsquigarrow  2_1 \, {\color{red} \xrightarrow{\sigma} } \,  0_1 , \\
p_{0_1 1_1}^3 \coloneqq 0_1 \mapsto 1_1 \rightsquigarrow 2_1 \rightsquigarrow 3_1 \rightsquigarrow 1_1 \rightsquigarrow 2_1 \rightsquigarrow 3_1 \rightsquigarrow 1_1  \rightsquigarrow  2_1 \, {\color{red} \xrightarrow{\mu} } \,  0_1 ,
\end{aligned}
\nonumber
\end{equation}
the transient subflows can also be computed using the governing equations, Eqs.~\ref{eq:out_in_fs} and~\ref{eq:out_in_fs2}, similar to the transient substorages, as discussed above (see Fig.~\ref{fig:sir}). Note that these paths represent the death (external output) of the newborn mice (external input) following two complete infection cycles (after three infections). The only difference between them is the last links which represent the death due to the disease ($\sigma$) and the natural death ($\mu$). The numerical results for the transient external outputs at compartment $2$, $f^{2}_{0_1 2_1 1_1}(t)$ and $f^{3}_{0_1 2_1 1_1}(t)$, which correspond to the last links of these two paths, are depicted in Fig.~\ref{fig:sir_transient}. Because of the corresponding parameter values, $\sigma$ and $\mu$, the death rate due to the disease is 10 times greater than that due to the natural death.
One of these transient output rates at $t=500$ days, $f^{2}_{0_1 2_1 1_1}(500)=0.027$, for example, indicates that $2.7$ out of $100$ mice that were born during $[0,500]$ die per day on the $500^{th}$ day of the experiment due to the disease, after being recovered and getting infected for the third time.

The proposed dynamic subsystem decomposition methodology, consequently, enables tracking the evolution of the health states of an arbitrary population in any of the SIR groups along a particular infection path. Therefore, the effect of an arbitrary population on any SIR group through not only direct but also indirect interactions can be determined. As a result, the spread of the disease from an arbitrary population to the entire population can dynamically be determined and monitored.

The dynamic \texttt{diact} flows and storages are introduced in Section~\ref{apxsec:flows}. The indirect storage transfers from compartment $1$ to $3$ through $2$ can be computed within both the subsystems and initial subsystems, $x^\texttt{i}_{31}(t)$ and $\ubar{x}^\texttt{i}_{31}(t)$, respectively, as formulated in Table~\ref{tab:flow_stor1} \cite{Coskun2017DCSAM}. The graphs of these indirect storage functions are depicted in Fig.~\ref{fig:sir_diact}, as well as the total indirect storage transfer in the same direction, $x^\texttt{i}_{31}(t) + \ubar{x}^\texttt{i}_{31}(t)$. Epidemiologically, these results indicate that the newborn population transferred from the susceptible population indirectly through the infectious population to the recovered population, ${x}^\texttt{i}_{31}(t)$, increases from $0$ to $4.27$ mice during the first 500 days of the experiment. In other words, the number of newborn mice that are recovered after getting infected at least once reaches to $4.27$ by the end of the experiment, i.e., ${x}^\texttt{i}_{31}(500) = 4.27$ mice. This is closer to the steady-state value of the indirect population transfer in the same direction, ${x}^\texttt{i}_{31} = 4.28$ mice, which can be computed as formulated in Table~\ref{tab:flow_stor}. In contrast, the number of initially susceptible mice that are recovered after being infected at least once vanishes in time due to death: $\lim_{t \to \infty} \ucbar{x}^\texttt{i}_{31}(t) = 0 = \ucbar{x}^\texttt{i}_{31}$.

The diagonal residence time matrix is a mathematical system analysis tool introduced in this work as quantitative system indicator \cite{Coskun2017DESM}. The $i^{th}$ diagonal entry of $\mathcal{R}(t,{\bm x})$ at time $t_1$, ${r}_{i} (t_1,{\bm x})$, can be interpreted as the time required for the outward throughflow, at the constant rate of $\hat{\tau}_i(t_1,{\bm x})$, to completely empty compartment $i$, with the storage of $x_{i}(t_1)$. The diagonal structure of the residence time matrix indicates that all subcompartments of compartment $i$ vanish simultaneously, The residence times measure compartmental activity levels \cite{Coskun2017SCSA}. The smaller the residence time the more active the corresponding compartment. The derivative of the residence time matrix will be called the {\em reverse activity rate} matrix \cite{Coskun2017DESM}. The residence times for this model are depicted in Fig.~\ref{fig:sir_diact}. Interestingly, the residence times for both the infectious and recovered mice populations are constant: $r_2(t,{\bm x}) = 9.43$ days and $r_3(t,{\bm x}) = 37.04$ days. This implies that their activity levels do not change in time. The graph of $r_1(t,{\bm x})$ indicates that the oscillations in the residence time of the susceptible mice population fade away in time as the function converges to its steady-state value, $r_1 = 45.08$ days. At time $t = 500$, for example, it would take $r_1(500,{\bm x}) = 45.04$ days at the constant outward throughflow rate of $\hat{\tau}_2(500,{\bm x}) = 0.42$ mice day$^{-1}$ for the susceptible population of $x_1(500)=18.92$ mice to completely vanish\textemdash enter the other two groups or exit the system due to death. The graph shows that the susceptible mice population gets less active, in terms of population transfer, during the first $66.8$ days of the experiment. The residence time gradually increases from the overall minimum $r_1(0,{\bm x}) = 16.13$ days to the maximum $r_1(66.8,{\bm x}) = 80.8$ days and then starts to decrease to $r_1(158.6,{\bm x}) = 38.95$ days during the first period of oscillations, $[0,158.6]$. Accordingly, during this period, the reverse activity rate for the susceptible mice population is positive, $\dot{r}_1(t,{\bm x}) > 0$, over the time interval $[0,66.8)$ and negative, $\dot{r}_1(t,{\bm x}) < 0$, over $(66.8,158.6)$. It can also be seen from the graphical representations that the infectious mice population is more active than the susceptible and recovered mice populations at all times during the course of experiment. The activity levels of the SIR groups are ordered as follows:
\[ r_2(t,{\bm x}) < r_3(t,{\bm x}) < r_1(t,{\bm x}) \quad \mbox{for} \quad t>18.4 . \]

The system approaches an epidemic equilibrium by the end of the experiment, as presented in the graphs of Fig.~\ref{fig:sir_ssm}. At this steady state, the system information becomes
\begin{equation}
\label{eq:ss_info}
F = \left[
\begin{array}{ccc}
0 & 0 & 0.090 \\
0.306 & 0 & 0 \\
0 & 0.116 & 0
\end{array}
\right], \,
\bm{x} = \left[
\begin{array}{ccc}
   18.929 \\
    2.890 \\
    4.281
\end{array}
\right],  \,
\bm{y} = \left[
\begin{array}{ccc}
    0.114 \\
    0.191 \\
    0.026
\end{array}
\right],  \,
\bm{z} = \left[
\begin{array}{ccc}
    0.33 \\
         0 \\
         0
\end{array}
\right] .
\nonumber
\end{equation}
\begin{figure}[h]
\begin{center}
\includegraphics[width=.45\textwidth]{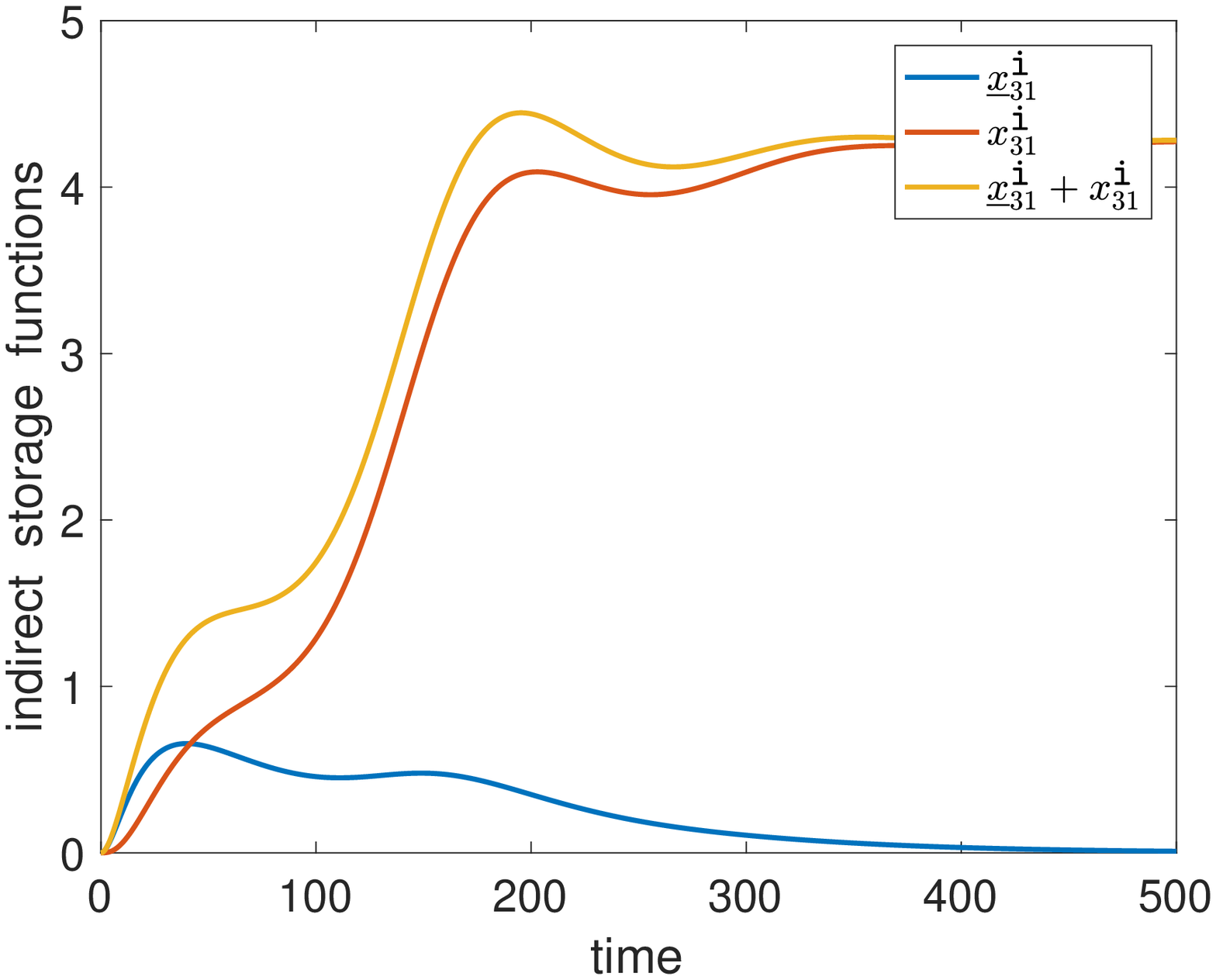}
\includegraphics[width=.47\textwidth]{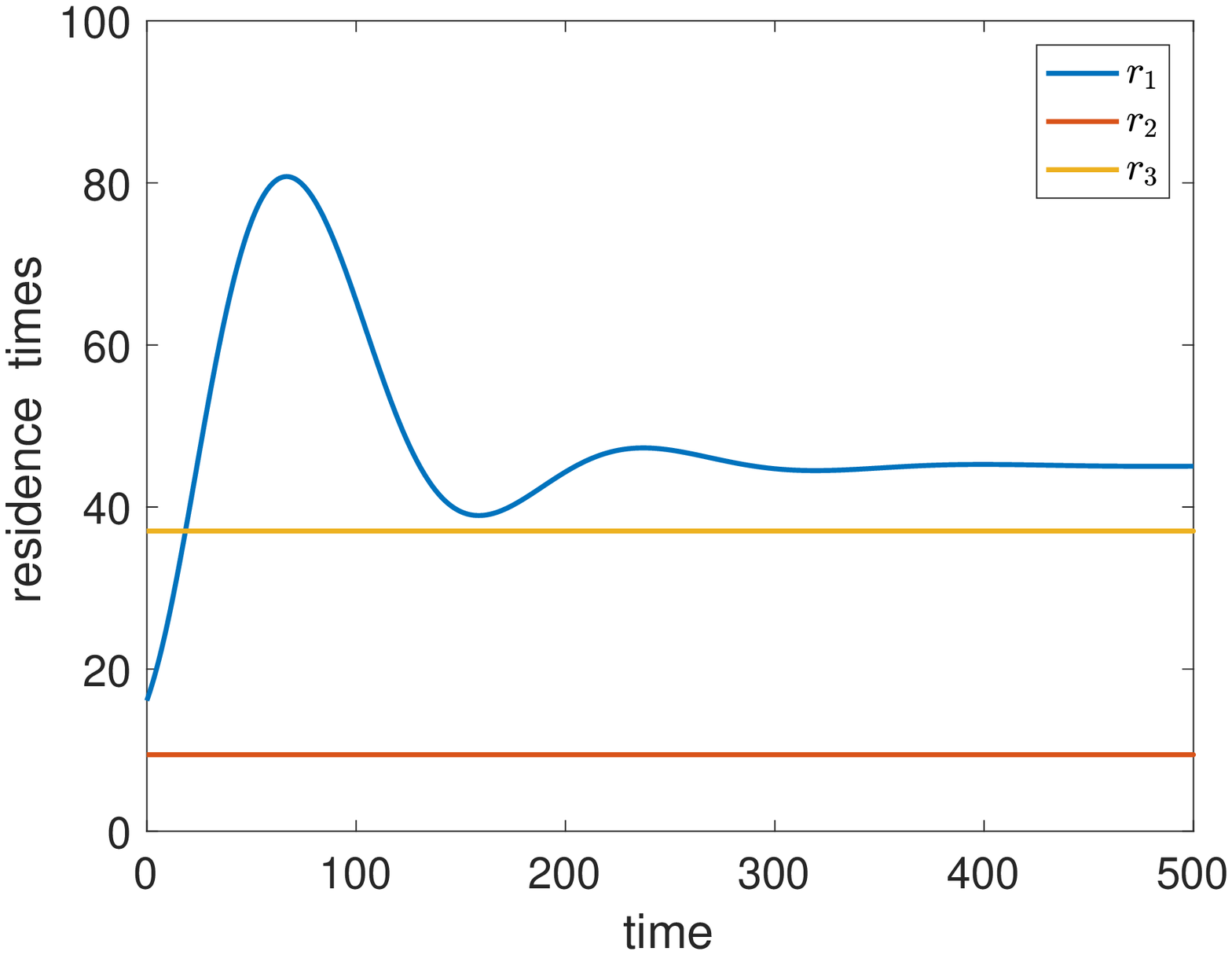}
\end{center}
\caption{The graphical representation for the indirect storages from compartment $1$ to $3$ through $2$ within both the subsystems and initial subsystems, $x^\texttt{i}_{31}(t)$ and $\ucbar{x}^\texttt{i}_{31}(t)$, respectively, and the residence times, $r_i(t)$ (Case study~\ref{ex:sir}). }
\label{fig:sir_diact}
\end{figure}

The static $\texttt{diact}$ flows and storages in matrix form are listed in Table~\ref{tab:flow_stor} as formulated by \cite{Coskun2017SCSA}. For this SIRS model, the composite indirect flow and storage matrices become
\begin{equation}
\label{eq:diact}
T^\texttt{i} = \left[
\begin{array}{ccc}
0.090  &  0.090  &  0 \\
0  &  0.066  &  0.066 \\
0.116  &  0  & 0.025
\end{array}
\right]
\quad \mbox{and} \quad
X^\texttt{i} = \left[
\begin{array}{ccc}
4.053  & 4.053 &  0 \\
0 & 0.619 & 0.619 \\
4.281 & 0 & 0.917
\end{array}
\right] .
\end{equation}
The zero entries of the matrices indicate that there is no indirect flow, and therefore, no storage generated in the corresponding flow direction. For example, the $(1,3)-$entry of $T^\texttt{i}$ is $\tau^\texttt{i}_{13} = 0$, which indicates that although there is a direct flow from the recovered population to the susceptible group, $f_{13}=0.090$, there is no indirect flow in the same direction through the infectious population at the epidemic equilibrium.

Due to the reflexivity of cycling flows, the diagonal elements of $T^\texttt{i}$ and $X^\texttt{i}$ represent the cycling flows and storages \cite{Coskun2017DCSAM}. For example, the cycling flow rate at compartment $3$, $\tau^\texttt{c}_{3} = \tau^\texttt{i}_{33} = 0.025$, indicates that $0.025$ recovered mice become susceptible by losing immunity, getting infected, and being recovered again per day. In other words, $2.5$ out of $100$ mice complete such a recovery cycle per day. This cycling flow constantly maintains the cycling storage of $x^\texttt{c}_{3} = x^\texttt{i}_{3 3} = 0.917$ mice within the recovered mice population. The other static $\texttt{diact}$ flows and associated storages can similarly be computed and interpreted.

Evidently, the detailed information and inferences enabled by the proposed\break methodology cannot be obtained by analyzing the original system through the state-of-the-art techniques.

\section{Discussion}
\label{sec:disc}

We introduced a nonlinear decomposition principle for dynamic compartmental systems based on the novel system and subsystem decomposition methodologies in the present paper. A comprehensive deterministic mathematical method is then developed for the dynamic analysis of nonlinear compartmental systems. The method is applied to various models from literature to demonstrate its efficiency and wide applicability. The results indicate that the proposed theory and methodology provides significant advancements in the theory, methodology, and practicality of the current nonlinear dynamic compartmental system analysis.

Many natural phenomena can be modeled through compartmental systems. Although good rationales are offered in the literature for the analysis of compartmental systems, they are mainly for special cases, such as linear models and static networks. Realistically, nature is always on the move and its systems are constantly changing to meet ever-renewing circumstances. Therefore, the need for the dynamic analysis of nonlinear compartmental systems has always been present.

This is the first manuscript in literature that proposes a theory and develops a comprehensive method for holistic analysis of nonlinear dynamic compartmental systems. The proposed theory is based on the dynamic system and subsystem decomposition methodologies. Nonlinear compartmental systems can be decomposed into mutually exclusive and exhaustive subsystems through the system decomposition methodology. The subsystems are then further decomposed along a given set of mutually exclusive and exhaustive subflow paths through the subsystem decomposition methodology. While the dynamic system decomposition determines the distribution of initial stocks and external inputs, as well as the organization of the associated storages generated by these stocks and inputs individually and separately within the system, the dynamic subsystem decomposition determines the distribution of intercompartmental flows and the organization of associated storages within the subsystems. The proposed mathematical method, therefore, as a whole, yields the dynamic decomposition of system flows and storages to the utmost level.

The {\em system} decomposition methodology yields the subthroughflows and substorages that respectively represent the flows at and storages in each compartment separately derived from individual initial stocks and external inputs. More specifically, the composite compartmental storage and throughflow, $x_i(t)$ and $\tau_i(t)$, are dynamically partitioned into the subcompartmental substorage and subthroughflow segments, $x_{i_k}(t)$ and $\tau_{i_k}(t)$ ($\ubar{x}_{i_k}(t)$ and $\ubar{\tau}_{i_k}(t)$), respectively, based on their constituent sources from external inputs (initial stocks), $z_k(t)$ ($x_{i_0}$). In other words, this methodology enables tracking the evolution of initial stocks and external inputs, as well as associated storages individually and separately within the system. The {\em subsystem} decomposition methodology then yields the transient and the dynamic direct, indirect, acyclic, cycling, and transfer ($\texttt{diact}$) flows and associated storages transmitted along a given flow path or from one compartment, directly or indirectly, to any other. The subsystem partitioning, therefore, enables tracking the fate of intercompartmental flows and associated storages along a particular or all subflow paths within the subsystems. Therefore, the spread of an arbitrary flow or storage segment from one compartment to the entire system can be determined and monitored. Moreover, a history of compartments visited by a particular system flow and storage can also be compiled.

The proposed dynamic methodology also constructs a base for the development of new mathematical system analysis tools as quantitative system indicators. Multiple novel measures, such as the substorage, subthroughflow, \texttt{diact} flow and storage matrices, are already introduced in the present manuscript. The illustrative case studies in Section~\ref{sec:results} and Appendix~\ref{apxsec:ex} demonstrate efficiency and wide applicability of the proposed methodology and measures. More dynamic and static measures and indices have recently been introduced in the context of ecosystem ecology by \cite{Coskun2017DCSAM,Coskun2017DESM,Coskun2017SCSA,Coskun2017SESM}.

The proposed methodology can easily be extended for similar analyses to systems of higher order ordinary and partial differential equations whose source terms governing the intercompartmental interactions are in the form of the conservative compartmental systems as defined in Definition~\ref{def:conservative}.

In summary, we consider that the proposed mathematical theory and methodology brings a novel, formal, deterministic, complex system theory to the service of urgent and challenging natural problems of the day.

\section*{Appendices}

The dynamic system and subsystem decomposition methodologies for the initial subsystem in parallel to the decomposition of the original system, the applications of the proposed method to special cases, such as linear and static models, and additional method formulations are presented in this section.

\appendix

\section{Initial system decomposition}
\label{sec:initsystemd}

The dynamic system decomposition methodology is introduced in Section~\ref{sec:systemd}. In order to analyze the flow distribution and the storage organization derived from the initial stocks individually and separately within the system, the initial subsystem will further be decomposed into {\em initial subsystems} in parallel to the system decomposition. The {\em initial system decomposition methodology} dynamically decomposes the composite subthroughflows and substorages within the initial subsystem into segments based on their constituent sources from the initial stocks. In other words, the initial system decomposition enables dynamically tracking the evolution of the initial stocks individually and separately within the system.

With some abuse of terminology, we will call the subsystems of the initial subsystem the initial subsystems, instead of the initial sub-sub-systems. Each initial subsystem is driven by an initial stock. Therefore, if there is no initial stock in a compartment (initial condition is zero) the corresponding initial subsystem becomes null (see Fig.~\ref{fig:sc} and~\ref{fig:fd}). The number of initial subsystems, therefore, is equal to the number of initial conditions (or compartments). Consequently, there are $n$ initial subsystems, one for each initial stock, indexed by $k=1,\ldots,n$. The dynamic initial system decomposition methodology is introduced in this section.

\subsection{Initial state decomposition}
\label{sec:isd}

Similar to the original system, the initial subcompartment can further be decomposed into $n$ subcompartments (see Fig.~\ref{fig:fd}). We will use the notation $x_{i_{k,0}} (t)$ for $k^{th}$ substate of $i^{th}$ initial substate or ${\ubar x}_{i_k} (t) \coloneqq x_{i_{k,0}} (t)$ for notational convenience. Based on this further decomposition of the initial substates, we have
\begin{equation}
\label{eq:apxi_decomposition}
x_{i_0}(t) = {\ubar x}_{i}(t) = \sum \limits_{k=1}^n {\ubar x}_{i_k}(t)
\end{equation}
for $i=1,\ldots,n$, and the corresponding initial conditions are
\begin{equation}
\label{eq:apxi_inits}
{\ubar x}_{i_k}(t_0) = \delta_{ik} \, x_{i_0}(t_0)  = \left \{
\begin{aligned}
& x_{i_0}(t_0) = x_{i,0}, \quad i = k \\
& 0. \quad \quad \quad \quad \quad \quad \, \, \, i \neq k
\end{aligned}
\right .
\end{equation}

The {\em initial substate matrix} function is then defined as
\begin{equation}
\label{eq:X0}
{\ubar X}(t) \coloneqq \left ( {\ubar x}_{i_k}(t) \right ) = \left [ \ubar{\mathbf{x}}_1(t) \, \ldots \, \ubar{\mathbf{x}}_n(t) \right ].
\end{equation}
We also define the {\em initial state} and the $k^{th}$ {\em initial substate matrix} functions, $\ubar{\mathcal{X}}(t)$ and $\ubar{\mathcal{X}}_k(t)$, as
\begin{equation}
\label{eq:ism}
\ubar{\mathcal{X}}(t) \coloneqq \mathcal{X}_0(t)  = \diag\left ( \ubar{\bm{x}}(t) \right ) \quad \mbox{and} \quad \ubar{\mathcal{X}}_k(t) \coloneqq \diag\left ( \ubar{\mathbf{x}}_k(t) \right )
\end{equation}
for $k=1,\ldots,n$. These matrices will, alternatively, be called the {\em initial substorage} and $k^{th}$ {\em initial substorage matrix} functions, respectively. The initial conditions of these matrices given in Eq.~\ref{eq:apxi_inits} can be expressed in matrix form as
\begin{equation}
\label{eq:X0_init}
\ubar{X}(t_0) = \ubar{\mathcal{X}}(t_0) = \diag (\bm x_0) \quad \mbox{and} \quad \ubar{\mathcal{X}}_k(t_0) \coloneqq \diag \left ( {x}_{k,0} \, \bm{e}_k  \right ) .
\end{equation}
Note that
\begin{equation}
\label{eq:decom_sumI}
\begin{aligned}
\ubar{\bm{x}}(t) = \ubar{X}(t) \, \mathbf{1} \quad \mbox{and} \quad \ubar{\mathbf{x}}_k(t) & = \ubar{\mathcal{X}}_k(t) \, \mathbf{1} .
\end{aligned}
\end{equation}

The state decomposition methodology for the initial subsystem can be summarized as follows:
\begin{center}
\begin{tikzpicture}
  \matrix (m) [matrix of math nodes,row sep=1em,column sep=10em,minimum width=1em]
  {
     {
     \mathbf{x}_0 = \ubar{\bm x}(t) =
 \begin{bmatrix}
  \ubar{x}_{1} \\
  \ubar{x}_{2} \\
  \vdots  \\
  \ubar{x}_{n}
 \end{bmatrix}
 }
 &
 {
 \ubar{X}(t) =
 \begin{bmatrix}
  \ubar{x}_{1_1} & \ubar{x}_{1_2} & \cdots & \ubar{x}_{1_n} \\
  \ubar{x}_{2_1} & \ubar{x}_{2_2} & \cdots & \ubar{x}_{2_n} \\
  \vdots  & \vdots  & \ddots & \vdots  \\
  \ubar{x}_{n_1} & \ubar{x}_{n_2} & \cdots & \ubar{x}_{n_n}
 \end{bmatrix}
 }
 \\
 };
  \path[-stealth, decorate]
    (m-1-1.east|-m-1-2) edge node [below] {state decomposition}
            node [above] {$\ubar{x}_i(t) = \sum \limits_{k=1}^n \ubar{x}_{i_k}(t)$} (m-1-2) ;
\end{tikzpicture}
\end{center}

\subsection{Initial flow decomposition}
\label{sec:ifrd}

Similar to the flow decomposition of the original system introduced in Section~\ref{sec:rd}, the flows within the initial subsystem are also decomposed into {\em initial subflows} that represent the rate of subflow segments between the initial subcompartments (see Fig.~\ref{fig:fd}).

For notational convenience, we set
\begin{equation}
\label{eq:init_flow}
{\ubar f}_{ij}(t,{\mathbf x}) \coloneqq {f}_{i_0 j_0}(t,{\mathbf x}), \, \, \, \ubar{z}_{i} (t,{\bf x}) \coloneqq \ubar{f}_{i 0}(t,\mathbf{x}) = 0, \, \, \, \mbox{and} \quad \ubar{y}_{i} (t,{\bf x}) \coloneqq \ubar{f}_{0i}(t,\mathbf{x}) ,
\end{equation}
for $i,j=1,\ldots,n$. We accordingly define the {\em initial subflow rate}, {\em input}, and {\em output matrix} functions as
\begin{equation}
\label{eq:init_FIO}
\ubar{F}(t,{\mathbf x}) \coloneqq F_0(t,{\mathbf x}), \quad \ubar{\mathcal Z}(t,{\mathbf x}) \coloneqq \mathcal{Z}_0(t,{\mathbf x}) = \bm{0}, \quad \mbox{and} \quad \ubar{\mathcal Y}(t,{\mathbf x}) \coloneqq \mathcal{Y}_0(t,{\mathbf x}) .
\end{equation}
We will use the notation $f_{{i_{k,0}}{j_{k,0}}}(t,{\mathbf x})$ for subflow from initial subcompartment $j_k$ with substorage of ${\ubar x}_{j_k}(t)= {x}_{j_{k,0}}(t)$ to $i_k$ with substorage of ${\ubar x}_{i_k}(t) = {x}_{i_{k,0}}(t)$, for $i,j,k=1,\ldots,n$. For notational convenience, we simplify the notation as ${\ubar f}_{i_k j_k}(t, {\mathbf x}) \coloneqq f_{{i_{k,0}}{j_{k,0}}}(t,{\mathbf x})$ and set $\ubar{z}_{i_k}(t,\mathbf{x}) \coloneqq \ubar{f}_{i_k 0}(t,\mathbf{x}) =0$ and $\ubar{y}_{i_k}(t,\mathbf{x}) \coloneqq \ubar{f}_{0 i_k}(t,\mathbf{x}) $.

The subflow decomposition within the initial subsystem can then be formulated as follows:
\begin{equation}
\label{eq:ic_rates}
\begin{aligned}
{\ubar f}_{i_k j_k}(t,{\mathbf x}) \coloneqq \frac{ {\ubar x}_{j_k}(t)} { {\ubar x}_j(t)} \, {\ubar f}_{ij}(t,{\mathbf x}) =  \frac{\ubar{x}_{j_k}(t)} {x_{j}(t)} \, {f_{ij}(t,{\bm x})}
\end{aligned}
\end{equation}
for $i,j,k=1,\ldots,n$. The second equality is due to Eq.~\ref{eq:csc_dense}. The initial state decomposition formulated in Eq.~\ref{eq:apxi_decomposition} implies that
\begin{equation}
\label{eq:rate_decompositionB}
{\ubar f}_{ij}(t,{\mathbf x}) = \sum \limits_{k=1}^n  {\ubar f}_{i_k j_k}(t,{\mathbf x})
\end{equation}
for $i,j=1,\ldots,n$. It can be seen from Eq.~\ref{eq:ic_rates} that the {\em initial flow} and {\em subflow rate intensities} between the same compartments in the same flow direction are the same. That is,
\begin{equation}
\label{eq:csc_denseI}
\frac{ \ubar{f}_{{i_k}{j_k}}(t,{\mathbf x}) }{ \ubar{x}_{j_k}(t) }  = \frac{ \ubar{f}_{ij}(t,{\mathbf x}) } { \ubar{x}_{j}(t)} = \frac{ f_{{i_0}{j_0}}(t,{\mathbf x}) }{ x_{j_0}(t) } =  \frac{ f_{ij}(t,\bm{x}) } {x_{j}(t)}
\end{equation}
for $i,j,k=1,\ldots,n$, where the denominators are nonzero. The last two equalities are due to Eq.~\ref{eq:csc_dense}.

The decomposition factors, $\ubar{d}_{j_k}({\mathbf x})$, for the initial subflows with the following definition and properties
\begin{equation}
\label{eq:ic_cf}
\ubar{d}_{j_k}({\mathbf x}) \coloneqq \frac{ {\ubar x}_{j_{k}}(t)} { {\ubar x}_{j}(t)} , \quad 0 \leq \ubar{d}_{j_k}({\mathbf x}) \leq 1,  \quad \mbox{and} \quad  \sum_{k=0}^{n} \ubar{d}_{j_k}({\mathbf x}) = 1 ,
\end{equation}
form another continuous partition of unity. Note also that, due to Eq.~\ref{eq:apxi_decomposition}, $\ubar{x}_{j_k}(t) \leq \ubar{x}_{j}(t)$. The decomposition factors are, therefore, well-defined even if $\ubar{x}_j(t) \to 0$. The respective {\em initial decomposition} and $k^{th}$ {\em initial decomposition matrices}, $\ubar{D}({\mathbf x}) \coloneqq \left ( \ubar{d}_{j_k}({\mathbf x}) \right )$ and $\mathcal{\ubar D}_k({\mathbf x}) \coloneqq \diag{([\ubar{d}_{1_k}({\mathbf x}), \ldots, \ubar{d}_{n_k}({\mathbf x})])}$, for the initial subsystems can be formulated, accordingly, as
\begin{equation}
\label{eq:apxcfactorsM2}
\begin{aligned}
\ubar{D}({\mathbf x}) &= \mathcal{\ubar X}^{-1}(t) \, \ubar{X}(t) \quad \mbox{and} \quad
\mathcal{\ubar D}_k({\mathbf x}) = \mathcal{\ubar X}^{-1}(t) \, \mathcal{\ubar X}_k(t)
\end{aligned}
\end{equation}
for $k=1,\ldots,n$. Equations~\ref{eq:decom_sumI} and~\ref{eq:ic_cf} imply that
\begin{equation}
\label{eq:cfactorsI}
\begin{aligned}
\bm{1} = \ubar{\mathcal{X}}^{-1}(t) \, \mathbf{x}_0(t) & = \ubar{\mathcal{X}}^{-1}(t) \, \ubar{X}(t) \, \bm{1} = \ubar{D}({\mathbf x}) \, \bm{1} .
\end{aligned}
\end{equation}
From Eqs.~\ref{eq:cfactorsM} and~\ref{eq:cfactorsI}, we also have
\begin{equation}
\label{eq:cfactors_new}
\begin{aligned}
\bm{x}(t) = \ubar{\bm{x}}(t) + \bar{\bm{x}}(t) = \ubar{X}(t) \, \bm{1} + {X}(t) \, \bm{1} ,
\end{aligned}
\end{equation}
similar to Eq.~\ref{eq:decom_sum}.

We define the {\em $k^{th}$ initial subflow matrix} function as
\begin{equation}
\label{eq:matrix_in_subrates}
\begin{aligned}
{\ubar F}_k(t, {\mathbf x} ) \coloneqq \left ( {\ubar f}_{i_k j_k}(t,{\mathbf x}) \right )
\end{aligned}
\end{equation}
for $k=1,\ldots,n$. Using Eq.~\ref{eq:ic_rates}, ${\ubar F}_k(t, {\mathbf x} )$ can be expressed in the following form:
\begin{equation}
\label{eq:matrix_subratesM2}
\ubar{F}_k(t,{\mathbf x}) = \ubar{F}(t,{\mathbf x}) \, \ubar{\mathcal{D}}_k({\mathbf x}) = \ubar{F}(t,{\mathbf x}) \, \ubar{\mathcal{X}}^{-1}(t) \, \ubar{\mathcal{X}}_k(t) = {F}(t,\bm{x}) \, \mathcal{X}^{-1}(t) \, \ubar{\mathcal{X}}_k(t) .
\end{equation}
That is, the $k^{th}$ initial decomposition matrix, $\mathcal{\ubar D}_k({\mathbf x})$, decomposes the direct initial flow matrix, $\ubar{F}(t,\bf{x})$, into the initial subflow matrices, $\ubar{F}_k(t,{\mathbf x})$. The last equality is derived from Eqs.~\ref{eq:ism} and~\ref{eq:matrix_subratesM}, as these equations imply that
\begin{equation}
\label{eq:matrix_subratesM3}
\ubar{F}(t,{\mathbf x}) = {F}_0(t,{\mathbf x}) =  {F}(t,{\bm x}) \, \mathcal{X}^{-1}(t) \, \mathcal{X}_0(t) = {F}(t,\bm{x}) \, \mathcal{X}^{-1}(t) \, \ubar{\mathcal{X}}(t) .
\end{equation}
Similarly, the $k^{th}$ {\em initial output matrix} function,
\begin{equation}
\label{eq:kthomI}
\ubar{\mathcal{Y}}_k(t,{\mathbf x}) \coloneqq \diag{\left (\left [ \ubar{f}_{0 1_k}(t,{\mathbf x}) , \ldots, \ubar{f}_{0 n_k}(t,{\mathbf x}) \right] \right ) },
\end{equation}
can be expressed in matrix form as
\begin{equation}
\label{eq:output_subratesMI}
\ubar{\mathcal{Y}}_k(t,{\mathbf x}) = \ubar{\mathcal{Y}}(t,{\mathbf x}) \, \ubar{\mathcal{D}}_k ({\mathbf x}) = \ubar{\mathcal{Y}}(t,{\mathbf x}) \, \ubar{\mathcal{X}}^{-1}(t) \, \ubar{\mathcal{X}}_k(t) = \mathcal{Y}(t,\bm{x}) \, \mathcal{X}^{-1}(t) \, \ubar{\mathcal{X}}_k(t)
\end{equation}
and the {\em $k^{th}$ initial input matrix} function becomes
\begin{equation}
\label{eq:kinputmI}
\ubar{\mathcal{Z}}_k(t,{\mathbf x}) = \bm{0}
\end{equation}
for $k=1,\ldots,n$. The {\em initial output} and $k^{th}$ {\em output vector} for the $k^{th}$ initial subsystem, $\ubar{\bm y} (t,{\mathbf x})$ and $\ubar{\mathbf{y}}_k (t,{\mathbf x})$, can be defined as follows:
\begin{equation}
\label{eq:apxin_out_k}
\begin{aligned}
\ubar{\bm{y}} (t,{\mathbf x}) \coloneqq \ubar{\mathcal{Y}} (t,{\mathbf x}) \, \mathbf{1}
\quad \mbox{and} \quad
\ubar{\mathbf{y}}_k (t,{\mathbf x}) \coloneqq \ubar{\mathcal{Y}}_k (t,{\mathbf x}) \, \mathbf{1} .
\end{aligned}
\end{equation}
The {\em initial input} and $k^{th}$ {\em input vectors} are $\ubar{\bm{z}} (t,{\mathbf x}) \coloneqq \ubar{\mathcal{Z}} (t,{\mathbf x}) \, \mathbf{1} = \bm{0} $ and $ \ubar{\mathbf{z}}_k (t,{\mathbf x}) \coloneqq  \ubar{\mathcal{Z}}_k (t,{\mathbf x}) \, \mathbf{1}  = \bm{0} $.

Using these notations, the flow decomposition given in Eq.~\ref{eq:rate_decompositionB} can be expressed in the following matrix form:
\begin{equation}
\label{eq:apxflow_decomposition}
\ubar{F}(t,\bm{x}) = \sum \limits_{k=1}^n  \ubar{F}_k(t,{\mathbf x}) , \, \, \, \ubar{\mathcal{Y}}(t,\bm{x}) = \sum \limits_{k=1}^n \ubar{\mathcal{Y}}_k(t,{\mathbf x}) , \, \, \, \ubar{\mathcal{Z}}(t,\bm{x}) = \sum \limits_{k=1}^n  \ubar{\mathcal{Z}}_k(t,{\mathbf x}) = \bm{0} .
\end{equation}
The equivalence of the flow, initial flow, and initial subflow intensities given in Eq.~\ref{eq:csc_denseI} can also be expressed in matrix form as follows:
\begin{equation}
\label{eq:intensity_flowI}
\ubar{F}_k(t,{\mathbf x})  \, \ubar{\mathcal{X}}_k^{-1}(t)  = \ubar{F}(t,\mathbf{x}) \, \ubar{\mathcal{X}}^{-1}(t) = F_0(t,{\mathbf x})  \, \mathcal{X}_0^{-1}(t) = F(t,{\bm x}) \, \mathcal{X}^{-1}(t)
\end{equation}
for $k=1,\ldots,n$, similar to Eq.~\ref{eq:intensity_flow}.

The initial flow decomposition methodology for the initial subsystem given in Eq.~\ref{eq:ic_rates} can be schematized as follows:
\begin{center}
\begin{tikzpicture}
  \matrix (m) [matrix of math nodes,row sep=.0em,column sep=13em,minimum width=1em]
  {
   \ubar{F} (t,{\mathbf x} ) = {F}_0 (t,{\mathbf x} )
   &
   \ubar{F}_k (t,{\mathbf x}) , \, \, \, k=1,\ldots,n
   \\
 \veq &    \hspace{-2cm} \veq \\
    {
      \begin{bmatrix}
  \ubar{f}_{11} & \cdots & \ubar{f}_{1n} \\
  \ubar{f}_{21} & \cdots & \ubar{f}_{2n} \\
  \vdots  & \ddots & \vdots  \\
  \ubar{f}_{n1} & \cdots & \ubar{f}_{nn}
 \end{bmatrix}
       }
 &
     {
     \begin{bmatrix}
  \ubar{f}_{1_k 1_k} & \cdots & \ubar{f}_{1_k n_k} \\
  \ubar{f}_{2_k 1_k} & \cdots & \ubar{f}_{2_k n_k} \\
  \vdots  & \ddots & \vdots  \\
  \ubar{f}_{n_k 1_k} & \cdots & \ubar{f}_{n_k n_k}
 \end{bmatrix}
      }
 \\ } ;
  \path[-stealth, decorate]
    (m-3-1.east|-m-3-2) edge node [below] {flow decomposition}
            node [above] {$\ubar{f}_{{i_k}{j_k}}(t,{\mathbf x}) = \frac{\ubar{x}_{j_k} (t)} {\ubar{x}_{j} (t)} \, {\ubar{f}_{ij}  (t,{\mathbf x}) } $} (m-3-2) ;
\end{tikzpicture}
\end{center}

\subsection{Initial subsystems}
\label{sec:dic}

The dynamic initial system decomposition methodology is composed of the substate and subflow decomposition components as detailed above. Through this methodology, the initial system is explicitly decomposed into mutually exclusive and exhaustive {\em initial subsystems}, each of which is generated by a single initial stock (see Fig.~\ref{fig:sc} and~\ref{fig:fd}).

The $k^{th}$ subcompartments of each initial subcompartment together with the corresponding $k^{th}$ initial substates, initial subflow rates, inputs, and outputs constitute the $k^{th}$ {\em initial subsystem}. Therefore, the system decomposition methodology generates mutually exclusive and exhaustive subsystems that are running within the initial subsystem and have the same structure and dynamics as the initial subsystem itself, except for their initial conditions. These otherwise-decoupled subsystems are coupled through the decomposition factors. In a system with $n$ compartments, each initial subcompartment has $n$ subcompartments. Consequently, the system has $n$ initial subsystems indexed by $k=1,\ldots,n$. The substorage of each of these $n$ initial subcompartments is derived from a single initial stock. If no initial stock in a compartment (the initial conditions is zero), the corresponding initial subsystem becomes null.

Similar to the governing equations for the original system Eq.~\ref{eq:model1}, the governing equations for the initial subsystem can be written in the following vector form:
\begin{equation}
\label{eq:subsystemsVI}
\begin{aligned}
\dot {\ubar{\bm x}}(t) & = \ubar{\check{\bm \tau}}(t,{\mathbf x}) - \ubar{\hat{\bm \tau}}(t,{\mathbf x}), \quad \ubar{\bm{x}}(t_0) = \bm{x}_0
\end{aligned}
\end{equation}
as given in Eq.~\ref{eq:subsystemsV} where $\ubar {\bm x}(t) = {\mathbf x}_0$, $\ubar{\check{\bm \tau}}(t,{\mathbf x}) = {\check{\bm \tau}}_0(t,{\mathbf x})$, and $\ubar{\hat{\bm \tau}}(t,{\mathbf x}) = {\hat{\bm \tau}}_0(t,{\mathbf x})$. The governing equations for the $k^{th}$ initial subsystem can then be written in vector form as
\begin{equation}
\label{eq:initial_subsystemV1}
\begin{aligned}
\dot {\ubar{\mathbf x}}_k(t) & = \check{\ubar{\bm{\tau}}}_k(t,{\mathbf x}) - \hat{\ubar{\bm{\tau}}}_k(t,{\mathbf x}) \\
& = {\ubar F}_k(t,\mathbf x) \, \bm{1} - \left ( \ubar{\mathbf{y}}_k(t,\mathbf x)  + \ubar{F}_k^T(t,\mathbf x) \, \bm{1} \right )
\end{aligned}
\end{equation}
for $k = 1,\ldots,n$. The corresponding initial conditions become $\ubar{\mathbf x}_k(t_0) = {x}_{k,0} \, \bm{e}_k$.

The $k^{th}$ {\em inward} and {\em outward subthroughflow vectors}, $\check{\ubar{\bm{\tau}}}_k(t,{\mathbf x})$ and $\hat{\ubar{\bm{\tau}}}_k(t,{\mathbf x})$, for the $k^{th}$ initial subsystem can be expressed in the following forms:
\begin{equation}
\label{eq:matrix_subratesB}
\begin{aligned}
\ubar{\check{\bm \tau}}_k (t,{\mathbf x}) & \coloneqq \ubar{F}_k (t,{\mathbf x}) \, \mathbf{1} \\
& =  \ubar{F} (t,\mathbf{x}) \, \ubar{\mathcal{X}}^{-1}(t) \, \ubar {\mathbf x}_k(t)  \\
& = {F} (t,\bm{x}) \, {\mathcal{X}}^{-1}(t) \, \ubar {\mathbf x}_k(t) , \\
\ubar{\hat{\bm \tau}}_k (t,{\mathbf x}) & \coloneqq \ubar{\mathbf y}_k (t,{\mathbf x}) + \ubar{F}^T_k (t,{\mathbf x}) \, \mathbf{1} \\
& = \ubar{\mathcal{Y}}(t,{\mathbf x}) \, \ubar{\mathcal{X}}^{-1}(t) \, \ubar{\mathcal{X}}_k(t) \, \mathbf{1} + \ubar{\mathcal{X}}_k(t) \, \ubar{\mathcal{X}}^{-1}(t) \, \ubar{F}^T(t,{\mathbf x}) \, \mathbf{1} \\
& = \left ( \mathcal{\ubar Y}(t,\bm{x}) + \diag \left ( \ubar{F}^T (t,\bm{x}) \, \mathbf{1} \right ) \right ) \, {\mathcal{\ubar X}}^{-1}(t) \, \ubar{\mathbf x}_k(t) \\
& = \left ( \mathcal{Y}(t,\bm{x}) + \diag \left ( {F}^T (t,\bm{x}) \, \mathbf{1} \right ) \right ) \, {\mathcal{X}}^{-1}(t) \, \ubar{\mathbf x}_k(t) \\
& = {\mathcal{T}}(t,\bm{x}) \, {\mathcal{X}}^{-1}(t) \, \ubar{\mathbf x}_k(t) .
\end{aligned}
\end{equation}
The $k^{th}$ {\em net subthroughflow rate} vector, $\ubar{\bm \tau}_k (t,{\mathbf x}) = \left [ \ubar{\tau}_{1_k}(t,{\mathbf x}), \ldots, \ubar{\tau}_{n_k}(t,{\mathbf x}) \right ]^T$, for the initial subsystem then becomes
\begin{equation}
\label{eq:total_flowsI}
\begin{aligned}
\ubar {\bm \tau}_k (t,{\mathbf x}) & \coloneqq \ubar {\check{\bm \tau}}_k (t,{\mathbf x}) - \ubar {\hat{\bm \tau}}_k (t,{\mathbf x})
= {A} (t,\bm{x}) \, \ubar{\mathbf x}_k(t) .
\end{aligned}
\end{equation}

For each fixed $j$, Eq.~\ref{eq:ic_rates} implies that
\begin{equation}
\label{eq:thr_denseI}
\begin{aligned}
\frac{{\hat{\tau}}_{j_0}(t,\bf{x}) }{{x}_{j_0}(t)} = \frac{\ubar{\hat{\tau}}_{j}(t,\mathbf x) }{\ubar{x}_{j}(t)} = \frac{ \sum_{i=0}^{n} \ubar{f}_{i j}(t,\mathbf x) }{\ubar{x}_{j}(t)} = \frac{  \sum_{i=0}^{n} \ubar{f}_{i_k j_k}(t,{\mathbf x}) }{\ubar{x}_{j_k}(t)} = \frac{\ubar{\hat{\tau}}_{j_k}(t,{\mathbf x}) }{\ubar{x}_{j_k}(t)}
\end{aligned}
\end{equation}
for $k=1,\ldots,n$, where the denominators are nonzero. This relationship indicates the equivalence of outward throughflow and subthroughflow intensities for the initial subsystems. The second and last equalities of Eq.~\ref{eq:thr_denseI} can be expressed in matrix form as
\begin{equation}
\label{eq:thr_dense2I}
\begin{aligned}
\mathcal{R}^{-1}(t,\bm{x}) & = \mathcal{T}(t,\bm{x}) \, \mathcal{X}^{-1}(t)  =
\hat{\mathcal{T}}_0(t,{\bf x}) \, \mathcal{X}_0^{-1} (t) \\
& = {\ubar{\mathcal{T}}}(t,{\bf x}) \, \ubar{\mathcal{X}}^{-1}(t)
= \hat{\ubar{\mathcal{T}}}_k(t,{\bf x}) \, \ubar{\mathcal{X}}_k^{-1}(t)
= \ubar{\mathcal{R}}^{-1}(t,\bm{x})
\end{aligned}
\end{equation}
where $ \mathcal{\hat T}_0 = \ubar{\mathcal{T}} := \mathcal{\ubar Y}(t,\bm{x}) + \diag \left ( \ubar{F}^T (t,\bm{x}) \, \mathbf{1} \right ) $. Equations~\ref{eq:thr_dense},~\ref{eq:ic_rates}, and~\ref{eq:thr_denseI} also imply the proportionality of the parallel initial subflows and the corresponding subthroughflows and substorages.
\begin{equation}
\label{eq:thr_denseI2}
\begin{aligned}
\frac{{\hat{\tau}}_{j_k}(t,{\mathbf x}) }{{\hat{\tau}}_{j_\ell}(t,{\mathbf x}) } = \frac{\ubar{\hat{\tau}}_{j_k}(t,{\mathbf x}) }{\ubar{\hat{\tau}}_{j_\ell}(t,{\mathbf x}) } = \frac{\ubar{x}_{j_k}(t) }{\ubar{x}_{j_\ell}(t)} = \frac{\ubar{f}_{i_k j_k}(t,{\mathbf x}) }{\ubar{f}_{i_\ell j_\ell}(t,{\mathbf x})}
\end{aligned}
\end{equation}
for $k,\ell=1,\ldots,n$, where the denominators are nonzero.

The $k^{th}$ {\em inward} and {\em outward subthroughflow matrices}, $\check{\mathcal{\ubar T}}_k(t,{\bf x}) \coloneqq \diag{ \left( \ubar{\check{\bm \tau}}_k (t,{\mathbf x}) \right)}$ and $\hat{\mathcal{\ubar T}}_k(t,{\bf x}) \coloneqq \diag{ \left( \ubar{\hat{\bm \tau}}_k (t,{\mathbf x}) \right)}$, for the $k^{th}$ initial subsystem can be expressed as
\begin{equation}
\label{eq:kth_thrmat}
\begin{aligned}
\check{\mathcal{\ubar T}}_k(t,{\bf x})  & = \diag{ \left( {F} (t,\bm{x}) \, \mathcal{X}^{-1}(t) \, \mathcal{\ubar X}_k(t)  \, \bm{1} \right) }  \\
\hat{\mathcal{\ubar T}}_k(t,{\bf x})  & = \mathcal{T}(t,\bm{x}) \, \mathcal{X}^{-1}(t)  \, \mathcal{\ubar X}_k (t)
\end{aligned}
\end{equation}
using Eq.~\ref{eq:matrix_subratesB}. The second relationship in Eq.~\ref{eq:kth_thrmat} can be used to reformulate ${\ubar F}_k (t,{\mathbf x})$, given in Eq.~\ref{eq:matrix_subratesM2}, in terms of the system flows only:
\begin{equation}
\label{eq:Tinoutb}
\begin{aligned}
\ubar{F}_k (t,{\bf x}) & = {F} (t,\bm{x}) \, \mathcal{T}^{-1}(t,\bm{x})  \, \hat{\mathcal{\ubar T}}_k(t,{\bf x})  .
\end{aligned}
\end{equation}

Equations~\ref{eq:TinoutA} and~\ref{eq:Tinoutb} imply that
\begin{equation}
\label{eq:intensity_flowIF}
\ubar{F}_k(t,{\mathbf x})  \, \ubar{\mathcal{\hat T}}_k^{-1}(t)  = \ubar{F}(t,\mathbf{x}) \, \ubar{\mathcal{T}}^{-1}(t) = F_0(t,{\mathbf x})  \, \mathcal{T}_0^{-1}(t) = F(t,{\bm x}) \, \mathcal{T}^{-1}(t)
\end{equation}
for $k=1,\ldots,n$. Equation~\ref{eq:intensity_flowIF} together with Eq.~\ref{eq:intensity_flowI} indicates the equivalence of the flow, initial flow, and initial subflow intensities both per unit throughflow and storage, respectively. Consequently, the flow and storage distribution matrices for both the subsystems and initial subsystems are the same. That is,
\begin{equation}
\label{eq:thr_QxQt}
\begin{aligned}
\ubar{Q}^{x}(t,{\mathbf  x}) & \coloneqq \ubar{F}(t,{\mathbf x}) \, \ubar{\mathcal{X}}^{-1}(t) = {Q}^{x}(t,{\bm x}) , \\
\ubar{Q}^{\tau}(t,{\mathbf  x}) & \coloneqq \ubar{F}(t,{\mathbf x}) \, \ubar{\mathcal{T}}^{-1}(t) = {Q}^{\tau}(t,{\bm x}) .
\end{aligned}
\end{equation}

We define the {\em inward} and {\em outward subthroughflow matrices}, $\ubar{\check{T}}(t,{\mathbf x})$ and $\ubar{\hat{T}}(t,{\mathbf x})$, for the initial subsystems as the matrices whose $k^{th}$ columns are the inward and outward initial subthroughflow vectors, $\ubar {\check{\bm \tau}}_k(t,{\mathbf x})$ and $\ubar {\hat{\bm \tau}}_k(t,{\mathbf x})$, $k=1,\ldots,n$, respectively:
\begin{equation}
\label{eq:matrix_subratesMB}
\begin{aligned}
\ubar{\check{T}} (t,{\mathbf x}) & \coloneqq \left(  \check{\ubar{\tau}}_{i_k}(t,{\mathbf x}) \right) = \left [ \ubar{\check{\bm \tau}}_1 (t,{\mathbf x}) \, \ldots \, \ubar{\check{\bm \tau}}_n (t,{\mathbf x}) \right ] ,
\\
\ubar{\hat{T}} (t,{\mathbf x}) & \coloneqq \left( \hat{\ubar{\tau}}_{i_k}(t,{\mathbf x}) \right) = \left [\ubar{\hat{\bm \tau}}_1 (t,{\mathbf x}) \, \ldots \, \ubar{\hat{\bm \tau}}_n (t,{\mathbf x}) \right ] .
\end{aligned}
\end{equation}
Using the relationships in Eq.~\ref{eq:matrix_subratesB}, these subthroughflow matrices for the initial subsystems can be expressed in matrix form as follows:
\begin{equation}
\label{eq:matrix_subrates2}
\begin{aligned}
\check{\ubar T} (t,{\mathbf x}) &
= \ubar{F} (t,\bm{x}) \,  {\mathcal{\ubar X}}^{-1} (t) \, {\ubar X}(t)
= {F} (t,\bm{x}) \,  {\mathcal{X}}^{-1} (t) \, {\ubar X}(t)  , \\
\hat{\ubar T} (t,{\mathbf x})
& = {\mathcal{\ubar T}}(t,\bm{x}) \, {\mathcal{\ubar X}}^{-1} (t) \, \ubar{X}(t)
= {\mathcal{T}}(t,\bm{x}) \, {\mathcal{X}}^{-1} (t) \, \ubar{X}(t) ,
\end{aligned}
\end{equation}
where the second equalities are due to Eqs.~\ref{eq:thr_QxQt} and~\ref{eq:thr_dense2I}. We then define the {\em net subthroughflow matrix}, $\ubar{T}(t,{\mathbf x})$, for the initial subsystems as
\begin{equation}
\label{eq:total_flowMI}
\begin{aligned}
\ubar{T}(t,{\mathbf x}) &\coloneqq \ubar{\check{T}} (t,{\mathbf x}) - \ubar{\check{T}} (t,{\mathbf x})
= A (t,\bm{x}) \, \ubar{X}(t) .
\end{aligned}
\end{equation}

The decomposition matrix $\ubar{D}({\mathbf x})$ can be expressed in terms of the subthroughflow functions instead of the substate functions:
\begin{equation}
\label{eq:apxmatrix_Pss0}
\begin{aligned}
\ubar{D}({\mathbf x}) = {\mathcal{\ubar X}}^{-1}(t) \, \ubar{X}(t) = {\mathcal{\ubar T}}^{-1}(t,\bm{x}) \, \ubar{\hat{T}} (t,{\mathbf x})
\end{aligned}
\end{equation}
due to Eq.~\ref{eq:matrix_subrates2}. Using the notations developed above, the initial subthroughflow matrices can be written in the following various forms:
\begin{equation}
\label{eq:matrix_subrates1_I2}
\begin{aligned}
\ubar{\check{T}} (t,{\mathbf x})
& = \ubar{F}(t,\bm{x}) \, \ubar{D}(\mathbf{x}) = Q^x(t,\bm{x}) \, \ubar{X}(t) = Q^\tau(t,\bm{x}) \, \ubar{\hat{T}}(t,\mathbf{x})  \\
\ubar{\hat{T}} (t,{\mathbf x}) & = \mathcal{R}^{-1}(t,\bm{x}) \, \ubar{X}(t) .
\end{aligned}
\end{equation}
The $k^{th}$ initial decomposition matrix, $\mathcal{\ubar D}_k({\mathbf x})$, can also be written as
\begin{equation}
\label{eq:matrix_Pss0KI}
\begin{aligned}
\mathcal{\ubar D}_k({\mathbf x}) = \mathcal{\ubar X}^{-1}(t) \, \mathcal{\ubar X}_k(t) = \mathcal{\ubar T}(t,\bm{x})^{-1} \, \hat{\mathcal{\ubar T}}_k (t,{\mathbf x}) ,
\end{aligned}
\end{equation}
using Eq.~\ref{eq:thr_dense2I}, similar to Eq.~\ref{eq:apxmatrix_Pss0}. The counterparts of the relationships in Eq.~\ref{eq:matrix_subrates1_I2} for the $k^{th}$ initial subsystem then become:
\begin{equation}
\label{eq:F_k2I}
\begin{aligned}
\ubar{F}_k(t,{\bf x}) & = \ubar{F}(t,{\bm x}) \, \mathcal{\ubar D}_k({\bf x}) =  Q^x(t,{\bm x}) \, \mathcal{\ubar X}_k(t)  = Q^\tau(t,{\bm x}) \, \hat{\mathcal{\ubar T}}_k(t,{\bf x})  \\
\hat{\mathcal{\ubar T}}_k(t,{\bf x}) & = \mathcal{R}^{-1}(t,{\bm x}) \, \mathcal{\ubar X}_k(t) .
\end{aligned}
\end{equation}

The flow and subthroughflow matrices for the initial subsystems can be integrated as follows:
\begin{equation}
\label{eq:tau_init}
\begin{aligned}
\ubar{F}(t,{\mathbf x}) \, \mathbf{1} & = \sum_{k=1}^n \ubar{F}_k (t,{\mathbf x}) \, \mathbf{1} = \sum_{k=1}^n \ubar{\check{\bm \tau}}_k (t,{\mathbf x})  = \ubar{\check{\bm \tau}} (t,{\mathbf x}) = \ubar{\check{T}} (t,{\mathbf x}) \, \mathbf{1}, \\
\ubar{\bm{y}}(t,{\mathbf x}) + \ubar{F}^T(t,{\mathbf x}) \, \mathbf{1} & = \sum_{k=1}^n \ubar{\mathbf y}_k (t,{\mathbf x}) + \ubar{F}^T_k (t,{\mathbf x}) \, \mathbf{1}  = \sum_{k=1}^n \ubar{\hat{\bm \tau}}_k (t,{\mathbf x}) = \ubar{\hat{\bm \tau}} (t,{\mathbf x})  \\ & = \ubar{\hat{T}} (t,{\mathbf x}) \, \mathbf{1} . \nonumber
\end{aligned}
\end{equation}

The governing equations for the initial subsystems can then be written in vector form as
\begin{equation}
\label{eq:initial_subsystemsV}
\begin{aligned}
\dot{\ubar{\mathbf x}}_k(t) = A(t,\bm{x}) \, \ubar{\mathbf{x}}_k(t) ,
\quad \ubar {\mathbf x}_k(t_0) = x_{k,0} \, \mathbf{e}_k
\end{aligned}
\end{equation}
for $k = 1,\ldots,n$. In matrix from, the governing equations for the decomposed system become
\begin{equation}
\label{eq:model_mat0}
\begin{aligned}
\dot{\ubar X}(t) = {\ubar T}(t,{\mathbf x}) = \ubar{\check{T}}(t,{\mathbf x}) - \ubar{\hat{T}}(t,{\mathbf x}), \quad \ubar{X}(t_0) = \mathcal{X}_0 .
\end{aligned}
\end{equation}
This system can also be expressed in terms of the flow intensity matrix, $A(t,\bm{x})$:
\begin{equation}
\label{eq:model_M2}
\begin{aligned}
\dot{\ubar X}(t) &
= A(t,\bm{x}) \, \ubar{X}(t) , \quad
{\ubar X}(t_0) = \mathcal{X}_0 .
\end{aligned}
\end{equation}

The system decomposition methodology that yields the governing equations for the initial subsystems in vector form, Eq~~\ref{eq:initial_subsystemsV}, or for the entire initial subsystem in matrix form, Eq.~\ref{eq:model_mat0}, can be schematized as follows:
\begin{center}
\begin{tikzpicture}
  \matrix (m) [matrix of math nodes,row sep=0em,column sep=7em,minimum width=1em]
  {
     \dot{\ubar{\bm x}}(t) = \ubar{\bm \tau}(t,{\mathbf x})
   &
     \hspace{2cm} \dot{\ubar{\mathbf x}}_k(t) = \ubar{\bm \tau}_k(t,{\mathbf x}), \, \, \, k=1,\ldots,n
   \\
 \veq & \veq \\
     {
 \begin{bmatrix}
  \dot{\ubar{x}}_{1} \\
  \dot{\ubar{x}}_{2} \\
  \vdots  \\
  \dot{\ubar{x}}_{n}
 \end{bmatrix}
=
 \begin{bmatrix}
  \ubar{\tau}_{1} \\
  \ubar{\tau}_{2} \\
  \vdots  \\
  \ubar{\tau}_{n}
 \end{bmatrix}
 }
 &
      {
 \begin{bmatrix}
  \dot{\ubar{x}}_{1_k} \\
  \dot{\ubar{x}}_{2_k} \\
  \vdots  \\
  \dot{\ubar{x}}_{n_k}
 \end{bmatrix}
=
 \begin{bmatrix}
  \ubar{\tau}_{1_k} \\
  \ubar{\tau}_{2_k} \\
  \vdots  \\
  \ubar{\tau}_{n_k}
 \end{bmatrix}
 }
\\
{  }
&
\hspace{-.5cm}
 {
 \begin{bmatrix}
  \dot{\ubar{x}}_{1_1} & \cdots & \dot{\ubar{x}}_{1_n} \\
  \dot{\ubar{x}}_{2_1} & \cdots & \dot{\ubar{x}}_{2_n} \\
  \vdots  & \ddots & \vdots  \\
  \dot{\ubar{x}}_{n_1} & \cdots & \dot{\ubar{x}}_{n_n}
 \end{bmatrix}
=
 \begin{bmatrix}
  \ubar{\tau}_{1_1} & \cdots & \ubar{\tau}_{1_n} \\
  \ubar{\tau}_{2_1} & \cdots & \ubar{\tau}_{2_n} \\
  \vdots & \ddots & \vdots  \\
  \ubar{\tau}_{n_1} & \cdots & \ubar{\tau}_{n_n}
 \end{bmatrix}
 }
 \\
  {} & {\hspace{-.5cm} } \veq  \\
   {} & \dot{\ubar{X}}(t) = \ubar{T}(t,{\mathbf x}) \\
};
  \path[-stealth, decorate]
    (m-3-1.east|-m-3-2) edge node [below] {system decomposition}
            node [above] {vector form} (m-3-2) ;
\hspace{-.5cm}
  \path[-stealth, decorate]
    (m-4-1.east|-m-4-2) edge node [below] {system decomposition}
            node [above] { matrix form } (m-4-2) ;
\end{tikzpicture}
\end{center}

\section{Analytic solution to linear systems}
\label{sec:linsys}

In this section, we formulate analytic solutions to linear systems with time-dependent inputs. Due to the cancellations in the formulation of $A(t,{\bm x})$ defined in Eq.~\ref{eq:matrix_A}, if the original system, Eq.~\ref{eq:model_orgV}, is linear, it can be expressed as
\begin{equation}
\label{eq:model_lin}
\begin{aligned}
\dot{\bm{x}}(t) = \bm{z}(t) + A (t) \, \bm{x}(t), \quad \bm{x}(t_0) = \bm{x}_0 .
\end{aligned}
\end{equation}
The system partitioning methodology yields a linear system, if the original system is linear. That is, the decomposed system, Eq.~\ref{eq:model_M}, takes the following linear form:
\begin{equation}
\label{eq:model_M2lin}
\begin{aligned}
\dot{X}(t) & = \mathcal{Z} (t) + A(t) \, X(t) ,
\quad X(t_0) = \mathbf{0}  , \\
\dot{\ubar X}(t) & = A(t) \, \ubar{X}(t) ,
\quad \quad \quad \quad \, \, {\ubar X}(t_0) = \mathcal{X}_0 . 
\end{aligned}
\end{equation}
Let $V(t)$ be the fundamental matrix solution of the system Eq.~\ref{eq:model_lin}, that is, the unique solution of the system
\begin{equation}
\label{eq:model_U}
\begin{aligned}
{\dot {V}}(t) & = A(t) \, {V}(t) , \quad  {V}(t_0) = I .
\end{aligned}
\end{equation}
The solutions for $X(t)$ and $\ubar X(t)$ in terms of $V(t)$ become
\begin{equation}
\label{eq:model_XXb_sln}
\begin{aligned}
X(t) & =  \int_{t_0}^{t}  {V}(t) \, {V}^{-1}(\tau) \, \mathcal{Z}(\tau) \, d\tau
\quad \mbox{and} \quad
\ubar X(t) = {V}(t) \, \mathcal{X}_0 .
\end{aligned}
\end{equation}
Therefore, we have the following observation.
\begin{remark}
\label{rem:linfun}
The initial substate matrix of the decomposed system, $\ubar X(t)$, scaled by the invertible matrix of the initial conditions, $\mathcal{X}_0>0$, is the {\em fundamental matrix solution} to the original linear system, Eq.~\ref{eq:model_lin}. That is, $V(t) = \ubar{X}(t) \, \mathcal{X}^{-1}_0$.
\end{remark}

The solution for $X(t)$ can then be expressed, in terms of this fundamental matrix solution, as
\begin{equation}
\label{eq:model_Xsln}
\begin{aligned}
X(t) & =  \int_{t_0}^{t}  \ubar{X}(t) \, \ubar{X}^{-1}(\tau) \, \mathcal{Z}(\tau) \, d\tau .
\end{aligned}
\end{equation}
For linear systems, Eq.~\ref{eq:model_MC} becomes
\begin{equation}
\label{eq:model_MClin}
\begin{aligned}
\dot{\sf X}(t) & = \mathcal{Z} (t) + A(t) \, {\sf X} (t) ,
\quad {\sf X}(t_0) = \mathcal{X}_0 .
\end{aligned}
\end{equation}
The solution for ${\sf X}(t)$ can be written, in terms of $\ubar{X}(t)$:
\begin{equation}
\label{eq:model_Xsln2}
\begin{aligned}
{\sf X}(t) = \ubar{X}(t) + X(t) =  \ubar{X}(t) + \int_{t_0}^{t}  \ubar{X}(t) \, \ubar{X}^{-1}(\tau) \, \mathcal{Z}(\tau) \, d\tau .
\end{aligned}
\end{equation}
Multiplying both sides by $\mathbf{1}$, we get
\begin{equation}
\label{eq:model_lin_sln}
\begin{aligned}
\bm{x} (t)
= \mathbf{x}_0(t) + \int_{t_0}^{t} \ubar{X}(t) \, \ubar{X}^{-1}(\tau) \, \bm{z}(\tau) \, d\tau
\end{aligned}
\end{equation}
which is the general solution to the original system, Eq.~\ref{eq:model_lin}.

For the special case of constant diagonalizable flow intensity matrix $A(t) = A$, the fundamental matrix solution can be written in the following form:
\begin{equation}
\label{eq:model_linfund_sln}
\begin{aligned}
{V} (t)  & = {\exp} \left( \int_{t_0}^t A \, ds \right ) = {\mathrm{e}}^{ \left(t-t_0 \right) \,A} = U \, {\mathrm{e}}^{(t-t_0) \, \Lambda} \, U^{-1}
\end{aligned}
\end{equation}
where $U$ is the matrix whose columns are the eigenvectors of A, and $\Lambda$ is the diagonal matrix whose diagonal elements are the corresponding eigenvalues of $A$. The solution to the matrix equation, Eq.~\ref{eq:model_MC}, given in Eq.~\ref{eq:model_Xsln2} then becomes
\begin{equation}
\label{eq:model_Xsln3}
\begin{aligned}
{\sf X}(t) = \ubar{X}(t) + X(t) = e^{ (t-t_0) A} \, \mathcal{X}_0 + \int_{t_0}^{t}  e^{ (t-\tau) A} \, \mathcal{Z}(\tau) \, d\tau .
\end{aligned}
\end{equation}
Consequently, Eq.~\ref{eq:model_lin_sln} takes the following form:
\begin{equation}
\label{eq:model_lin_sln2}
\begin{aligned}
\bm{x} (t) & = e^{ (t-t_0) \, A} \, \bm{x}_0 + \int_{t_0}^{t} e^{ (t-\tau) A} \, \bm{z}(\tau) \, d\tau .
\end{aligned}
\end{equation}

For linear systems, the governing equation for the subsystem decomposition, Eq.~\ref{eq:out_in_fs2}, can be solved explicitly for $x^w_{n_k \ell_k j_k}(t)$ as well. The solution becomes
\begin{equation}
\label{eq:sln_transient}
\begin{aligned}
x^w_{n_k \ell_k j_k}(t) = \int_{t_1}^t {\rm e}^{-\int_s^t r_\ell^{-1}(s',{\bm x}) \, ds'} \, f^w_{\ell_k j_k i_k}(s) \, ds
\end{aligned}
\end{equation}
where $r_\ell^{-1}(t,{\bm x}) = { {\hat{\tau}}_{\ell}(t,\bm{x}) } / { x_{\ell}(t) }$, as defined in Eq.~\ref{eq:thr_dense2}. Equation~\ref{eq:sln_transient} formulates the transient storage $x^w_{n_k \ell_k j_k}(t)$ generated by the transient inflow, $f^w_{\ell_k j_k i_k}(t)$, at subcompartment $\ell_k$ at any time $t \geq t_1$ (see Fig.~\ref{fig:subsystemp}).

\subsection{Static compartmental system analysis}
\label{sec:ss}

At steady state, the time derivatives of the state variables are zero. That is,
\begin{equation}
\label{eq:model_ss}
\begin{aligned}
\dot{X}(t) = \dot{\ubar{X}}(t) =\bm{0} .
\end{aligned}
\end{equation}
Clearly, if the decomposed system, Eq.~\ref{eq:model_M}, is at steady state, the original system, Eq.~\ref{eq:model_orgM2}, is also at steady state, due to Eq.~\ref{eq:xdervs}. Since $A$ is a strictly diagonally dominant constant matrix, it is invertible, and
\begin{equation}
\label{eq:model_ssS}
\begin{aligned}
{T} = \ubar{T} = \bm{0} \quad \Rightarrow \quad X = - A^{-1} \, \mathcal{Z}  \quad \mbox{and} \quad \ubar{X} = \mathbf{0}
\end{aligned}
\end{equation}
because of the relationships given in Eq.~\ref{eq:model_M}. The static systems can be decomposed and analyzed based on both external inputs and outputs. The proposed methodology for static systems in both the input- and output-orientations have recently been introduced and their duality is demonstrated by \cite{Coskun2017SCSA,Coskun2017SESM}.

\section{Subsystem decomposition methodology}
\label{apxsec:dsd}

The dynamic subsystem decomposition methodology is introduced in Section~\ref{sec:dsd}. The method formulation is based on the concept of directed subflow paths, which will be detailed below.

\subsection{Subflow paths}
\label{apxsec:subflow_paths}

A {\em link} will be defined as the connection between two system compartments that represents direct transactions between them. A link constitutes a {\em step} along a given flow path from one subcompartment to another. A {\em directed subflow path} in a subsystem will then be defined as a chain of connected links initiated at one subcompartment and ending in another of the same subsystem. The number of links or steps along a directed subflow path is called the {\em length} of the pathway. The {\em connection} of a subflow path to the ambient subsystem will be defined as the initial and only subcompartment on the path that receives inflow. The transient subflow and substorage computations along a subflow path must start at its connection. The link to the connection which represents the only inflow into the path will be called the {\em local input} and the source of this input will be called the {\em local source}. The outflow from the terminal subcompartment of the path will be called the {\em local output}. The environment can be taken as the local source or terminal subcompartment.

A subflow path $j_k \rightsquigarrow \ell_k \rightsquigarrow n_k \rightsquigarrow \cdots$ in subsystem $k$ with connection $j_k$ and the {\em local source} $i_k$ will be represented by $i_k \, {\color{red}{\mapsto}} \, j^{{\color{red}{*}}}_k \rightsquigarrow \ell_k \rightsquigarrow n_k \rightsquigarrow \cdots$. The connection is marked with red superscript $(*)$, and the local input with red arrow. A subflow link that does not directly contribute to the particular subflow or substorage in question will be represented by $(\rightsquigarrow)$ symbol. If the local input or output is the external input or output, the corresponding subcompartments will be denoted by $0_k$ on the pathway. Therefore, $0_k$ indicates that the local input is external input $z_k(t)$, and, since there is no external input for the initial subsystem, $0_0$ indicates that the local input is zero. Assuming that the terminal subcompartment of the path partially defined above is $v_k$, the complete subflow path will be represented by $p_{v_k j_k}=i_k \mapsto j^*_k \rightsquigarrow \ell_k \rightsquigarrow n_k \rightsquigarrow \cdots \to v_k$. If the initial and terminal subcompartments are the same, say $\ell_k$, a simpler notation will be used for the path $p_{\ell_k}$. For more than one path in the same subsystem, the path number will be represented by superscript $w$; $p^w_{v_k j_k}$. Having the connection and local source be the same, that is, having a path of type $i_k \mapsto i^*_k \rightsquigarrow \cdots$ implies that the local input will be taken as the subthroughflow into that subcompartment, $\check{\tau}_{i_k}(t, {\bf x})$, if not specified differently.

The subsystems can be decomposed into subflows and the associated substorages generated by these subflows along a set of mutually exclusive and exhaustive directed subflow paths. By {\em mutually exclusive} subflow paths, we mean that no given subflow path in a subsystem is a {\em subpath}, that is, completely inside of another path in the same subsystem. The {\em exhaustiveness}, in this context, means that such mutually exclusive subflow paths all together sum to the entire subsystem so partitioned. A subflow path that does not self intersect will be called the {\em linear path}, and the one with the same initial and terminal subcompartment will be called the {\em closed path}. A subflow path composed of linear and closed subpaths will be called the {\em mixed type path}. We will use the notation of $P_k$ for a set of mutually exclusive and exhaustive subflow paths in subsystem $k$, and $w_k$ for the number of paths in this set.

Each subsystem can be partitioned along a set of mutually exclusive and exhaustive subflow paths as follows: For subsystem $k$, subcompartment $k_k$ can be taken as the connection of all subflow paths with the local input being external input $z_k(t)$. The terminal subcompartments of linear and closed subpaths can be taken as the subcompartments with external output and the initial subcompartments, respectively. The cumulative transient outflow at the terminal subcompartment of a closed subpath will be considered as the local output. Consequently, the number of subflow paths in a subsystem obtained by this partitioning is equal to the number of local outputs. This subsystem partitioning will be called the {\em natural subsystem decomposition}. The natural subsystem partitioning of all subsystems yields a mutually exclusive and exhaustive partitioning of the entire system.

\subsection{Transient flows and storages}
\label{apxsec:transient}

The transient subflows and substorages are defined for linear subflow paths in Section~\ref{sec:dsd}. Additional relationships will be formulated in this section.

Self-intersecting directed flow paths are common within compartmental systems. The transient inflows and outflows and associated transient substorages change with each cycle along such paths. The cumulative values are obtained by summing up all transient subflows and substorages at the corresponding subcompartments with each cycle. The number of terms in these summations depends on the number of times the directed path pass through the corresponding subcompartments. In particular, transient subflows cycle along directed closed paths repeatedly and indefinitely. Therefore, in this case, the summations yield infinite series of functions that are convergent due to their construction.

The sum of the transient inflows from subcompartment $j_k$ to $\ell_k$ and the outflows from $\ell_k$ to $n_k$ generated at subcompartment $\ell_k$ at time $t$ by the local input into the connection of a given non-self-intersecting closed subflow path $p^w_{n_k j_k}$ during $[t_1,t]$, $t_1 \geq t_0$, will respectively be called the {\em inward} and {\em outward cumulative transient subflow} at subcompartment ${\ell_k}$ at time $t$. The associated storage generated by the inward cumulative transient subflow will be called {\em cumulative transient substorage}. These inward and outward cumulative transient subflows will be denoted by $\check{\tau}^{w}_{\ell_k}(t)$ and $\hat{\tau}^{w}_{\ell_k}(t)$, respectively, and associated cumulative transient substorage by ${x}^{w}_{\ell_k}(t)$. They can be formulated as
\begin{equation}
\label{eq:apxout_in_fs10}
\begin{aligned}
{x}^{w}_{\ell_k}(t) \coloneqq \sum_{m=1}^{m_w} {x}^{w,m}_{n_k \ell_k j_k}(t), \, \, \,
\check{\tau}^{w}_{\ell_k}(t) \coloneqq \sum_{m=1}^{m_w} {f}^{w,m}_{\ell_k j_k i_k}(t), \, \, \,
\hat{\tau}_{\ell_k}^{w}(t) \coloneqq \sum_{m=1}^{m_w} {f}^{w,m}_{n_k \ell_k j_k}(t) ,
\end{aligned}
\end{equation}
for $k=0,\ldots,n$, where the superscript $m$ represents the cycle number, and ${m_w}$ is the number of cycles, that is, the number of times the path $p^w_{n_k i_k}$ pass through subcompartment $\ell_k$. Large number of terms, ${m_w}$, in computation of these summations reduce truncation errors and, thus, improve the approximations. If the path is linear at ${\ell_k}$ ($m=1$), we write
\[ {x}^{w}_{\ell_k}(t) = {x}^w_{n_k \ell_k j_k}(t), \quad \check{\tau}^{w}_{\ell_k}(t) = {f}^w_{\ell_k j_k i_k}(t), \quad \mbox{and} \quad \hat{\tau}_{\ell_k}^{w}(t) = {f}^w_{n_k \ell_k j_k}(t) .\]

Let $P_k$ be the set of mutually exclusive and exhaustive subflow paths of the natural decomposition of subsystem $k$, and $w_k$ be the number of paths in this set. We also have
\begin{equation}
\label{eq:in_out_flows_new2}
x_{\ell_k}(t) = \sum^{w_k}_{w=1} x_{\ell_k}^{w}(t), \quad
{\check{\tau}}_{\ell_k}(t,{\bf x}) = \sum^{w_k}_{w=1} \check{\tau}^{w}_{\ell_k}(t),
\quad \mbox{and} \quad
{\hat{\tau}}_{\ell_k}(t,{\bf x}) = \sum^{w_k}_{w=1} \hat{\tau}^{w}_{\ell_k}(t) .
\nonumber
\end{equation}
That is, due to the mutual exclusiveness and exhaustiveness of the natural subsystem decomposition, the sum of the cumulative transient subflows and substorages are equal to the subthroughflows and substorages, respectively.

\subsection{Static subsystem decomposition}
\label{apxsec:ssd}

The static version of the dynamic subsystem decomposition introduced in Eqs.~\ref{eq:out_in_fs} and~\ref{eq:out_in_fs2} is formulated by \cite{Coskun2017SCSA}. Since the time derivatives of the state variables are zero at steady state, we set $\dot{x}^{w}_{n_k \ell_k j_k}(t) =0$ in Eq.~\ref{eq:out_in_fs2}. Then, the static transient outflow at subcompartment $\ell_k$, $f^{w}_{n_k \ell_k j_k}$, along subflow path $p^w_{n_k j_k}$ from $j_k$ to $n_k$, and the transient substorage, $x^{w}_{n_k \ell_k j_k}$, generated at $\ell_k$ by the transient inflow, $f^{w}_{\ell_k j_k i_k}$, are formulated as
\begin{equation}
\label{eq:out_in_fs2_SS}
\begin{aligned}
{x}^w_{n_k \ell_k j_k} =
 \frac{ x_{\ell} }{ \tau_{\ell} }  \,  f^w_{\ell_k j_k i_k}
\quad \mbox{and}
\quad f^w_{n_k \ell_k j_k}
= \frac{ f_{n \ell} }{x_{\ell} }  \, {x}^w_{n_k \ell_k j_k}
= \frac{ f_{n \ell}  }{ \tau_{\ell} }  \, f^w_{\ell_k j_k i_k}  .
\end{aligned}
\end{equation}

\section{The \texttt{diact} flows and storages}
\label{appsec:flows}

The \texttt{diact} flows and storages are introduced in Section~\ref{apxsec:flows} through the path-based approach based on the proposed dynamic subsystem decomposition methodology. The explicit formulations of the \texttt{diact} flows through the alternative dynamic approach introduced by \cite{Coskun2017DCSAM} are also presented in the same section. The detailed formulations of the \texttt{diact} transactions through the path-based approach are developed in this section. The \texttt{diact} flows and storages within the initial subsystems can be formulated similarly.
\begin{figure}[t]
\begin{center}
\begin{tikzpicture} [scale=.9]
   \draw[very thick, fill=blue!5, draw=black, text=blue] (-.05,-.05) rectangle node(R1) [pos=.5] {$ {x}_{i_k} $} (1.5,1.5) ;
   \draw[very thick, fill=blue!5, draw=black, text=blue] (8.2,-.05) rectangle node(R2) [pos=.5] { ${x}_{k_k}$ } (9.6,1.5) ;
    \draw[very thick,stealth-,draw=blue]  (9.7,1) -- (10.7,1) ;
    \node (z) [text=blue] at (10.3,1.3) {$z_k$};
    \draw[very thick,draw=blue,-stealth] (1.7,0.7) --  (8.1,0.7);
    \node[blue] (x) at (5,.3) {${f}_{k_k i_k}$};
   \node[text=black] at (5,1.1) { \small{direct (or cycling) subflow} };
   \node[text=black] at (5.2,1.9) {\small{indirect subflow} };
   \node[blue,anchor=east] at (6,2.4) (c) {${\tau}^\texttt{i}_{i_k k_k}$};
   \node[blue,anchor=west] at (.8,1.5) (d) {};
   \draw[very thick,draw=blue, dashed, -stealth]  (c) edge[out=180,in=90] (d);
   \node[blue,anchor=east] at (5.9,2.4) (h) {};
   \node[blue,anchor=west] at (8.8,1.5) (g) {};
   \draw[very thick,draw=blue, dashed]  (g) edge[out=90,in=0] (h);
\end{tikzpicture}
\end{center}
\caption{Schematic representation for the complementary nature of the simple indirect and cycling flows within the $k^{th}$ subsystem. The composite direct subflow, $f_{k_k i_k}(t,{\rm x})$, is represented by solid arrow. This subflow also contributes to the simple cycling flow at subcompartment $k_k$. The simple indirect subflow, ${\tau}^\texttt{i}_{i_k k_k}(t)$, is represented by dashed arrow.}
\label{fig:cycindices}
\end{figure}
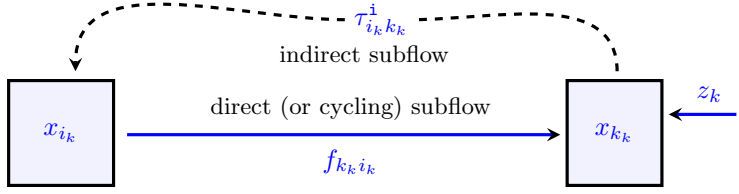

Let $P^\texttt{d}_{i_k j_k}$ and $P^\texttt{i}_{i_k j_k}$ be defined as the sets of mutually exclusive {\em direct} and {\em indirect subflow paths} $p^w_{i_k j_k}$ from subcompartment $j_k$, {\em directly} and {\em indirectly}, to $i_k$, respectively. The sets $P^\texttt{c}_{i_k j_k}$ and $P^\texttt{a}_{i_k j_k}$ are also defined as the sets of mutually exclusive {\em cyclic} and {\em acyclic subflow paths} $p^w_{i_k j_k}$ from $j_k$ to $i_k$ with a {\em closed} and {\em linear} subpath at terminal subcompartment $i_k$, respectively (see Fig.~\ref{fig:utilityfigs}). The cyclic subflow set, ${P}^{\texttt{c}}_{i_k}$, can alternatively be defined as the set of mutually exclusive subflow paths ${p}^w_{i_k}$ from subcompartment $i_k$ {\em indirectly} back to itself. The number of subflow paths in $P^\texttt{*}_{i_k j_k}$ will be denoted by $w_k$, where the superscript $(^\texttt{*})$ represent any of the \texttt{diact} symbols.

The simple and composite transfer flows, associated storages, and corresponding matrix measures are formulated in Section~\ref{apxsec:flows}, using the transfer subflow set, $P^\texttt{t}_{i_k j_k}$. All the other simple and composite $\texttt{diact}$ flows, associated storages, and matrix functions can then be formulated similarly by substituting the corresponding $\texttt{diact}$ flows and storages for their transfer counterparts in these equations and by using the corresponding $\texttt{diact}$ subflow sets instead. A compact derivation of the $\texttt{diact}$ flows and storages are presented below.

The {\em composite} \texttt{diact} {\em subflows} from subcompartment $j_k$ to $i_k$, $\tau^\texttt{*}_{i_k j_k}(t)$, can be expressed as the sum of the cumulative transient subflows, $\check{\tau}_{i_k}^{w}(t)$, generated by the outward subthroughflow at subcompartment $j_k$, $\hat{\tau}_{j_k}(t,{\bf x})$, during $[t_1,t]$, $t_1 \geq t_0$, and transmitted into $i_k$ at time $t$ along all subflow paths $p^w_{i_k j_k} \in P^\texttt{*}_{i_k j_k}$. The associated composite \texttt{diact} {\em substorage}, $x^\texttt{*}_{i_k j_k}(t)$, at subcompartment $i_k$ at time $t$ is the sum of the cumulative transient substorages, $x^{w}_{i_k}(t)$, generated by the cumulative transient inflows, $\check{\tau}_{i_k}^{w}(t)$, during $[t_1,t]$. Alternatively, $x^\texttt{*}_{i_k j_k}(t)$ can be defined as the storage segment generated by the composite \texttt{diact} subflow $\tau^\texttt{*}_{i_k j_k}(t)$ in subcompartment $i_k$ during $[t_1,t]$. The {\em simple} \texttt{diact} {\em subflows} can be defined similar to the their composite counterparts except for the following differences: the local inputs for the simple and composite \texttt{diact} subflows are $\hat{\tau}_{k_k}(t,{\bf x})$ and $\hat{\tau}_{j_k}(t,{\bf x})$, and the corresponding subflow sets are $P^\texttt{*}_{i_k k_k}$ and $P^\texttt{*}_{i_k j_k}$, respectively. Note that, for the cycling case, the first entrance of the transient subflows and substorages into $i_k$ are not considered as cycling subflows and substorages. Figure~\ref{fig:cycindices} depicts the complementary nature of the indirect and cycling subflows.

The composite \texttt{diact} subflows and substorages can then be formulated as
\begin{equation}
\label{eq:out_in_fsDT_diact}
\begin{aligned}
{\tau}^\texttt{*}_{i_k j_k}(t) \coloneqq
\sum_{w=1}^{w_k}  \check{\tau}_{i_k}^{w}(t)
\quad \mbox{and} \quad
x^\texttt{*}_{i_k j_k}(t) \coloneqq \sum_{w=1}^{w_k} x^{w}_{i_k}(t)
\end{aligned}
\end{equation}
where $w_k$ is the number of subflow paths $p^w_{i_k j_k} \in P^\texttt{*}_{i_k j_k}$. The sum of all the composite \texttt{diact} subflows and associated substorages from subcompartment $j_k$ to $i_k$ within each subsystem $k \neq 0$ will be called the {\em composite} \texttt{diact} {\em flow} and {\em storage} at time $t$, ${\tau}^\texttt{*}_{i j}(t)$ and $x^\texttt{*}_{i j}(t)$, from compartment $j$ to $i$:
\begin{equation}
\label{eq:out_in_diact}
\begin{aligned}
\tau^\texttt{*}_{i j}(t) \coloneqq \sum_{k=1}^{n} \tau^\texttt{*}_{i_k j_k}(t)
\quad & \mbox{and} \quad
x^\texttt{*}_{i j}(t) \coloneqq \sum_{k=1}^{n} x^\texttt{*}_{i_k j_k}(t) .
\end{aligned}
\end{equation}

For notational convenience, we define $n \times n$ matrix functions ${T}^{\texttt{*}}_k(t)$ and ${X}^{\texttt{*}}_k(t)$ whose $(i,j)-$elements are $\tau^\texttt{*}_{i_kj_k}(t)$ and $x^\texttt{*}_{i_kj_k}(t)$, respectively. That is,
\begin{equation}
\label{eq:fmT_diact}
\begin{aligned}
{T}^{\texttt{*}}_k(t) \coloneqq \left( \tau^\texttt{*}_{i_k j_k}(t) \right)
\quad \mbox{and} \quad
{X}^{\texttt{*}}_k(t) \coloneqq \left( x^\texttt{*}_{i_k j_k}(t) \right) ,
\end{aligned}
\end{equation}
for $k=0,\ldots,n$. These matrix measures ${T}^{\texttt{*}}_k(t)$ and ${X}^{\texttt{*}}_k(t)$ will be called the $k^{th}$ {\em composite} \texttt{diact} {\em subflow} and associated {\em substorage matrix} functions. The corresponding {\em composite} \texttt{diact} {\em flow} and associated {\em storage matrix} functions are then defined as ${T}^{\texttt{*}}(t) \coloneqq \left( \tau^\texttt{*}_{i j}(t) \right)$ and ${X}^{\texttt{*}}(t) \coloneqq \left( x^\texttt{*}_{i j}(t) \right)$, respectively.

The {\em simple} \texttt{diact} {\em flows} and {\em storages} can be formulated also in terms of their composite counterparts as follows:
\begin{equation}
\label{eq:simple_diact_in}
\begin{aligned}
{\tau}^\texttt{*}_{i_k}(t) = {\tau}^\texttt{*}_{i_k k_k}(t) \quad \mbox{and} \quad {x}^\texttt{*}_{i_k}(t) = {x}^\texttt{*}_{i_k k_k}(t) .
\end{aligned}
\end{equation}
To distinguish the composite and simple \texttt{diact} flow and storage matrices, we use a tilde notation over the simple versions. That is, the simple \texttt{diact} flow and storage matrices, for example, will be denoted by $\tilde{T}^{\texttt{*}}(t) \coloneqq \left( \tau^\texttt{*}_{i_k}(t) \right)$ and $\tilde{X}^{\texttt{*}}(t) \coloneqq \left( x^\texttt{*}_{i_k}(t) \right) $.

The {\em simple} \texttt{diact} {\em throughflow} and {\em compartmental storage matrices} and {\em vectors} can be formulated as
\begin{equation}
\label{eq:comp_diact}
\begin{aligned}
\tilde{\mathcal{T}}^\texttt{*}(t) \coloneqq \diag{( \tilde{T}^\texttt{*}(t) \, \bm{1} )} \, \, \, & \Rightarrow \, \, \,
\tilde{\tau}^\texttt{*}(t) \coloneqq \tilde{\mathcal{T}}^\texttt{*}(t) \, \bm{1}
\quad \mbox{and} \\
\tilde{\mathcal{X}}^\texttt{*}(t) \coloneqq \diag{( \tilde{X}^\texttt{*}(t) \, \bm{1} )} \, \, \, & \Rightarrow \, \, \,
\tilde{x}^\texttt{*}(t) \coloneqq \tilde{\mathcal{X}}^\texttt{*}(t) \, \bm{1} .
\end{aligned}
\end{equation}
The composite counterparts of these functions can similarly be formulated.

The difference between the composite and simple \texttt{diact} flows, $\tau^\texttt{*}_{ik}(t)$ and $\tau^\texttt{*}_{i_k}(t)$, and storages, $x^\texttt{*}_{ik}(t)$ and $x^\texttt{*}_{i_k}(t)$, is that the composite flow and storage from compartment $k$ to $i$ are generated by outward throughflow $\hat{\tau}_k(t,{\bm x}) - \hat{\tau}_{k_0}(t,{\bf x})$ derived from all external inputs and their simple counterparts from input-receiving subcompartment $k_k$ to $i_k$ are generated by outward subthroughflow $\hat{\tau}_{k_k}(t,{\bf x})$ derived from single external input $z_k(t)$ (see Fig.~\ref{fig:utilityfigs}). In that sense, the composite and simple \texttt{diact} flows and storages measure the influence of one compartment on another induced by all and a single external input, respectively.

\subsection{Static \texttt{diact} flows and storages}
\label{apxsec:sdfs}

The static $\texttt{diact}$ flows and storages are introduced by \cite{Coskun2017SCSA}, using the proportionality formulated in Eq.~\ref{eq:thr_dense3}, as listed in Table~\ref{tab:flow_stor}.
\begin{table}
     \centering
     \caption{The \texttt{diact} flow and storage distribution and the simple and composite \texttt{diact} (sub)flow and (sub)storage matrices. The superscript ($^\texttt{*}$) in each equation represents any of the \texttt{diact} symbols.}
     \label{tab:flow_stor}
     \begin{tabular}{c l l l l }
     \hline
\texttt{diact} & \multicolumn{2}{c} {flow and storage distribution matrices} & \multicolumn{1}{c} {flows} & \multicolumn{1}{c} {storages} \\
     \hline
     \noalign{\vskip 2pt}
\texttt{d} & $ N^\texttt{d} =  F \, \mathcal{T}^{-1}  $ & \multirow{5}{*}{ ${S}^\texttt{*} = \mathcal{R} \, {N}^\texttt{*}$ }  &
\multirowcell{5}{ $ {T}^\texttt{*} = {N}^\texttt{*} \, \mathcal{T} $ \\ $ {T}^\texttt{*}_\ell = {N}^\texttt{*} \, \mathcal{T}_\ell $ \\ $ \tilde{T}^\texttt{*} = {N}^\texttt{*} \, \mathsf{T} $ }
& \multirowcell{5}{ $ {X}^\texttt{*} = {S}^\texttt{*} \, \mathcal{T} $ \\ $ {X}^\texttt{*}_\ell = {S}^\texttt{*} \, \mathcal{T}_\ell $ \\ $ \tilde{X}^\texttt{*} = {S}^\texttt{*} \, \mathsf{T} $ } \\
\texttt{i} & $ N^\texttt{i} = \displaystyle (N-I) \, \mathcal{N}^{-1} - F \,\mathcal{T}^{-1}  $ & & & \\
\texttt{a} & $ {N}^\texttt{a} = \displaystyle (\mathcal{N}^{-1} \, N - I) \, \mathcal{N}^{-1} $ & & & \\
\texttt{c} & $ {N}^\texttt{c} = \displaystyle (N - \mathcal{N}^{-1} N) \, \mathcal{N}^{-1} $ & & & \\
\texttt{t} & $ N^\texttt{t} = \displaystyle (N-I) \, \mathcal{N}^{-1}   $ & & & \\
\noalign{\vskip 1pt}
\hline
     \end{tabular}
\end{table}
All quantities in the table are the static counterparts of their dynamic versions introduced in the present paper. For invertible matrix $\mathcal{Z}$, we also define
\[ N \coloneqq \check{T} \, \mathcal{Z}^{-1} = \hat{T} \, \mathcal{Z}^{-1}, \quad \mathcal{N} \coloneqq \operatorname{diag}(N), \quad \mbox{and} \quad \mathsf{T} \coloneqq \operatorname{diag}( \check{T} ) = \operatorname{diag}( \hat{T} ) . \]
See Case study~\ref{apxsec:tilly} for an application of these formulations to a static ecosystem model.

\section{Case studies}
\label{apxsec:ex}

A linear dynamic and a static model from ecosystem ecology are analyzed in this section. The linear dynamic model is solved analytically through the proposed dynamic methodology. The static model is used for an application of the static version of the proposed method.

\subsection{Case study}
\label{ex:hippe}

In this example, a linear dynamic model introduced by \cite{Hippe1983} is solved {\em analytically}. The model has two state variables, $x_1(t)$ and $x_2(t)$ (see Fig.~\ref{fig:hippe_diag}). The external inputs, $\bm{z}(t,\bm{x}) = [z_1(t,\bm{x}), z_2(t,\bm{x})]^T$, outputs, $\bm{y}(t,\bm{x}) = [y_1(t,\bm{x}), y_2(t,\bm{x})]^T$, and flow rate functions, $F(t,\bm{x})$, are given as
\begin{equation}
\label{eq:hippe_flows}
\begin{aligned}
F(t,\bm{x}) =
 \begin{bmatrix}
  0 & \frac{2}{3} \, x_2 \\
  \frac{4}{3} \, x_1 & 0 \\
 \end{bmatrix} ,
\quad
\bm{y}(t,\bm{x})=
\begin{bmatrix}
  \frac{1}{3} \, x_1 \\
  \frac{5}{3} \, x_2
\end{bmatrix},
\quad \mbox{and} \quad
\bm{z}(t,\bm{x})=
\begin{bmatrix}
  z_1 \\
  z_2
\end{bmatrix} .
 \end{aligned}
\end{equation}
The input, output, and state matrices become
\begin{equation}
\label{eq:hippe_flows2}
\begin{aligned}
\mathcal{Z}(t,\bm{x})=
\begin{bmatrix}
  z_1 & 0\\
  0 & z_2
\end{bmatrix},
\,
\mathcal{Y}(t,\bm{x})=
\begin{bmatrix}
  \frac{1}{3} \, x_1 & 0 \\
  0 & \frac{5}{3} \, x_2
\end{bmatrix},
\, \mbox{and} \, \, \,
\mathcal{X}(t,\bm{x}) =
 \begin{bmatrix}
  x_1 & 0 \\
  0 & x_2 \\
 \end{bmatrix}.
 \end{aligned}
\end{equation}

The system of governing equations, Eq.~\ref{eq:model2}, take the following form:
\begin{equation}
\label{eq:hippe_model}
\begin{aligned}
\dot x_1(t) & = z_1(t) + \frac{2}{3} x_2(t) - \left(\frac{4}{3} + \frac{1}{3} \right) \, x_1(t) \\
\dot x_2(t) & = z_2(t) + \frac{4}{3} x_1(t) - \left(\frac{2}{3} + \frac{5}{3} \right) \, x_2(t)
\end{aligned}
\end{equation}
with the initial conditions $\bm{x}_{0} = [x_{1.0}, x_{2,0}]^T = [3,3]^T$. In vector form, using the notation of Eq.~\ref{eq:model_orgM1}, the governing system can be expressed as
\begin{equation}
\label{eq:hippe_modelV}
\begin{aligned}
\dot{\bm{x}}(t) = \left ( \bm{z}(t,\bm{x}) + F(t,\bm{x}) \, \mathbf{1} \right ) - \left ( \bm{y}(t,\bm{x}) + F^T(t,\bm{x}) \, \mathbf{1} \right )
\end{aligned}
\end{equation}
with the initial conditions of $\bm{x}(t_0) = \bm{x}_0$. In matrix form, as given in Eq.~\ref{eq:model_orgV}, the governing system of equations becomes
\begin{equation}
\label{eq:hippe_mat_ss}
\dot{\bm{x}}(t) = \bm{z}(t) + A \, \bm{x}(t), \quad {\bm x}(t_0) = \bm{x}_{0}
\end{equation}
where $A$ is the constant flow intensity matrix given in Eq.~\ref{eq:matrix_A}:
\begin{equation}
\label{eq:matrix_As}
\begin{aligned}
A = \left ( F - \mathcal{T} \right ) \, \mathcal{X}^{-1} .
\end{aligned}
\end{equation}
It can be written explicitly as
\begin{equation}
\label{eq:hippe_matAS}
A = \begin{bmatrix}
  -(\frac{4}{3} + \frac{1}{3} ) & \frac{2}{3} \\
  \frac{4}{3} &  -(\frac{2}{3} + \frac{5}{3} )
\end{bmatrix}
=
\begin{bmatrix}
   - \frac{5}{3}  & \frac{2}{3}  \\
  \frac{4}{3} &  - \frac{7}{3}
\end{bmatrix} .
\end{equation}
\begin{figure}[t]
\begin{center}
\begin{tikzpicture}
\centering
   \draw[very thick,  fill=blue!5, draw=black] (-.05,-.05) rectangle node(R1) {$x_1(t)$} (1.5,1.5) ;
   \draw[very thick,  fill=blue!5, draw=black] (3.95,-.05) rectangle node(R2) {$x_2(t)$} (5.5,1.5) ;
       \draw[very thick,-stealth,draw=red, line width=5pt, opacity=.3]  (1.6,1) -- (3.8,1) -- (4.5,1) -- (4.5,.5) -- (1.6,.5) ;
       \node (p) [text=red] at (3,1.5) {$p^1_{1_1}$};
       \draw[very thick,stealth-,draw=black]  (1.6,.5) -- (3.8,.5) ;
       \draw[very thick,-stealth,draw=black]  (1.6,1) -- (3.8,1) ;
       \draw[very thick,stealth-,draw=black]  (5.6,.5) -- (6.6,.5) ;
       \draw[very thick,-stealth,draw=black]  (5.6,1) -- (6.6,1) ;
       \draw[very thick,stealth-,draw=black]  (-1.2,.5) -- (-0.15,.5) ;
       \draw[very thick,-stealth,draw=black]  (-1.2,1) -- (-0.15,1) ;
\end{tikzpicture}
\end{center}
\caption{Schematic representation of the model network. Subflow path $p^1_{1_1}$, along which the cycling flow and storage functions are computed, is red (subsystems are not shown)  (Case study~\ref{ex:hippe}).}
\label{fig:hippe_diag}
\end{figure}
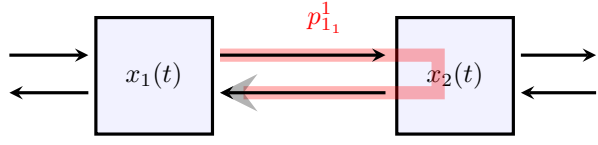

For the following state decomposition,
\begin{equation}
\label{eq:hippe_sc_exB}
\begin{aligned}
x_i(t) & = \sum_{k=0}^2 x_{i_k} (t)
\end{aligned}
\end{equation}
the substate and subflow rate functions become
\begin{equation}
\label{eq:hippe_flows_sc}
\begin{aligned}
X(t) & =
 \begin{bmatrix}
  x_{1_1} & x_{1_2} \\
  x_{2_1} & x_{2_2} \\
 \end{bmatrix} ,
\\
F_k(t,{\bf x}) & =
 \begin{bmatrix}
  f_{1_k1_k} & f_{1_k2_k} \\
  f_{2_k1_k} & f_{2_k2_k} \\
 \end{bmatrix}
=
 \begin{bmatrix}
  0 & \frac{2}{3} \, d_{2_k} \, x_{2} \\
  \frac{4}{3} \, d_{1_k} \, x_{1} & 0 \\
 \end{bmatrix}
=
 \begin{bmatrix}
  0 & \frac{2}{3} \, x_{2_k} \\
  \frac{4}{3} \, x_{1_k} & 0 \\
 \end{bmatrix},
\\
\bm{z}_k(t,{\bf x}) &=
\begin{bmatrix}
  z_{1_k} \\
  z_{2_k}
\end{bmatrix}
=
\begin{bmatrix}
  \delta_{1k} \, z_{1} \\
  \delta_{2k} \, z_{2}
\end{bmatrix}
=
\begin{bmatrix}
  \delta_{1k} \\
  \delta_{2k}
\end{bmatrix},
\\
\bm{y}_k(t,{\bf x}) & =
\begin{bmatrix}
  y_{1_k} \\
  y_{2_k}
\end{bmatrix}
=
\begin{bmatrix}
  \frac{1}{3} \, d_{1_k} \, x_{1} \\
  \frac{5}{3} \, d_{2_k} \, x_{2}
\end{bmatrix}
=
\begin{bmatrix}
  \frac{1}{3} \, x_{1,k} \\
  \frac{5}{3} \, x_{2,k}
\end{bmatrix} .
\\
 \end{aligned}
\end{equation}
The decomposed system, Eq.~\ref{eq:model_d1}, can be expressed as
\begin{equation}
\label{eq:hippe_sc_exA}
\begin{aligned}
\dot x_{1_k}(t) & = z_{1k}(t) + \frac{2}{3} x_{2_k}(t) - \left(\frac{4}{3} + \frac{1}{3} \right) \, x_{1_k}(t) \\
\dot x_{2_k}(t) & = z_{2k}(t) + \frac{4}{3} x_{1_k}(t) - \left(\frac{2}{3} + \frac{5}{3} \right) \, x_{2_k}(t)
\end{aligned}
\end{equation}
with the initial conditions of $x_{i_k} (t_0) = 0$ for $i,k = 0,1,2$.

Similarly, for the following decomposition of the initial subsystem ($k=0$),
\begin{equation}
\label{eq:hippe_sc_exSystem}
\begin{aligned}
\ubar{x}_{i}(t) & = \sum_{k=1}^2 \ubar{x}_{i_k} (t)
\end{aligned}
\end{equation}
the initial subflow rate functions become
\begin{equation}
\label{eq:hippe_flows_sc2}
\begin{aligned}
\ubar{X}(t) & =
 \begin{bmatrix}
  \ubar{x}_{1_1} & \ubar{x}_{1_2} \\
  \ubar{x}_{2_1} & \ubar{x}_{2_2} \\
 \end{bmatrix} ,
\\
\ubar{F}_{k}(t,{\bf x}) & =
 \begin{bmatrix}
  \ubar{f}_{1_k 1_k} & \ubar{f}_{1_k 2_k } \\
  \ubar{f}_{2_k 1_k} & \ubar{f}_{2_k 2_k } \\
 \end{bmatrix}
=
 \begin{bmatrix}
  0 & \frac{2}{3} \, \ubar{d}_{2_k} \, \ubar{x}_2 \\
  \frac{4}{3} \, \ubar{d}_{1_k} \, \ubar{x}_1 & 0 \\
 \end{bmatrix}
=
 \begin{bmatrix}
  0 & \frac{2}{3} \, \ubar{x}_{2_k} \\
  \frac{4}{3} \, \ubar{x}_{1_k} & 0 \\
 \end{bmatrix},
\\
\ubar{\bm{z}}_k(t,{\bf x}) & =
\begin{bmatrix}
  \ubar{z}_{1_k} \\
  \ubar{z}_{2_k}
\end{bmatrix}
=
\begin{bmatrix}
  0 \\
  0
\end{bmatrix},
\\
\ubar{\bm{y}}_k(t,{\bf x}) & =
\begin{bmatrix}
  \ubar{y}_{1_k} \\
  \ubar{y}_{2_k}
\end{bmatrix}
=
\begin{bmatrix}
  \frac{1}{3} \, \ubar{d}_{1_k} \, \ubar{x}_1 \\
  \frac{5}{3} \, \ubar{d}_{2_k} \, \ubar{x}_2
\end{bmatrix}
=
\begin{bmatrix}
  \frac{1}{3} \, \ubar{x}_{1,k} \\
  \frac{5}{3} \, \ubar{x}_{2,k}
\end{bmatrix} .
\\
 \end{aligned}
\end{equation}
The initial subsystem decomposition yields the following governing equations, as formulated in Eq.~\ref{eq:model_d2}:
\begin{equation}
\label{eq:hippe_sc_ex_model}
\begin{aligned}
\dot{\ubar x}_{1_k}(t) & = \frac{2}{3} \ubar{x}_{2_k}(t) - \left(\frac{4}{3} + \frac{1}{3} \right) \, \ubar{x}_{1_k}(t) \\
\dot{\ubar x}_{2_k}(t) & = \frac{4}{3} \ubar{x}_{1_k}(t) - \left(\frac{2}{3} + \frac{5}{3} \right) \, \ubar{x}_{2_k}(t)
\end{aligned}
\end{equation}
with the initial conditions of $\ubar{x}_{i_k} (t_0) = x_{i_{k,0}} (t_0) = 3 \, \delta_{ik}$ for $i,k = 1,2$. Thus, there are $2n^2 = 8$ equations in the decomposed system; $n^2=4$ of them are for the substates and the other $n^2$ equations are for the initial substates.

The governing equations for the decomposed system can be written in vector form, as given in Eq.~\ref{eq:model_V}:
\begin{equation}
\label{eq:hippe_sc_exV}
\begin{aligned}
\dot{\mathbf{x}}_k(t) & = {\mathbf z}_k + A \, {\mathbf x}_k(t), \quad
\mathbf{x}_k(t_0) = \mathbf{0} , \\
\dot{\ubar{\mathbf{x}}}_k(t) & = A \, \ubar{\mathbf{x}}_k(t) ,
\quad \quad \quad \, \, \ubar{\mathbf{x}}_k(t_0) = x_{k,0} \,  \mathbf{e}_k ,
\end{aligned}
\end{equation}
for $k=1,2$ or in matrix form, as given in Eq.~\ref{eq:model_M}:
\begin{equation}
\label{eq:hippe_sc_ex0}
\begin{aligned}
\dot X(t) & = \mathcal{Z} + A \, X(t), \quad
X(t_0) = \mathbf{0} , \\
\dot{\ubar X}(t) & =  A \, \ubar{X}(t) , \quad   \quad \quad
\ubar{X}(t_0) = \mathcal{X}_0 .
\end{aligned}
\end{equation}

The governing system, Eq.~\ref{eq:hippe_sc_ex0}, is linear. It can, therefore, be solved analytically as formulated in Section~\ref{sec:linsys}. Since the flow intensity matrix, $A$, is constant, the fundamental matrix solution becomes
\begin{equation}
\label{eq:fund_sln}
\begin{aligned}
V(t) = \left[
\begin{array}{cc}
\frac{2\,{\mathrm{e}}^{-t}}{3}+\frac{{\mathrm{e}}^{-3\,t}}{3} &
\frac{{\mathrm{e}}^{-t}}{3}-\frac{{\mathrm{e}}^{-3\,t}}{3} \\
\frac{2\,{\mathrm{e}}^{-t}}{3}-\frac{2\,{\mathrm{e}}^{-3\,t}}{3} &
\frac{{\mathrm{e}}^{-t}}{3}+\frac{2\,{\mathrm{e}}^{-3\,t}}{3}
\end{array}
\right] ,
\end{aligned}
\end{equation}
as formulated in Eq.~\ref{eq:model_linfund_sln}. The fundamental matrix solutions for the matrix equation, Eq.~\ref{eq:hippe_sc_ex0}, as formulated in Eq.~\ref{eq:model_XXb_sln}, then become
\begin{equation}
\label{eq:XXb_sln}
\begin{aligned}
\ubar{X}(t) & =
\left[
\begin{array}{cc}
2\,{\mathrm{e}}^{-t}+{\mathrm{e}}^{-3\,t} &
{\mathrm{e}}^{-t}-{\mathrm{e}}^{-3\,t} \\
2\,{\mathrm{e}}^{-t}-2\,{\mathrm{e}}^{-3\,t} &
{\mathrm{e}}^{-t}+2\,{\mathrm{e}}^{-3\,t}
\end{array}
\right] ,
\\
X(t) & =
\left[
\begin{array}{cc}
\frac{7}{9} - \frac{{\mathrm{e}}^{-3\,t}}{9}-\frac{2\,{\mathrm{e}}^{-t}}{3} &
\frac{2}{9} + \frac{{\mathrm{e}}^{-3\,t}}{9}-\frac{{\mathrm{e}}^{-t}}{3} \\
\frac{4}{9} + \frac{2\,{\mathrm{e}}^{-3\,t}}{9}-\frac{2\,{\mathrm{e}}^{-t}}{3} &
\frac{5}{9} - \frac{2\,{\mathrm{e}}^{-3\,t}}{9}-\frac{{\mathrm{e}}^{-t}}{3}
\end{array}
\right] .
\end{aligned}
\end{equation}
Therefore, the solution to the original system is
\begin{equation}
\label{eq:vec_sln}
\begin{aligned}
\bm{x}(t) &= {\sf X}(t) \, \bm{1} =
\left[
\begin{array}{c}
2\,{\mathrm{e}}^{-t}+1\\ 2\,{\mathrm{e}}^{-t}+1
\end{array}
\right]
\end{aligned}
\end{equation}
where
\begin{equation}
\label{eq:vec_sln2}
\begin{aligned}
{\sf X}(t) &= \ubar X(t) + X(t) =
\left[
\begin{array}{cc}
\frac{7}{9} + \frac{4\,{\mathrm{e}}^{-t}}{3}+\frac{8\,{\mathrm{e}}^{-3\,t}}{9} &
\frac{2}{9} + \frac{2\,{\mathrm{e}}^{-t}}{3}-\frac{8\,{\mathrm{e}}^{-3\,t}}{9} \\
\frac{4}{9} + \frac{4\,{\mathrm{e}}^{-t}}{3}-\frac{16\,{\mathrm{e}}^{-3\,t}}{9} &
\frac{5}{9} + \frac{2\,{\mathrm{e}}^{-t}}{3}+\frac{16\,{\mathrm{e}}^{-3\,t}}{9}
\end{array}\right] ,
\end{aligned}
\end{equation}
as given in Eqs.~\ref{eq:model_Xsln3} and~\ref{eq:model_lin_sln2}.

The steady state solutions, as formulated in Eq.~\ref{eq:model_ssS}, also become
\begin{equation}
\label{eq:hippe_sc_ex3}
\begin{aligned}
X & = - A^{-1} \, \mathcal{Z} = \mathcal{X} \, \left ( \mathcal{T} - F \right )^{-1} \, \mathcal{Z} =
\left[
\begin{array}{cc}
\frac{7}{9} & \frac{2}{9} \\
\frac{4}{9} & \frac{5}{9}
\end{array}
\right]
\quad \mbox{and} \quad \ubar{X} = \bm{0}.
\end{aligned}
\end{equation}
It can easily be seen that the dynamic solution, Eq.~\ref{eq:XXb_sln}, converges to this steady state solution as $t \rightarrow \infty$.

We also analyze the system with a time dependent input $z(t) = [3+\sin(t),3+\sin(2 t)]^T$. The fundamental matrix (initial substate matrix), $\ubar X(t)$, is the same as the one given in Eq.~\ref{eq:XXb_sln} for the constant input case above. Similar computations lead us to the following initial substate vector and substate matrix components:
\begin{equation}
\label{eq:hippe_pedX}
\begin{aligned}
x_{1_0}(t) & = x_{2_0}(t) =  3\,{\mathrm{e}}^{-t} , \\
x_{1_1}(t) & = \frac{7}{3} - \frac{11\,\cos\left(t\right)}{30} + \frac{13\,\sin\left(t\right)}{30} -\frac{5\,{\mathrm{e}}^{-t}}{3} - \frac{3\,{\mathrm{e}}^{-3\,t}}{10} , \\
x_{1_2}(t) & = \frac{2}{3} - \frac{16\,\cos\left(2\,t\right)}{195} - \frac{2\,\sin\left(2\,t\right)}{195} - \frac{13\,{\mathrm{e}}^{-t}}{15} + \frac{11\,{\mathrm{e}}^{-3\,t}}{39} , \\
x_{2_1}(t) & = \frac{4}{3} -\frac{4\,\cos\left(t\right)}{15} + \frac{2\,\sin\left(t\right)}{15} -\frac{5\,{\mathrm{e}}^{-t}}{3} + \frac{3\,{\mathrm{e}}^{-3\,t}}{5} , \\
x_{2_2}(t) & = \frac{5}{3} -\frac{46\,\cos\left(2\,t\right)}{195} + \frac{43\,\sin\left(2\,t\right)}{195} - \frac{13\,{\mathrm{e}}^{-t}}{15} - \frac{22\,{\mathrm{e}}^{-3\,t}}{39} .
\end{aligned}
\end{equation}
The solutions to the system, given in Eq.~\ref{eq:XXb_sln}, for the constant external input case are the same as the ones given by \cite{Hippe1983}. The authors, however, did not provide an explicit solution for the time-dependent input case for a comparison. Asymptotically, the solution for the substate matrix becomes
\begin{equation}
\label{eq:hippe_ped_ss}
\begin{aligned}
\lim_{t \rightarrow \infty} X(t) =
\left[
\begin{array}{cc}
\frac{7}{3} - \frac{11\,\cos\left(t\right)}{30} + \frac{13\,\sin\left(t\right)}{30} &
\frac{2}{3} - \frac{16\,\cos\left(2\,t\right)}{195} - \frac{2\,\sin\left(2\,t\right)}{195} \\
\frac{4}{3} -\frac{4\,\cos\left(t\right)}{15} + \frac{2\,\sin\left(t\right)}{15} &
\frac{5}{3} -\frac{46\,\cos\left(2\,t\right)}{195} + \frac{43\,\sin\left(2\,t\right)}{195}
\end{array}
\right] .
\end{aligned}
\end{equation}

Similarly, the elements of the inward initial subthroughflow vector and inward subthroughflow matrix are
\begin{equation}
\label{eq:hippe_thr}
\begin{aligned}
\check{\tau}_{1_0}(t) & = 2\,{\mathrm{e}}^{-t}, \quad \check{\tau}_{2_0}(t) = 4\,{\mathrm{e}}^{-t} , \\
\check{\tau}_{1_1}(t) & = \frac{35}{9} -\frac{8\,\cos\left(t\right)}{45}+\frac{49\,\sin\left(t\right)}{45} -\frac{10\,{\mathrm{e}}^{-t}}{9} + \frac{2\,{\mathrm{e}}^{-3\,t}}{5} , \\
\check{\tau}_{1_2}(t) & = \frac{742}{585} - \frac{184\, \cos^2\left(t\right)}{585} +\frac{86\,\sin\left(2\,t\right)}{585} -\frac{26\,{\mathrm{e}}^{-t}}{45}-\frac{44\,{\mathrm{e}}^{-3\,t}}{117} , \\
\check{\tau}_{2_1}(t) & = \frac{28}{9} -\frac{22\,\cos\left(t\right)}{45} + \frac{26\,\sin\left(t\right)}{45}  -\frac{20\,{\mathrm{e}}^{-t}}{9} - \frac{2\,{\mathrm{e}}^{-3\,t}}{5} , \\
\check{\tau}_{2_2}(t) & = \frac{2339}{585} -\frac{128\, \cos^2\left(t\right)}{585} +\frac{577\,\sin\left(2\,t\right)}{585} -\frac{52\,{\mathrm{e}}^{-t}}{45}+\frac{44\,{\mathrm{e}}^{-3\,t}}{117} .
\end{aligned}
\end{equation}
The graphical representations of the substate and the inward subthroughflow matrices, $X(t)$ and $\check{T}(t)$, given in Eqs.~\ref{eq:hippe_pedX} and~\ref{eq:hippe_thr}, are depicted in Fig.~\ref{fig:hippe}.
\begin{figure}[t]
\begin{center}
\includegraphics[width=.32\textwidth]{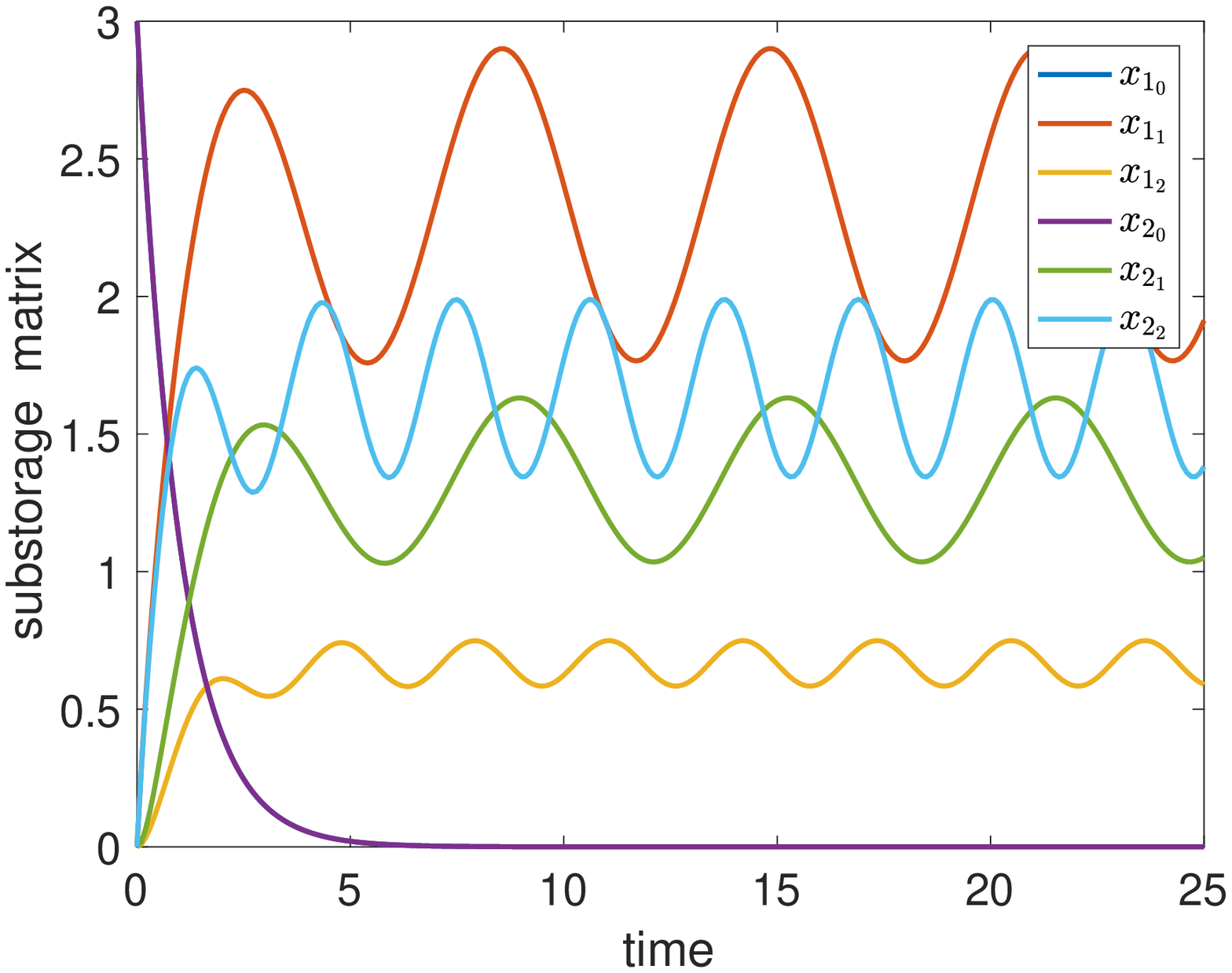}
\includegraphics[width=.31\textwidth]{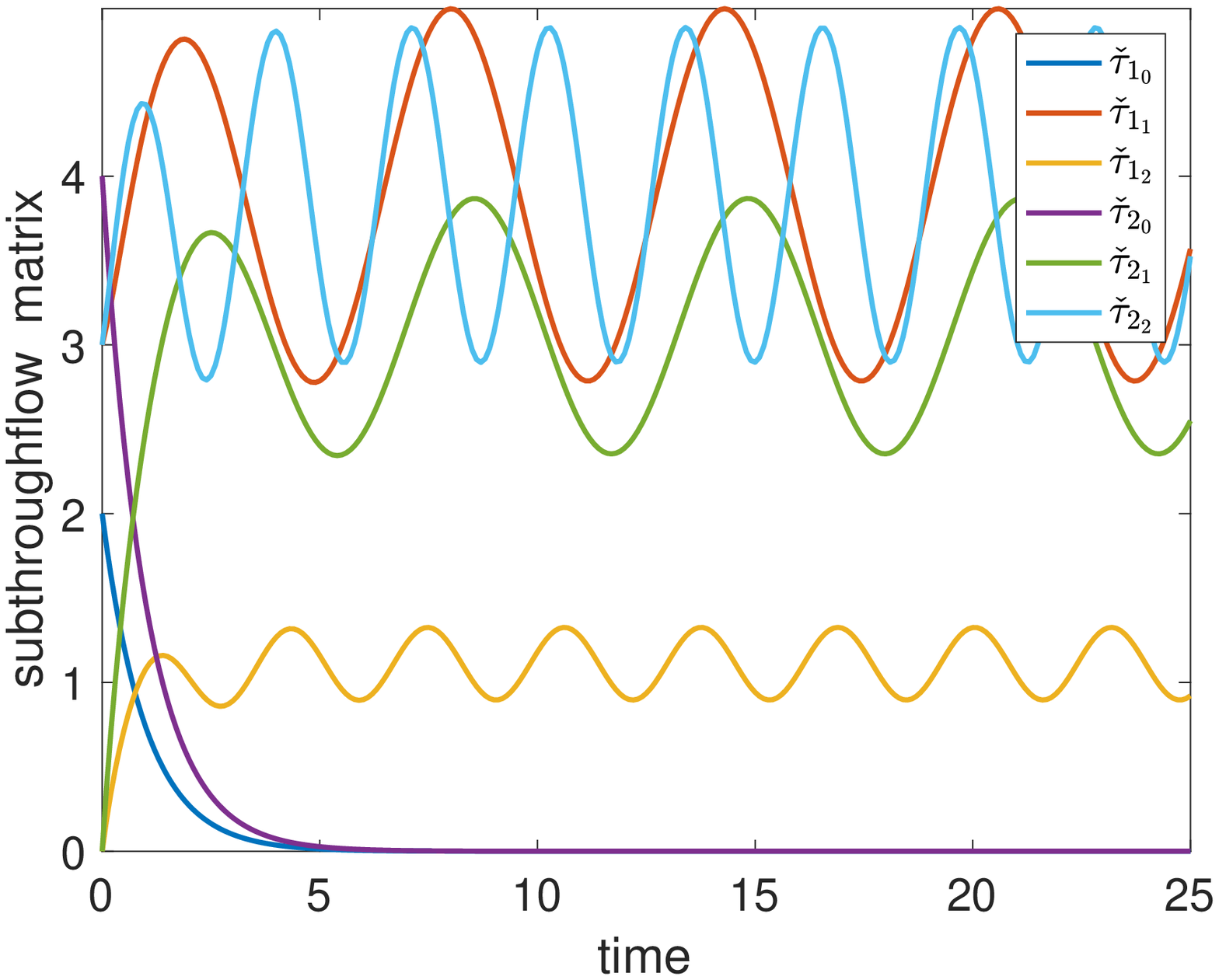}
\includegraphics[width=.32\textwidth]{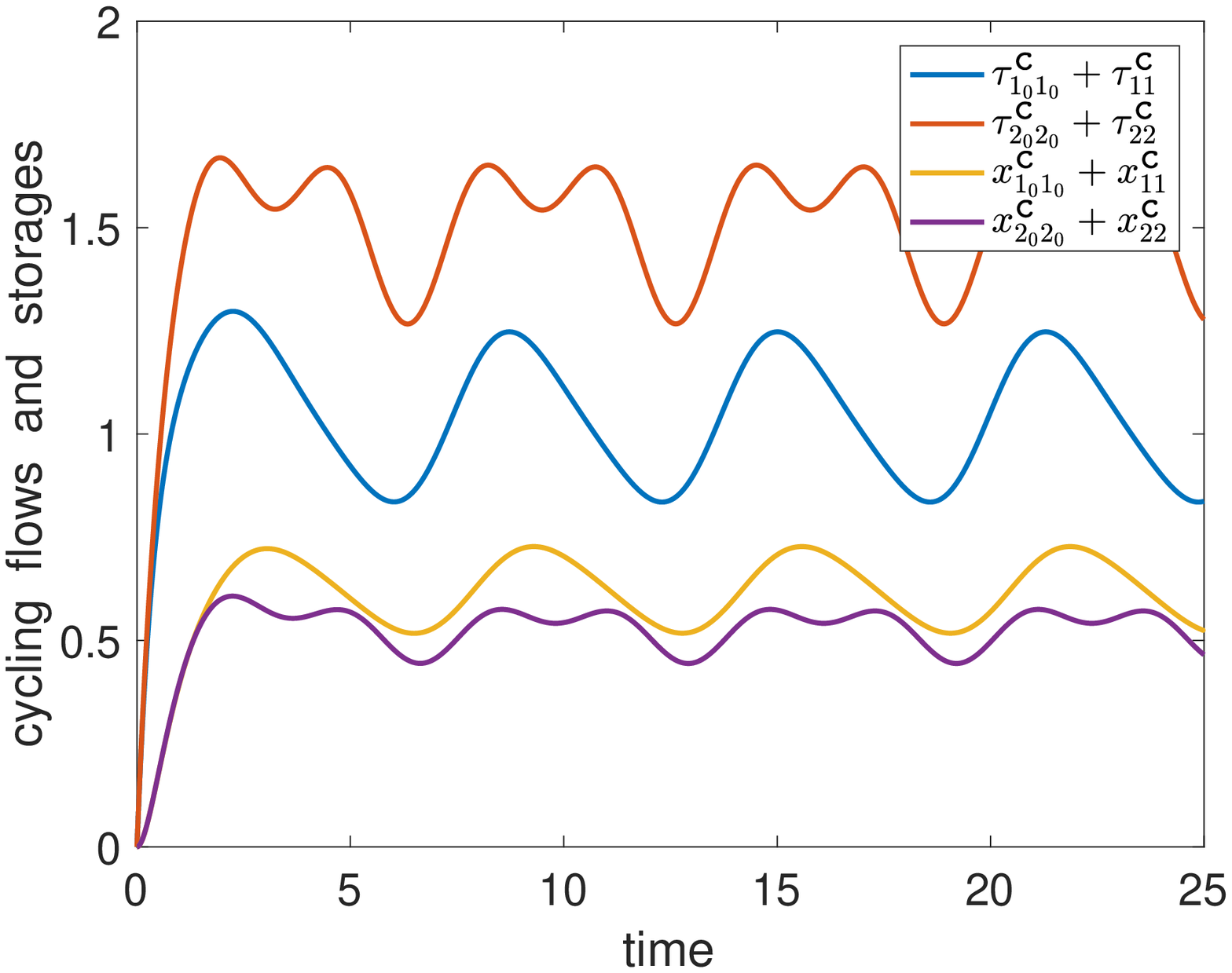}
\end{center}
\caption{The graphical representation of the substorage and inward subthroughflow matrices, $X(t)$ and $\check{T}(t)$, for time dependent input $z(t) = \left[ 3+\operatorname{sin}(t), 3+\operatorname{sin}(2 \, t) \right]^T$, and composite cycling flows and storages, $\tau^c_{i_0 i_0}(t)+\tau^c_i(t)$ and $x^c_{i_0 i_0}(t)+x^c_i(t)$ (Case study~\ref{ex:hippe}).}
\label{fig:hippe}
\end{figure}
Using these matrices, the dynamic distribution of the external inputs as inward throughflows and the organization of the associated storages generated by the inputs within the system can be analyzed individually and separately.

The cycling flows and the associated storages generated by these flows are also calculated below, as an application of the subsystem partitioning methodology. The sets of mutually exclusive subflow paths from subcompartment $k_k$ to $1_k$ with a closed subpath at $1_k$, $P^\texttt{c}_{1_k k_k}$, are given as $P^\texttt{c}_{1_1 1_1} = \{ p^1_{1_1 1_1} \}$ and $P^\texttt{c}_{1_2 2_2} = \{ p^1_{1_2 2_2}\}$, where $p_{1_1 1_1}^1 \coloneqq  0_1 \mapsto 1_1 \rightsquigarrow 2_1 \to 1_1$ and $p_{1_2 2_2}^1 \coloneqq 0_2 \mapsto 2_2 \rightsquigarrow 1_2 \rightsquigarrow 2_2 \to 1_2$. The set of mutually exclusive subflow paths within the initial subsystem with a closed subpath at $1_0$ is $P^\texttt{c}_{1_0 0_0} = \{p^1_{1_0 1_0} , p^2_{1_0 2_0}\}$ where $p_{1_0 1_0}^1 \coloneqq 0_0 \mapsto 1_0 \rightsquigarrow 2_0 \to 1_0$, $p_{1_0 2_0}^2 \coloneqq 0_0 \mapsto 2_0 \rightsquigarrow 1_0 \rightsquigarrow 2_0 \to 1_0 $. For the subflow paths in $P^\texttt{c}_{1_0 0_0}$, the composite cycling subflows are derived from the initial stocks, and for the ones in $P^\texttt{c}_{1_1 1_1}$ and $P^\texttt{c}_{1_2 2_2}$, the simple cycling flows are generated by the respective environmental inputs of $z_1(t)$ and $z_2(t)$. The sets of subflow paths for subsystem $2$, $P^\texttt{c}_{2_k k_k}$, $k=0,1,2$, can similarly be defined.

The simple cycling subflow at subcompartment $1_2$ along the only subflow path ($w_2 = 1$) in subsystem $2$, $p^1_{1_2 2_2} \in P^\texttt{c}_{1_2 2_2}$, and associated substorage are
\begin{equation}
\label{eq:cyc1}
\begin{aligned}
{\tau}^{c}_{1_2}(t) = \sum_{w=1}^1 \check{\tau}^{w}_{1_2}(t) = \check{\tau}^{1}_{1_2}(t) \quad & \mbox{and} \quad x^c_{1_2}(t) = \sum_{w=1}^1 x^{w}_{1_2}(t) = x^{1}_{1_2}(t) ,
\end{aligned}
\nonumber
\end{equation}
as formulated in Eq.~\ref{eq:out_in_fsDT_diact}. The links contributing to the cycling flow along the path are marked with red cycle numbers in the extended subflow diagram below:
\begin{equation}
\begin{aligned}
p^1_{1_2 2_2} &= 0_2 \mapsto 2_2 \, \rightsquigarrow \, 1_2 \rightsquigarrow 2_2  \, \xrightarrow{ {\color{red} 1 } } \,  1_2 \rightsquigarrow 2_2  \, \xrightarrow{ {\color{red} 2 } } \,  1_2 \rightsquigarrow   \cdots
\end{aligned}
\nonumber
\end{equation}
Note that the first flow entrance into $1_2$ is not considered as cycling flow. The cumulative transient inflow $\check{\tau}^1_{1_2}(t)$ and substorage $x^1_{1_2}(t)$ can be approximated by two terms ($m_1=2$) along the closed subpath as formulated in Eq.~\ref{eq:apxout_in_fs10}:
\begin{equation}
\label{eq:hippe_ss10}
\begin{aligned}
{x}^{1}_{1_2}(t) & = \sum_{m=1}^2 {x}^{1,m}_{2_2 1_2 2_2}(t) \approx {x}^{1,1}_{2_2 1_2 2_2}(t) + {x}^{1,2}_{2_2 1_2 2_2}(t) , \\
\check{\tau}^{1}_{1_2}(t) & =  \sum_{m=1}^2 {f}^{1,m}_{1_2 2_2 1_2}(t) \approx  f^{1,1}_{1_2 2_2 1_2}(t) + f^{1,2}_{1_2 2_2 1_2}(t) .
\end{aligned}
\nonumber
\end{equation}
The governing equations for the transient subflows and associated substorages, $f^{w,m}_{1_2 2_2 1_2}(t)$ and $x^{w,m}_{2_2 1_2 2_2}(t) $, as well as the other transient subflows and substorages involved in Eq.~\ref{eq:cyc2_sln}, as formulated in Eqs.~\ref{eq:out_in_fs} and~\ref{eq:out_in_fs2}, are coupled and solved simultaneously together with the decomposed system, Eqs.~\ref{eq:model_d1} and~\ref{eq:model_d2}. The numerical results for the composite cycling flow and associated storage functions
\begin{equation}
\label{eq:cyc2_sln}
\begin{aligned}
{\tau}^\texttt{c}_{i_0 i_0}(t) + {\tau}^\texttt{c}_{ii}(t) = \sum_{k=0}^2 {\tau}^\texttt{c}_{i_k i_k}(t)  \quad & \mbox{and} \quad x^\texttt{c}_{i_0 i_0}(t) + x^\texttt{c}_{ii}(t) = \sum_{k=0}^2 x^\texttt{c}_{i_k i_k}(t)
\end{aligned}
\end{equation}
for $i=1,2$, are presented in Fig.~\ref{fig:hippe}.

Note that, due to the reflexivity of cycling flows, the same computations can be done more practically in only two steps using the sets of closed subflow paths, $P^\texttt{c}_{i_k}$, instead (for example, along $p_{1_1}^1 \coloneqq  1_1 \mapsto 1_1 \rightsquigarrow 2_1 \to 1_1$ in subcompartment $1_1$, see Fig.~\ref{fig:hippe_diag}), with the local inputs being the corresponding outwards subthroughflows, $\hat{\tau}_{k_k}(t,{\mathbf x})$. The cycling flows can also be computed along closed paths at the compartmental level as well, where the local inputs become outward throughflows, $\hat{\tau}_{k}(t,{\bm x})$.

The cycling subflows can also be computed analytically through the dynamic approach as formulated in Eq.~\ref{eq:comp_diact_subs} \cite{Coskun2017DCSAM}. For example, the composite cycling flow at initial subcompartment $1_0$, $\tau^{\texttt{c}}_{1_0 1_0}(t)$, becomes
\begin{equation}
\label{eq:cyc_analytic}
\begin{aligned}
\tau^{\texttt{c}}_{1_0 1_0}(t) = -\frac{ 36 \, {\mathrm{e}}^{-t} + 80\,{\mathrm{e}}^{2\,t}-100\,{\mathrm{e}}^{1\,t}-16\,{\mathrm{e}}^{2\,t}\,\cos\left(t\right)+8\,{\mathrm{e}}^{2\,t}\,\sin\left(t\right) }
{9 + 50\,{\mathrm{e}}^{2\,t}-70\,{\mathrm{e}}^{3\,t}+11\,{\mathrm{e}}^{3\,t}\,\cos\left(t\right)-13\,{\mathrm{e}}^{3\,t}\,\sin\left(t\right) } .
\end{aligned}
\end{equation}
The cycling substorages can then be obtained by coupling Eq.~\ref{eq:out_in_diact2}, for the cycling flows and storages, with the decomposed system, Eqs.~\ref{eq:model_d1} and~\ref{eq:model_d2}, and solving them simultaneously. Since the system is linear, the cycling substorages can, alternatively, be obtained analytically as formulated in Eq.~\ref{eq:sln_transient}. Because of the lengthy analytical formulations of the other cycling subflows and substorages, only $\tau^{\texttt{c}}_{1_0}(t)$ is presented in Eq.~\ref{eq:cyc_analytic} as an example.

\subsection{Case study}
\label{apxsec:tilly}

To demonstrate an application of the static version of the proposed dynamic methodology, a commonly studied ecosystem network proposed by \cite{Tilly1968} is used as an example in this section. This ecosystem has already been analyzed in detail by \cite{Coskun2017SCSA,Coskun2017SESM}.

Cone Spring is a small, shallow spring-brook located in Louisa County, Iowa. The study area consists of $116 \, m^2$. The network has $5$ compartments representing $1-$plants, $2-$detritus, $3-$bacteria, $4-$detritus feeders, and $5-$carnivores. These compartments are connected by the transaction of energy between them. The conserved quantity needs to be investigated within the system is energy. The system flow information is given as follows:
\begin{equation}
\label{eq:diact2}
\begin{aligned}
F =
\left[
\begin{array}{ccccc}
0 & 0 & 0 & 0 & 0\\ 8881 & 0 & 1600 & 200 & 167\\ 0 & 5205 & 0 & 0 & 0\\ 0 & 2309 & 75 & 0 & 0\\ 0 & 0 & 0 & 370 & 0
\end{array}
\right] ,
\quad
y = \left[
\begin{array}{c}
2303\\ 3969\\ 3530\\ 1814\\ 203
\end{array}
\right] ,
\quad
z =
\left[
\begin{array}{c}
11184\\ 635\\ 0\\ 0\\ 0
\end{array}
\right] .
\end{aligned}
\nonumber
\end{equation}
The unit for energy flows and storages are kkal m$^{-2}$ y$^{-1}$ and kkal m$^{-2}$, respectively.

The system decomposition methodology yields the subthroughflow and substorage matrices:
\begin{equation}
\label{appex:tr_fl_1}
T = \left[
\begin{array}{ccccc}
11184 & 0 & 0 & 0 & 0 \\
    10717 & 766 & 0 & 0 & 0 \\
    4858 & 347 & 0 & 0 & 0 \\
    2225 & 159 & 0 & 0 & 0 \\
    345 & 25 & 0 & 0 & 0
\end{array}
\right]
\quad \mbox{and} \quad
X = \left[
\begin{array}{ccccc}
    285 & 0 & 0 & 0 & 0 \\
    3340.5 & 238.9 & 0 & 0 & 0 \\
    108.8 & 7.8 & 0 & 0 & 0 \\
    56.0 & 4.0 & 0 & 0 & 0 \\
    15.9 & 1.1 & 0 & 0 & 0
\end{array}
\right] .
\nonumber
\end{equation}

The subsystem decomposition methodology then yields the $\texttt{diact}$ flows and storages as given in Table~\ref{tab:flow_stor}. The direct flow matrix is $T^\texttt{d} = F$. The composite indirect flow matrix becomes
\begin{equation}
\label{appex:tr_fl_2}
T^{\texttt{i}} = \left[
\begin{array}{ccccc}
         0      &   0     &    0    &     0  &   0 \\
    1835.74  &  1967.00  &  63.02  &  226.41  &  31.04 \\
    4857.67   &      0  &  753.81  &  193.28   &    89.76 \\
    2224.91  &  75.00  &  334.40  &  88.52   &    41.11 \\
    345.31  &  370.00  &  63.53   &      0   &    6.38
\end{array}
\right] .
\nonumber
\end{equation}
Note the $\tau^\texttt{i}_{32} = \tau^\texttt{i}_{54} = 0$. This is because of the fact that there is no indirect flow from compartment $2$ to $3$ or from $4$ to $5$ (see Fig.~\ref{fig:cone_sp}). There is no indirect flow to compartment $1$ from any other compartments either, so the first row vector of $T^\texttt{i}$ is zero.
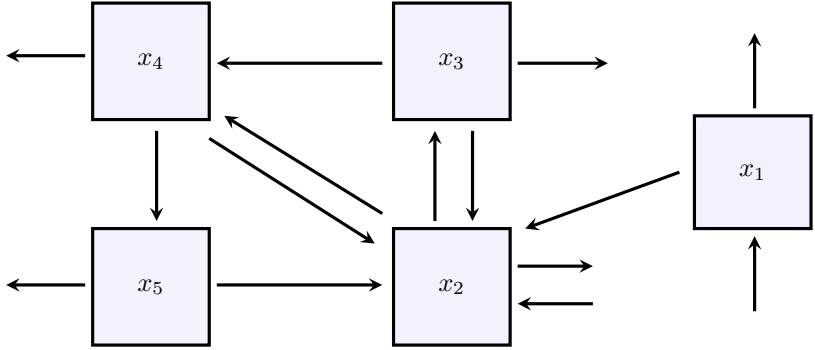
\begin{figure}[t]
\begin{center}
\begin{tikzpicture}
\centering
   \draw[very thick,  fill=blue!5, draw=black] (-.05,-.05) rectangle node(R1) {$x_5$} (1.5,1.5) ;
   \draw[very thick,  fill=blue!5, draw=black] (3.95,-.05) rectangle node(R2) {$x_2$} (5.5,1.5) ;
   \draw[very thick,  fill=blue!5, draw=black] (-.05,2.95) rectangle node(R3) {$x_4$} (1.5,4.5) ;
   \draw[very thick,  fill=blue!5, draw=black] (3.95,2.95) rectangle node(R3) {$x_3$} (5.5,4.5) ;
   \draw[very thick,  fill=blue!5, draw=black] (7.95,1.5) rectangle node(R3) {$x_1$} (9.5,3) ;
       \draw[very thick,-stealth,draw=black]  (5.6,3.7) -- (6.8,3.7) ;   
       \draw[very thick,stealth-,draw=black]  (1.6,3.7) -- (3.8,3.7) ;   
       \draw[very thick,stealth-,draw=black]  (1.7,3) -- (3.8,1.7) ;   
       \draw[very thick,-stealth,draw=black]  (1.5,2.7) -- (3.7,1.3) ;   
       \draw[very thick,stealth-,draw=black]  (5.6,.5) -- (6.6,.5) ;       
       \draw[very thick,-stealth,draw=black]  (5.6,1) -- (6.6,1) ;          
       \draw[very thick,stealth-,draw=black]  (-1.2,.75) -- (-0.15,.75) ;      
       \draw[very thick,-stealth,draw=black]  (4.5,1.6) -- (4.5,2.8) ;  
       \draw[very thick,stealth-,draw=black]  (5,1.6) -- (5,2.8) ;  
       \draw[very thick,stealth-,draw=black]  (.8,1.6) -- (.8,2.8) ;       
       \draw[very thick,-stealth,draw=black]  (8.75,0.4) -- (8.75,1.4) ;       
       \draw[very thick,-stealth,draw=black]  (8.75,3.1) -- (8.75,4.1) ;       
       \draw[very thick,stealth-,draw=black]  (5.7,1.5) -- (7.75,2.25) ; 
       \draw[very thick,stealth-,draw=black]  (3.8,.75) -- (1.6,0.75) ;   
       \draw[very thick,stealth-,draw=black]  (-1.2,3.8) -- (-.15,3.8) ;   
\end{tikzpicture}
\end{center}
\caption{Schematic representation of the model network. (subsystems are not shown) (Case study~\ref{apxsec:tilly}).}
\label{fig:cone_sp}
\end{figure}

The composite transfer flow matrix as formulated in Table~\ref{tab:flow_stor} becomes
\begin{equation}
\label{appex:tr_fl}
T^{\texttt{t}} = \left[
\begin{array}{ccccc}
         0      &   0     &    0    &     0  &   0 \\
    10716.74    &  1967.00    &   1663.02    &   426.42    &  198.04 \\
    4857.67    &   5205.00   &    753.81    &   193.28   &   89.76 \\
    2224.91    &   2384.00    &   409.40  &   88.53   &   41.12 \\
    345.31   &    370.00    &   63.54   &   370.00    &   6.38
    \end{array}
\right] .
\end{equation}
The simple cycling and acyclic subflow matrices can also be expressed as follows:
{\small\begin{equation}
\label{appex:tr_fl2}
\tilde{T}^{\texttt{c}} = \left[
\begin{array}{ccccc}
         0      &   0     &    0    &     0          & 0 \\
    1835.74  &  131.26     &    0    &     0  & 0 \\
    703.51  &  50.30     &    0     &    0  & 0 \\
    82.62  &  5.91    &     0     &    0  & 0 \\
    5.96  &  0.43   &      0     &    0  & 0
    \end{array}
\right]
\quad \mbox{and} \quad
\tilde{T}^{\texttt{a}} = \left[
\begin{array}{ccccc}
    11184  &       0     &    0     &    0       & 0 \\
    8881  &  635  &       0   &      0   & 0 \\
    4154.16  &  297.03   &      0   &      0   & 0 \\
    2142.30  &  153.18    &     0    &     0   & 0 \\
    339.35  &  24.26     &    0    &     0   & 0
    \end{array}
\right] .
\nonumber
\end{equation}}

\section*{Acknowledgments}
The author would like to thank Hasan Coskun for useful discussions and his helpful comments that improved the manuscript.

\bibliographystyle{siamplain}
\bibliography{references}

\end{document}